\documentclass[aps,a4paper,twocolumn,superscriptaddress,floatfix,nofootinbib]{revtex4-2}

\usepackage[T1]{fontenc}
\usepackage[utf8]{inputenc}
\usepackage{graphicx}
\usepackage[plainpages=false,pdfpagelabels,hypertexnames=false]{hyperref}
\usepackage[usenames,dvipsnames,table]{xcolor}
\usepackage{braket}
\usepackage{amsmath,amssymb}
\usepackage{amsthm}
\usepackage{bm}
\usepackage{mathtools}
\usepackage{bbm}
\usepackage{tikz}
\usetikzlibrary{arrows,calc,shapes.geometric,decorations.pathmorphing,patterns}
\usepackage{diagbox}
\usepackage{multirow}
\usepackage{booktabs}

\newcommand{\id}{\mathbbm{1}}
\newcommand{\Tr}{\mathrm{Tr}}
\newcommand{\Pg}{P_{\mathrm{g}}}
\newcommand{\Pgn}{P_{\mathrm{g},n}}
\newcommand{\Imin}{I_{\mathrm{min}}}
\newcommand{\cset}{\mathcal{C}}
\newcommand{\qset}{\mathcal{Q}}

\newcommand{\pkey}{p_{\mathrm{k}}}
\newcommand{\Gmin}{G_{\mathrm{cl}}}
\newcommand{\vcl}{v_{\mathrm{cl}}}
\newcommand{\prot}[2]{(#1,#2)}

\newtheorem{proposition}{Proposition}
\newtheorem{lemma}[proposition]{Lemma}

\theoremstyle{definition}

\newtheorem{remark}[proposition]{Remark}

\begin{document}

\title{Prepare-and-measure quantum key distribution from a bound on information leakage}

\author{Abhishek Mishra}
\author{Michele Masini}
\author{Maria Balanz\'o-Juand\'o}
\affiliation{Laboratoire d'Information Quantique, CP 225, Universit\'e libre de Bruxelles (ULB), Belgium}

\author{Armin Tavakoli}
\affiliation{Physics Department and NanoLund, Lund University, Box 118, 22100 Lund, Sweden}

\author{Stefano Pironio}
\affiliation{Laboratoire d'Information Quantique, CP 225, Universit\'e libre de Bruxelles (ULB), Belgium}

\date{October 1, 2026}

\begin{abstract}
Every quantum key distribution (QKD) protocol must limit the information available outside the users' laboratories. In device-independent QKD this is essentially the only assumption made on the devices. In a prepare-and-measure setting it cannot be imposed strictly, since the sender's device must emit a signal that depends on her private input. The natural replacement, and a minimal one, is a leakage bound: a bound on how much information about the input leaves the sender's laboratory, through the intended signal or through any side channel. We ask how far this unavoidable requirement can replace a detailed model of the source. We show that this single bound, quantified by the probability of guessing the input from everything the source emits, suffices to prove security against collective attacks. In our analysis, nothing is assumed about the precise form of the emitted states, their dimension, energy or purity, nor about the receiver's measurement device, beyond the standard requirement that it does not leak Bob's data. The standard B92 and BB84 protocols do not certify any key under this minimal leakage bound assumption. But closely related protocols can and, in the ideal limit, achieve one secret bit per key round. For noisy implementations, we compute key rates using semidefinite programming relaxations. Our best protocols tolerate visibilities down to $0.84$, corresponding to an $8\%$ quantum bit error rate, compared with the $11\%$ threshold of standard device-dependent BB84.
\end{abstract}

\maketitle

\section{Introduction}

Quantum key distribution (QKD)~\cite{BB84,Ekert91} aims to establish secrecy from physical principles, without assumptions about the computational power of an eavesdropper. Its security proofs, however, also rely on explicit or implicit assumptions about the devices, and these may fail in practice. Imperfect sources and detectors have opened loopholes that an adversary can exploit~\cite{PhysRevLett.85.1330,Lydersen2010,Xu2020,Zapatero2025}. This is known as the \emph{implementation problem} of QKD.

Two complementary strategies address this problem. The first characterises the imperfections of the devices and incorporates the resulting models or quantitative bounds into the security proof~\cite{gottesman2004securityquantumkeydistribution,KoashiPreskill2003,LoPreskill2007,Maroy2010,Tamaki2014,Lucamarini2015,Pereira2019,Pereira2020,Zhang2021,Nahar2023,CurrasLorenzo2025}. The second, \emph{device independence} (DI), dispenses with a detailed model of the quantum devices and certifies security from the violation of a Bell inequality~\cite{PhysRevLett.98.230501,Pironio2009,PhysRevLett.113.140501,ArnonFriedman2018,Zapatero2023}. In between, semi-device-independent (SDI) approaches retain limited assumptions on the devices, such as bounds on dimension, energy or state overlaps, while leaving other device properties uncharacterised~\cite{onewaysdiqkd,Vanhimbeeck,Ioannou_2022}.

In any approach, some restriction on the information that escapes the users' laboratories is nevertheless unavoidable, since no protocol can protect a key that a device simply copies to an adversary. This remains true of DI QKD, where Bell correlations replace a detailed description of the devices, but do not remove the need to control unwanted communication and to protect the classical data (inputs, outputs, raw key) of Alice and Bob~\cite{PhysRevLett.98.230501,Pironio2009,BarrettColbeckKent2013,CurtyLo2019}. In the entanglement-based DI setting, this requirement can in principle be imposed in its strict form, because once the entangled systems have been distributed nothing needs to leave the laboratories\footnote{Relaxed versions, in which the devices may leak a bounded amount of information, have also been analysed~\cite{Tan2023,ZapateroCurty2026}.}.

In a prepare-and-measure (P\&M) protocol, the strict form is untenable. Alice's device must send out a physical system that depends on her private input $x$, since this is precisely how the key is communicated. A P\&M source is thus leaky by design, and the relevant requirement is not the absence of leakage but a \emph{bound} on it: the information about $x$ carried by everything that leaves Alice's laboratory, the intended signal and any unintended side channel alike, must be limited. Some bound of this kind is unavoidable. If the emitted systems could reveal $x$ completely, for instance because a malicious source attaches a classical copy of $x$ to the signal, Eve could learn the key without altering anything that Bob observes.

This suggests a direct question for P\&M QKD. Can a quantitative bound on the information leaving the source be a sufficient assumption for security, without a detailed model of the preparations and of the measurements? In this work we show that it can. Following Refs.~\cite{Tavakoli2020,Tavakoli2022}, we quantify the leakage operationally by the optimal probability $\Pg$ of guessing Alice's input $x$, drawn uniformly among $n$ possible values, from everything the source emits. We then assume that $\Pg\leq G/n$ for a constant $G$ fixed in advance. A blind guess succeeds with probability $1/n$, so $G$ is the factor by which Eve's best guess beats it, and $G=2$ corresponds to one bit of information about $x$. We refer to this assumption as the \emph{leakage bound} (LB). Constraints of this form were introduced in Refs.~\cite{Tavakoli2020,Tavakoli2022} under the name of information restriction, and they are equivalent to a bound on the one-shot accessible information of the emitted ensemble.

The LB characterises the source by a single number. It is a restriction on what the source reveals about $x$, not on the dimension, energy, purity or precise form of the emitted states, which may be mixed and possibly even correlated with Eve. We also make no assumption about Bob's measurement device, which can be fully uncharacterised. Our assumption is thus minimal in the sense that it is the quantitative version of a requirement that every QKD protocol has to meet anyway.

As in DI QKD, security under a leakage bound requires non-classical correlations. If the observed statistics can be reproduced by classical communication carrying the same amount of information about $x$, Eve can copy the classical message and learn everything Bob learns. This is the P\&M counterpart of the fact that local correlations certify no DI key~\cite{acin2006from}. We show that this requirement rules out two families of standard protocols. The first consists of all protocols with two preparations, such as B92~\cite{B92}. The second consists of all protocols in which each preparation gives a deterministic outcome for one of Bob's measurements and uniformly random outcomes for the others, including BB84~\cite{BB84}, the six-state protocol and the three-state protocol of Ref.~\cite{Tamaki2014}.

Security can be restored by rotating Bob's measurement bases relative to Alice's preparation bases, producing non-classical statistics under the leakage bound. We study protocols with three, four or six preparations, including a P\&M analogue of the CHSH-based DI protocol~\cite{PhysRevLett.98.230501}. These protocols certify one secret bit per key round in the noiseless limit, and we analyse their tolerance to noise and loss.

Computing the key rate requires a lower bound on the conditional von Neumann entropy $H(A|E)$, minimised over all states and measurements of unbounded dimension compatible with the observed data and with the leakage bound. We obtain it by combining the variational bound of Koßmann and Schwonnek on the relative entropy~\cite{koßmann2025reliableentropyestimationobserved,kossmann2024_optimization,Frenkel2023}, adapted to the P\&M setting, with the semidefinite relaxation of correlations under a leakage bound introduced in Ref.~\cite{Tavakoli2022}. We point out that the latter can be simplified without weakening it. Only the moments that are linear in the emitted states are needed, which reduces the size of the programs considerably and allows us to reach the relaxation levels used here.

The paper is organised as follows. Section~\ref{sec:setting} defines the setting, states the leakage bound, describes the structure of the associated QKD protocols and the key-rate expression. Section~\ref{sec:protocols} introduces the qubit protocols we study. It first shows that non-classical correlations are necessary for key, so that B92, BB84 and the six-state protocols do not certify any. It then presents the protocols with tilted measurements and the inequalities that witness their non-classicality, and proves that their ideal implementations certify one bit of key per key round. Section~\ref{sec:imperfect} turns to imperfect implementations and discusses how the leakage bound is chosen for a real source. Section~\ref{sec:framework} develops the entropy bound and its relaxation, and Sec.~\ref{sec:results} reports the numerical key rates. We conclude in Sec.~\ref{sec:conclusion}.

\section{Setting and assumptions}
\label{sec:setting}

\subsection{Prepare-and-measure QKD}
\label{sec:pm}

We consider a prepare-and-measure setting in which Alice holds a source, Bob holds a measurement device, and Eve stands in the quantum channel between them (Fig.~\ref{fig:pnmqkd}). Alice draws an input $x\in\mathcal X=\{1,\dots,n\}$ and feeds it to her source, which emits a quantum system $Q$. In addition, the source may initially be correlated with a system $E_0$ held by Eve outside Alice's laboratory: for instance a classical variable describing the internal state of the device, or a quantum system entangled with it. Conditioned on $x$, the state of $QE_0$ is $\rho_x^{QE_0}$.

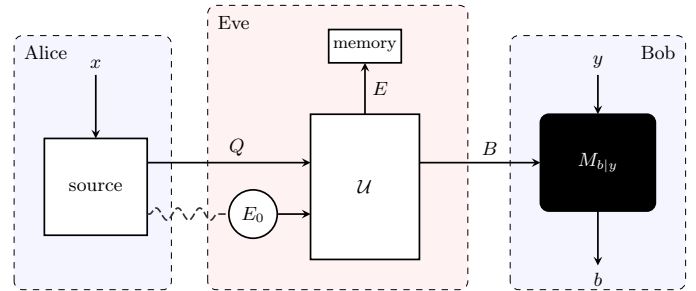
\begin{figure}[tb]
    \centering
    \scalebox{0.8}{\begin{tikzpicture}[>=stealth, font=\small]
    \draw[dashed, rounded corners, fill=blue!4] (-0.9,-1.0) rectangle (1.7,3.2);
    \node[anchor=north west] at (-0.85,3.15) {Alice};
    \draw[dashed, rounded corners, fill=red!5] (2.3,-1.0) rectangle (6.6,3.7);
    \node[anchor=north west] at (2.35,3.65) {Eve};
    \draw[dashed, rounded corners, fill=blue!4] (7.3,-1.0) rectangle (10.2,3.2);
    \node[anchor=north east] at (10.15,3.15) {Bob};
    \draw[thick, fill=white] (-0.4,-0.1) rectangle (1.3,1.5);
    \node at (0.45,0.7) {source};
    \draw[->, thick] (0.45,2.55) node[above] {$x$} -- (0.45,1.5);
    \draw[thick, fill=white] (4.0,-0.5) rectangle (5.8,1.9);
    \node at (4.9,0.7) {$\mathcal U$};
    \draw[->, thick] (1.3,1.1) -- (4.0,1.1) node[pos=0.55, above] {$Q$};
    \node[draw, thick, circle, fill=white, minimum size=0.8cm, inner sep=0pt] (E0) at (3.05,0.25) {$E_0$};
    \draw[->, thick] (E0.east) -- (4.0,0.25);
    \draw[thick, black!80, dash pattern=on 3.5pt off 1.5pt, decorate, decoration={snake, amplitude=1mm, segment length=3.6mm, pre length=1pt, post length=1pt}] (1.3,0.25) -- (E0.west);
    \draw[->, thick] (4.9,1.9) -- (4.9,2.75) node[midway, right] {$E$};
    \draw[thick, fill=white] (4.3,2.75) rectangle (5.5,3.3);
    \node at (4.9,3.025) {\footnotesize memory};
    \draw[->, thick] (5.8,1.1) -- (7.8,1.1) node[pos=0.58, above] {$B$};
    \draw[thick, fill=black, rounded corners] (7.8,0.3) rectangle (9.7,1.9);
    \node at (8.75,1.1) {\color{white}$M_{b|y}$};
    \draw[->, thick] (8.75,2.55) node[above] {$y$} -- (8.75,1.9);
    \draw[->, thick] (8.75,0.3) -- (8.75,-0.6) node[below] {$b$};
\end{tikzpicture}}
    \caption{Prepare-and-measure QKD. Alice's source receives her input $x$ and emits the system $Q$. Eve applies an isometry $\mathcal U$ to $Q$ and a system $E_0$, which may be initially classically correlated or entangled with Alice's source (wavy dashed line). She forwards $B$ to Bob and keeps $E$ in a quantum memory. Bob's device measures $B$ according to his input $y$ and returns the outcome $b$. }
    \label{fig:pnmqkd}
\end{figure}

Eve controls the channel. In every round she applies an isometry $\mathcal U:QE_0\to BE$, forwards the system $B$ to Bob and stores $E$ in a quantum memory. The states after her intervention are
\begin{equation}
    \label{eq:postattack}
    \rho_x^{BE} = \mathcal U\,\rho_x^{QE_0}\,\mathcal U^\dagger .
\end{equation}
Bob draws an input $y\in\mathcal Y=\{1,\dots,m\}$, measures $B$ and records an outcome $b$. His measurement for input $y$ is described by a POVM $\{M_{b|y}\}_b$. The observed correlations are
\begin{equation}
    \label{eq:behaviour}
    p(b|x,y) = \Tr\bigl(\rho_x^{B} M_{b|y}\bigr), \qquad \rho_x^B=\Tr_E\,\rho_x^{BE} .
\end{equation}
We refer to the table $p(b|x,y)$ as the \emph{behaviour}. 

\subsection{The leakage bound}
\label{sec:ri}

Following Refs.~\cite{PhysRevLett.115.110501,Tavakoli2020,Tavakoli2022}, we quantify how much information a set of states $\{\rho_x\}$ carries about the label $x$ by the optimal probability of guessing $x$ from the state when $x$ is drawn with probabilities $q_x$,
\begin{equation}
    \label{eq:pg}
    \Pg\bigl(\{q_x,\rho_x\}\bigr) =
    \max_{\{N_x\}} \sum_x q_x \Tr(\rho_x N_x),
\end{equation}
the maximum being over all POVMs, $N_x\geq 0$ and $\sum_x N_x=\id$. The guessing probability ranges from $\max_x q_x$, attained when all states are identical and carry no information about $x$, to $1$, attained when they are perfectly distinguishable. We take the reference prior uniform, $q_x=1/n$, which is a natural choice for the symmetric ensembles used below and the most constraining one~\cite{tavakoli2021correlations}, and we rescale the guessing probability to
\begin{equation}
    \label{eq:pgn}
    \Pgn\bigl(\{\rho_x\}\bigr)=n\,\Pg\bigl(\{\tfrac1n,\rho_x\}\bigr)=\max_{\{N_x\}}\sum_x\Tr(\rho_xN_x),
\end{equation}
the expected number of inputs correctly guessed when each input is sent once. It ranges from $1$, when the states are identical, to $n$, when they are perfectly distinguishable. The \emph{leakage bound} (LB) is the assumption that the source states $\rho_x^{QE_0}$ satisfy
\begin{equation}
    \label{eq:ri}
    \Pgn\bigl(\{\rho_x^{QE_0}\}\bigr) \leq G ,
\end{equation}
where the value $G$ is fixed in advance and is part of the assumption. Equivalently, Eq.~\eqref{eq:ri} bounds the information that $QE_0$ carries about Alice's input, through the following one-shot version of the accessible information:
\begin{align}
    \Imin(X{:}QE_0) &= H_{\min}(X)-H_{\min}(X|QE_0) \nonumber\\
    &= \log_2\Pgn \leq \log_2 G .
    \label{eq:imin}
\end{align}
Thus $n$ states carry $\log_2 n$ bits of information about $x$ when they are perfectly distinguishable and $0$ bits when they are identical, and the bound $G=2$, which is the honest value for all the protocols of this work, corresponds to one bit of information about $x$ whatever $n$. Constraints of this form were introduced in Refs.~\cite{Tavakoli2020,Tavakoli2022}, where they are called information restrictions. 

 In the security analysis below, except for the leakage bound~\eqref{eq:ri}, we make no assumptions about the states $\rho_x^{QE_0}$ characterising Alice's source and the measurements $\{M_{b|y}\}$ characterising Bob's device, which may differ from the honest implementation. 

We make several remarks about this assumption.

\emph{The uniform prior is part of the assumption.} The bound~\eqref{eq:ri} is a characterisation of the source as a physical device. It says how much information the device would leak about its input if the inputs were drawn uniformly and the emitted systems were measured with an optimal POVM $\{N_x\}$, in the same way as a capacity characterises a channel independently of the messages actually sent through it. It is not a statement about the information leaked in an actual run of the protocol, where the inputs may be drawn with probabilities $p_x$ that differ from the uniform prior used to formulate the LB assumption. A non-uniform reference prior $q_x$ could be used instead, in which case $\Pg/\max_xq_x$ plays the role of $\Pgn$ in Eq.~\eqref{eq:imin}.

\emph{The states are not assumed pure.} Since no restriction is placed on Bob's POVMs, they may be taken projective without loss of generality by
Naimark dilation~\cite{Holevo1973,Nielsen_Chuang_2010}, since the dilation appends an
ancilla in a fixed, $x$-independent state, which changes neither
Eq.~\eqref{eq:behaviour} nor Eq.~\eqref{eq:pg}. 

It is not, however, generically possible to purify the source states $\rho_x^{QE_0}$ without violating the leakage bound, as the purification may depend on $x$ and increase the guessing probability~\cite{Tavakoli2022}. For instance, take two states $\rho_0=\ket0\bra0$ and $\rho_1=\id/2$. Then $P_g=3/4$. But any purifications $\ket{\psi_0}$ and $\ket{\psi_1}$, such as $\ket{\psi_0}=\ket{00}$ and $\ket{\psi_1}=(\ket{01}+\ket{10})/\sqrt2$, have $P_g\geq\tfrac12(1+1/\sqrt2)\approx0.85$.  Mixed preparations must thus be retained as genuine variables of the problem, and the techniques of Refs.~\cite{Wang2019,Ioannou_2022}, which rely on purity, do not apply here.

\emph{From the emitted states to the post-attack states.} The leakage bound is a property of the source, and Eq.~\eqref{eq:ri} is accordingly imposed on the emitted states $\rho_x^{QE_0}$. The behaviour $p(b|x,y)$, however, depends on the states $\rho_x^{BE}$ after Eve's intervention. The whole problem can nevertheless be formulated in terms of the post-attack states $\rho_x^{BE}$ only, since an LB on the emitted states implies the same LB on the post-attack states.
Indeed, Eve's isometry $\mathcal U$ is independent of $x$, and minimum-error discrimination cannot be improved by a common isometry. Hence
\begin{equation}
    \label{eq:ri_be}
    \Pgn\bigl(\{\rho_x^{BE}\}\bigr) \leq \Pgn\bigl(\{\rho_x^{QE_0}\}\bigr) \leq G ,
\end{equation}
and the same holds for any $x$-independent channel, by data processing. From now on, we thus drop the superscript and write $\rho_x$ for the joint states $\rho^{BE}_x$ of $B$ and $E$; it is this object that appears in the optimisation problems of Sec.~\ref{sec:framework}, and Bob's marginal is $\Tr_E(\rho_x)$. We reintroduce the superscripts only when we need to distinguish the states before or after Eve's intervention. 

\emph{Semidefinite form.} The guessing probability~\eqref{eq:pgn} is the optimal value of a semidefinite
program, whose dual is the minimum-error discrimination problem in its
Holevo--Yuen--Kennedy--Lax form~\cite{Holevo1973,YuenKennedyLax1975,Watrous2018},
\begin{equation}
    \label{eq:duality}
    \Pgn = \min_\sigma \bigl\{ \Tr\sigma \;:\;
          \sigma \geq \rho_x \ \ \forall x \bigr\} .
\end{equation}
Weak duality is immediate: for any feasible
$\sigma$ and any POVM $\{N_x\}$,
$\sum_x \Tr(\rho_x N_x) \leq \sum_x \Tr(\sigma N_x) = \Tr\sigma$.
Strong duality holds by Slater's condition, since the primal is strictly
feasible: the POVM $N_x=\id/n$ lies in the interior of the feasible set~\cite{Watrous2018}. Consequently
the assumption $\Pgn\leq G$ is equivalent to the existence of an auxiliary
operator $\sigma$ with
\begin{equation}
    \label{eq:sigma}
    \sigma \geq \rho_x \quad \forall x, \qquad \Tr\sigma \leq G .
\end{equation}
This reformulation replaces the maximisation in Eq.~\eqref{eq:pgn} by operator inequalities and a linear
trace condition, which can be conveniently handled by the moment relaxation used in Sec.~\ref{sec:framework}.

\emph{Qudit sources.}
Suppose that the $n$ states $\rho_x$ are supported on a common $d$-dimensional subspace with projector $\Pi$ (and $d\leq n$). Then $\sigma=\Pi$ satisfies Eq.~\eqref{eq:sigma}, since $\rho_x\leq\Pi$, and $\Tr\sigma=d$. Hence
\begin{equation}
    \label{eq:pg_dim}
    \Pgn\leq d, \qquad \Imin(X{:}QE_0)\leq \log_2 d\text{ bits}.
\end{equation}
This is the situation, e.g., of a qudit source that is not entangled with Eve: $\rho_{x}^{QE_0}=\rho_x^Q\otimes\rho^{E_0}$, with $Q$ a $d$-dimensional system. Thus a dimension bound is a special case of a leakage bound~\cite{Tavakoli2022,Pauwels2025}. The bound holds whatever the states within the qudit space are, pure or mixed, and whether or not they are the intended ones. As we discuss in Sec.~\ref{sec:imperfect_corr}, it also holds if the source is classically correlated with Eve, even if the correlations change the qudit subspace itself.

It does not hold, however, if the source is entangled with Eve. For instance, if the emitted qudit $Q$ is maximally entangled with a qudit $E_0$ held by Eve, the total space  $QE_0$ is of dimension $d^2$ and $I_{\min}(X{:}QE_0)$ can reach $2\log_2 d$ bits (this is just dense coding). This emphasises that the LB is a property of the source, not of the intended signal alone. The bound should include any external system $E_0$ that is correlated with the source. Similarly, $Q$ should include everything that leaves Alice's laboratory: the intended signal, which ideally reaches Bob, but also any unintended side channel, such as light in other optical modes, the light reflected back into the channel in a Trojan-horse attack, or electromagnetic radiation. The LB captures the total information about $x$ that leaves Alice's laboratory. Determining the value of $G$ thus relies on some informed physical model of the source and represents an assumption that is as good as the model. But our security analysis does not rely on the model itself, only on the resulting bound. For the moment, we suppose some value of $G$ is fixed;  we discuss how it can be determined in practice in Sec.~\ref{sec:imperfect}.

\subsection{QKD Protocols}
\label{sec:generic}

The protocols we study share the following structure, which generalises that of BB84. Alice's inputs are of two kinds. The key-generating inputs $x\in\{1,\ldots,2\nu\}$ form a subset of all $n$ inputs. We reinterpret them as pairs $x=(\theta,a)$ where $\theta=1,\dots,\nu$ are \emph{key bases} and $a=0,1$ are the \emph{raw key symbols}. In the honest implementation, the source emits two orthogonal qubit states $\ket{\psi(\theta,0)}$ and $\ket{\psi(\theta,1)}$ for each key basis $\theta$. The remaining inputs, if any, serve only to test the source. Bob's inputs likewise consist of key-generating measurements $y\in\{1,\dots,\nu\}$ and, possibly, additional test measurements. In the honest implementation, Bob's measurement $y=\theta$ is the projective measurement onto the basis $\{\ket{\psi(\theta,0)},\ket{\psi(\theta,1)}\}$, so that his outcome reproduces $a$. In BB84, for instance, $\nu=2$, the key basis for $\theta=1$ is $\{\ket 0,\ket 1\}$, the key basis for $\theta=2$ is $\{\ket +,\ket -\}$, measured with $Z$ for $y=1$ and $X$ for $y=2$, and there are no additional test inputs or measurements.

 The protocol runs as follows.
\begin{enumerate}
    \item \emph{Quantum phase.} In each of $N$ rounds, Alice draws $x$ with probability $p_x$ and Bob draws $y$ with probability $p_y$, independently. Alice's source emits $Q$; Bob measures the system he receives and records $b$.
    \item \emph{Sifting.} Over an authenticated public channel, Bob announces his inputs, and Alice announces for every round whether her input was a key or a test input and, for a key input, the basis $\theta$ but not the symbol $a$. The rounds in which Alice used the key basis $\theta$ and Bob the matching measurement $y=\theta$ are the \emph{key rounds}; Alice's raw key symbols in these rounds form her raw key, and Bob's outcomes form his. All other rounds, together with a random subset of the key rounds, are \emph{test rounds}.
    \item \emph{Parameter estimation.} In the test rounds, Alice and Bob announce their inputs and outcomes and perform a statistical test. They abort if the test fails and continue otherwise. Intuitively, the test is designed to verify that their device behaviour does not deviate too much from the expected one $p(b|x,y)$.
    \item \emph{Post-processing.} Alice and Bob perform information reconciliation, so that they end up with identical keys,     
    followed by privacy amplification, which yields the final key~\cite{wolf2021quantum}.
\end{enumerate}

\emph{Assumptions.} Beyond the validity of quantum theory and the leakage bound, we make the following assumptions common to all QKD protocols, including DI ones~\cite{PhysRevLett.98.230501,Pironio2009,BarrettColbeckKent2013,CurtyLo2019}: (i) Alice's and Bob's inputs are drawn from trusted randomness and are independent of the devices and of Eve; (ii) the classical post-processing is trusted, and all public messages are authenticated; (iii) no private data -- the inputs $x$ and $y$, the outcomes $b$, the data of the classical post-processing -- leave the laboratories other than through the public channel, and no signal travels from Bob's device back to Alice's source. Assumption (iii) may seem to sit uneasily with a bound on leakage. The two are compatible because they concern different parts of Alice's laboratory. The leakage bound concerns the quantum source, a partly untrusted device whose emissions $Q$ are precisely what Eve has access to, and it covers every physical system that this device sends out. The rest of Alice's laboratory, where $x$ is drawn, stored and processed, must be secure, exactly as in DI QKD (a source that leaks a bounded amount of information about $x$ is acceptable, but a random-number generator or a computer that leaks $x$ is not).

\subsection{Key rate}
\label{sec:keyrate}

We take the key-generating inputs to be equiprobable, so that in the key rounds each basis is used with probability $1/\nu$ and each of its two inputs with probability $1/2$. After sifting, Eve knows which basis was used in each key round, so the
classical--quantum state of Alice's raw key and Eve's information in a key round is
\begin{equation}
    \label{eq:rho_AE}
    \rho_{AE\Theta} = \frac{1}{2\nu}\sum_{\theta=1}^{\nu}\sum_{a=0}^{1}
    \ket{a}\bra{a}_A \otimes \Tr_B\bigl(\rho^{BE}_{(\theta,a)}\bigr)\otimes\ket{\theta}\bra{\theta}_{\Theta} ,
\end{equation}
where the classical register $\Theta$ holding the basis label belongs to Eve.

For one-way reconciliation from Alice to Bob, the
asymptotic key rate per key round against collective attacks is
bounded from below by the Devetak--Winter formula~\cite{DevetakWinter2005},
\begin{equation}
    \label{eq:dw}
    K_\infty \;\geq\; H(A|E\Theta)_{\rho_{AE\Theta}} - H(A|B) ,
\end{equation}
where $H(A|E\Theta)$ is the conditional von Neumann entropy of Alice's raw key given
Eve's quantum memory $E$ and the basis register $\Theta$, and $H(A|B)$ is the conditional Shannon entropy of Alice's raw
key given Bob's outcome in the key rounds. The rate per emitted signal is
\begin{equation}
    \label{eq:rate_per_signal}
    R = \pkey \, K_\infty ,
\end{equation}
where $\pkey$ is the probability that a round is a key round. If Alice and Bob bias their inputs towards the subset of key-generating ones, $\pkey$ can be made arbitrarily close to $1/\nu$, since given a key basis $\theta$, the probability that Bob chooses a matching measurement $y=\theta$ is $1/\nu$. 

The reconciliation term $H(A|B)$ is a function of the observed behaviour
alone. The security content of Eq.~\eqref{eq:dw} is therefore concentrated in
$H(A|E\Theta)$, which must be minimised over every state and measurement compatible
with the observed behaviour and with the LB assumption~\eqref{eq:ri_be}. That
minimisation is the subject of Sec.~\ref{sec:framework}. 

\section{Security and non-classicality}
\label{sec:protocols}

Many QKD security proofs assume from the outset that the devices behave quantumly. Device-dependent analyses of BB84, for instance, assume that the source prepares the non-orthogonal states $\ket0$, $\ket1$, $\ket+$ and $\ket-$ and that Bob performs incompatible measurements. Fidelity-based analyses such as those of Refs.~\cite{Pereira2019,Pereira2020,CurrasLorenzo2025,CurrasLorenzo2025num,Pereira2025} assume that the actual states are close to such quantum reference states, and thus also assume quantum behaviour. The leakage bound is not an assumption of this kind. It restricts how much the source reveals about the input $x$, not the nature of the systems that carry it, and purely classical devices can satisfy it -- just as they can satisfy the requirement, in DI QKD, that the devices do not communicate. In DI QKD, genuinely quantum properties are therefore not assumed from the start but inferred from sufficiently strong correlations in the observed statistics. The same holds in our analysis.

A behaviour $p(b|x,y)$ is \emph{classical}, given the bound $G$, if it can be produced by an encoding $p(z|x)$ of the input into a classical message $z$ and responses $p(b|z,y)$ of Bob's device to that message,
\begin{equation}
    \label{eq:classical_model}
    p(b|x,y) = \sum_z p(z|x)\,p(b|z,y),
\end{equation}
with an information content respecting the bound,
\begin{equation}
    \label{eq:classical_pg}
    \sum_z \max_x\, p(z|x) \;\leq\; G .
\end{equation}
We denote by $\cset_G$ the set of these classical behaviours. The left-hand side of Eq.~\eqref{eq:classical_pg} is exactly
Eq.~\eqref{eq:pg} evaluated on the classical ensemble
$\rho_x=\sum_z p(z|x)\ket z\bra z$, whose optimal discrimination strategy is to
read $z$ and output the most likely input. The set $\cset_G$ is a convex
polytope~\cite{Tavakoli2022}.

Classical behaviours certify no key. This is the P\&M counterpart, under a leakage bound, of the statement that
correlations admitting a local model certify no key~\cite{acin2006from} in a DI setting.
To see this, take a classical model~\eqref{eq:classical_model}--\eqref{eq:classical_pg}, and consider the attack where Eve intercepts the message $z$, keeps a copy and forwards the original to Bob. The behaviour is unchanged and the rescaled guessing probability $\Pgn$ is still bounded by $G$, since duplicating a classical register does not make it easier to read. Conditioned on Eve's copy of $z$, Bob's data $(y,b)$ are independent of Alice's, since $y$ is chosen independently of everything else and $b$ is drawn from $p(b|z,y)$. Hence $I(A{:}B|E)=0$, where $A$ and $B$ denote all the classical data of Alice and Bob in a round. The intrinsic information $I(A{:}B\!\downarrow\!E)\leq I(A{:}B|E)$ therefore vanishes. Since it upper-bounds the rate at which Alice and Bob can distil secret key from i.i.d.\ copies of the triple $(A,B,E)$ by any protocol with public discussion, including two-way protocols~\cite{MaurerWolf1999}, no secret key can be certified from $\bigl(p(b|x,y),G\bigr)$ when $p\in\cset_G$, for any choice of key rounds, sifting or post-selection, and any one-way or two-way post-processing.

\subsection{Protocols without certifiable key}
\label{sec:nokey}
Being outside $\cset_G$ at the honest $G$ is thus a necessary condition for any candidate protocol to generate a secret key. This condition rules out two broad classes of protocols. The first is any protocol that uses only $n=2$ input states, such as B92~\cite{B92}. Indeed, it was shown in Ref.~\cite{Tavakoli2022} that for two inputs, the set of classical behaviours $\cset_G$ coincides with the set of quantum behaviours $\qset_G $. That is, any behaviour obtained from two quantum states with $\Pgn\leq G$, for any number of measurements $y$ and outcomes $b$ on Bob's side, admits a classical model with the same leakage bound. 

The second class consists of matched-basis protocols. In these protocols, each input $x$ leads to a deterministic outcome $a_x$ for one of Bob's measurements, $y(x)$, and to uniformly random outcomes for all other measurements $y\neq y(x)$. Moreover, at least two of the preparations correspond to orthogonal states. BB84, the six-state protocol and the three-state protocol of Ref.~\cite{Tamaki2014} are of this kind, since each prepared state is an eigenstate of one of Bob's $m$ measurements and is unbiased with respect to the others.

Matched-basis protocols certify no key, whatever the value of $G$. The reason is that a classical message carrying a single bit of information reproduces their statistics. Given her input $x$, let Alice's device send a string $z=(z_1,\ldots,z_m)$ of $m$ bits, one for each of Bob's measurements. Each bit $z_i$ is uniformly random, except the bit $z_{y(x)}$ corresponding to the measurement $y(x)$, which is set to the deterministic outcome $a_x$. Upon receiving $z$ and his input $y$, Bob's device outputs $b=z_y$. This reproduces the behaviour. For BB84, for instance, the states $\ket0$, $\ket1$, $\ket+$ and $\ket-$ are encoded as $(0,r)$, $(1,r)$, $(r,0)$ and $(r,1)$, with $r$ a uniformly random bit, and Bob's device outputs the first bit for $Z$ and the second for $X$. Each input sends each of the $2^{m-1}$ strings compatible with it with probability $2^{-(m-1)}$, so the left-hand side of Eq.~\eqref{eq:classical_pg} is $2^m\cdot2^{-(m-1)}=2$. The behaviour is thus classical at $G=2$, with one bit of information revealed about $x$. The honest source cannot satisfy a smaller bound, because it prepares at least two orthogonal states, which can be distinguished perfectly, so that $\Pgn\geq2$. Since $\cset_G$ grows with $G$, the behaviour is classical at every possible bound $G\geq2$, and no key can be certified. In BB84, the first bit of the message that Eve copies reveals the key in the $Z$ basis and the second the key in the $X$ basis, which she can read once the basis is announced.

\subsection{Protocols considered here}
\label{sec:protchsh}

The entanglement-based version of BB84 does not provide a positive key rate in the DI setting, because its correlations are local; this is the DI counterpart of the observation above. Rotating Bob's measurements by $45^\circ$ from the $Z$ and $X$ axes, however, produces correlations that violate the CHSH inequality~\cite{Clauser1969}, and the resulting protocol certifies a positive key rate~\cite{PhysRevLett.98.230501}. The same holds in the prepare-and-measure setting with a leakage bound, where tilting Bob's measurements produces behaviours that lie outside $\cset_G$ at the honest bound, and the resulting protocols certify a positive key rate.

We study four protocols of this kind, summarised in Table~\ref{tab:protocols_def}. All of them are qubit protocols in the $XZ$ plane of the Bloch sphere. We write
\begin{equation}
    \label{eq:Ophi}
    O_\varphi=\cos\varphi\,Z+\sin\varphi\,X
\end{equation}
for the observable whose Bloch vector makes the angle $\varphi$ with the $Z$ axis, and denote by $\prot{n}{m}$ a protocol with $n$ preparations and $m$ measurements. In each protocol, Alice prepares eigenstates of some of these observables, and Bob measures some of them; Table~\ref{tab:protocols_def} specifies each protocol by the corresponding angles $\varphi$.

\begin{table}[tb]
    \centering
    \begin{tabular}{lllc}
        \toprule
        $(n,m)$ & Preparations \phantom{abc}& Measurements \phantom{abc} & $\nu$ \\
        \midrule
        $\prot43$ & $0$, ${\pi/2}$ & $0$, ${\pi/4}$, ${-\pi/4}$ & $1$ \\
        $\prot44$ & $0$, ${\pi/2}$ & $0$, ${\pi/2}$, ${\pi/4}$, ${-\pi/4}$ & $2$ \\
        $\prot63$ & $0$, ${\pi/3}$, ${2\pi/3}$ & $0$, ${\pi/3}$, ${2\pi/3}$ & $3$ \\
        $\prot33$ & $0$, $\pi/2^+$ & $0$, ${\pi/4}$, ${-\pi/4}$ & $1$ \\
        \bottomrule
    \end{tabular}
    \caption{The protocols studied in this work, specified by the Bloch angles $\varphi$ of the observables $O_\varphi$ of Eq.~\eqref{eq:Ophi}. For each angle in the second column, Alice prepares the two eigenstates of $O_\varphi$, except for $\pi/2^+$, which indicates that only the $+1$ eigenstate of $O_{\pi/2}=X$, the state $\ket+$, is prepared. For each angle in the third column, in the order $y=1,\dots,m$, Bob measures $O_\varphi$. The first $\nu$ of the preparation and measurement bases are the key bases $\theta=1,\dots,\nu$. With inputs biased towards the key bases, the fraction of key rounds tends to $\pkey=1/\nu$ (Sec.~\ref{sec:keyrate}). The honest bound is $G=2$ for every protocol.}
    \label{tab:protocols_def}
\end{table} 

Alice's first $\nu$ bases are the eigenbases of Bob's first $\nu$ observables. They are the key bases $\theta=1,\dots,\nu$ of Sec.~\ref{sec:generic}, measured by Bob's key measurements $y=\theta$, so that in the noiseless implementation Bob's outcome $b$ always equals Alice's raw key symbol $a$. The remaining preparations and measurements serve only as tests.

Since the states are qubit states, $\Pgn\leq2$ by Eq.~\eqref{eq:pg_dim}, and since each protocol contains an orthogonal pair of preparations, $\Pgn\geq2$ (Sec.~\ref{sec:nokey}). The honest leakage bound is therefore $G=2$. All four protocols thus assume, exactly as BB84 does, that the source reveals at most $\log_2G=1$ bit of information about Alice's input, in the sense of Eq.~\eqref{eq:imin}. They differ in the behaviour they produce.

\emph{The P\&M CHSH protocol $\prot43$.} Alice prepares one of the four BB84 states, which we label $x=(\theta,a)$ as in Sec.~\ref{sec:generic},
\begin{equation}
\label{eq:4states}
\begin{aligned}
    x=(1,0) &\to \ket 0, & x=(1,1) &\to \ket 1, \\
    x=(2,0) &\to \ket +, & x=(2,1) &\to \ket -,
\end{aligned}
\end{equation}
and Bob measures one of the three observables
\begin{equation}
\label{eq:3meas}
    y=1 \to Z, \qquad y=2 \to \frac{Z+X}{\sqrt 2}, \qquad
    y=3 \to \frac{Z-X}{\sqrt 2} .
\end{equation}
The key rounds are those with $\theta=1$ and $y=1$, so that $\nu=1$ and the fraction of key rounds $\pkey$ can be made arbitrarily close to $1$ (Sec.~\ref{sec:keyrate}). The preparations $(2,0)$ and $(2,1)$ and the two tilted measurements $y=2,3$ serve only as tests. This is the P\&M analogue of the DI protocol of Ref.~\cite{PhysRevLett.98.230501}, in which Alice measures $Z$ or $X$ on her half of a maximally entangled state, with the preparation of their eigenstates replacing her measurements. The test part of the protocol, the four preparations and the two tilted measurements, is the $2\to1$ random access code~\cite{814608,PhysRevLett.114.170502} in a rotated frame: the four states encode the two bits $(a,a\oplus(\theta-1))$, and Bob reads the first bit with $y=2$ and the second with $y=3$, each with probability $\cos^2(\pi/8)$ in the ideal implementation. The protocol thus consists of a random access code supplemented by the key measurement $Z$. The same arrangement was analysed under a qubit assumption in Refs.~\cite{erik_2013,PhysRevLett.115.150501}.

\emph{The $\prot44$ protocol} adds the measurement of $X$ to $\prot43$, so that the $X$ basis becomes a second key basis. With the convention $y=\theta$ of Sec.~\ref{sec:generic}, Bob measures $Z$ for $y=1$, $X$ for $y=2$, and the tilted observables $(Z+X)/\sqrt2$ and $(Z-X)/\sqrt2$ for $y=3,4$. Key is generated in both bases, so that $\nu=2$ and $\pkey\to1/2$. The protocol is BB84 supplemented by the two tilted measurements of $\prot43$, that is, by its random access code test.

\emph{The $\prot63$ protocol.} Alice prepares the eigenstates of three observables $O_{(\theta-1)\pi/3}$ that are $60^\circ$ apart on the Bloch sphere for $\theta=1,2,3$, and Bob measures the same three observables. Every measurement generates key in the rounds with $y=\theta$, so that $\nu=3$ and $\pkey\to1/3$, and the rounds in which Bob's basis differs from Alice's serve as tests. Unlike those of the six-state protocol, the three bases are not mutually unbiased, so that the protocol is not of the matched-basis type of Sec.~\ref{sec:nokey}.

\emph{The $\prot33$ protocol.}
Two preparations do not suffice to certify a key, but this protocol shows that three preparations do. Alice prepares $\ket0$, $\ket1$ and $\ket+$, that is, the preparations $(1,0)$, $(1,1)$ and $(2,0)$ of Eq.~\eqref{eq:4states}. These are the states of the three-state protocol of Ref.~\cite{Tamaki2014}, but Bob measures as in $\prot43$, Eq.~\eqref{eq:3meas}, instead of $Z$ and $X$, which would make it a matched-basis protocol. The key rounds are those with $\theta=1$ and $y=1$, so that $\nu=1$ and $\pkey\to1$, and the preparation $(2,0)$ serves as a test.

\subsection{Non-classicality witnesses}
\label{sec:facets}

To show that the behaviours of these protocols lie outside $\cset_G$ at the honest bound, it suffices to exhibit, for each of them, a linear inequality that every classical behaviour satisfies and that the target behaviour violates. We describe Bob's binary measurements by the biases
\begin{equation}
    \label{eq:bias}
    E_{x,y} = p(0|x,y)-p(1|x,y),
\end{equation}
and, for a basis $\theta$ whose two preparations are labelled $(\theta,0)$ and $(\theta,1)$, the correlation between Alice's bit $a$ and Bob's outcome by the correlator
\begin{equation}
    \label{eq:correlator}
    C_{\theta,y} = \tfrac12\bigl(E_{(\theta,0),y}-E_{(\theta,1),y}\bigr) .
\end{equation}
In the entanglement-based twin, in which Alice measures the basis $\theta$ on her half of a maximally entangled state, $C_{\theta,y}$ is the correlator between her measurement and Bob's measurement $y$.

\emph{The P\&M CHSH inequality.} For $\prot43$, consider the following P\&M CHSH value
\begin{equation}
    \label{eq:chsh}
    S_{\mathrm{CHSH}} = C_{1,2}+C_{1,3}+C_{2,2}-C_{2,3} .
\end{equation}
Classical behaviours satisfy
\begin{equation}
    \label{eq:chsh_bound}
    S_{\mathrm{CHSH}} \leq G .
\end{equation}
More precisely, the largest classical value of $S_{\mathrm{CHSH}}$ is $\min\{G,4G-4\}$ for every $G$, the second bound being the stronger one only for $G<4/3$ (Appendix~\ref{app:classical}). At $G=2$, Eq.~\eqref{eq:chsh_bound} gives $2$, the local bound of the CHSH inequality.

The target states and measurements give $C_{1,2}=C_{1,3}=C_{2,2}=-C_{2,3}=1/\sqrt2$, and hence
\begin{equation}
    S_{\mathrm{CHSH}} = 2\sqrt2 \;>\; G = 2 .
\end{equation}
This is the largest value compatible with a one-bit bound~\cite{Tavakoli2022}, the analogue of Tsirelson's bound.

The witness and its bound connect to earlier work. $S_{\mathrm{CHSH}}$ is the score of the random access code of Sec.~\ref{sec:protchsh}: $S_{\mathrm{CHSH}}=8P_{\mathrm{rac}}-4$, where $P_{\mathrm{rac}}$ is the probability, averaged over the four inputs and the two tilted measurements, that Bob outputs the requested bit. Equation~\eqref{eq:chsh_bound} thus reads $P_{\mathrm{rac}}\leq\tfrac12+G/8$. In this form, it is the two-bit case of a bound on random access codes under a distinguishability constraint proved in Ref.~\cite{Manna2024} by the counting argument of Appendix~\ref{app:classical}. As noted in Ref.~\cite{Manna2024}, this bound is not tight for all values of the success probability; the refinement $4G-4$ for $G<4/3$ gives the exact classical maximum. At $G=2$, the bound reproduces the familiar value $3/4$ for one bit of classical communication, although $\cset_{2}$ is defined by the information carried by the message rather than by its alphabet. The same witness, with a bound on the dimension of the emitted system in place of the leakage bound, underlies the semi-device-independent QKD protocols of Refs.~\cite{onewaysdiqkd,erik_2013,PhysRevLett.115.150501,Woodhead_2016}. Under the information restriction, it was used in Ref.~\cite{Tavakoli2022} for randomness generation, where its largest quantum value at one bit of information was shown to coincide with the qubit value $2\sqrt2$.

\emph{The other protocols.} For $\prot44$, the CHSH value~\eqref{eq:chsh}, built from the tilted measurements $y=3,4$ instead of $y=2,3$, obeys the same bound and has the same ideal value.

For $\prot63$, whose ideal correlators are $C_{\theta,y}=\cos\bigl((\theta-y)\pi/3\bigr)$, three inequalities certify non-classicality. The first one is
\begin{equation}
    \label{eq:witness63}
    \sum_{\theta,y}C_{\theta,y}-2\bigl(C_{1,3}+C_{3,1}\bigr)\leq3+G ,
\end{equation}
the sum of all nine correlators, with the sign reversed for the two pairs of bases $120^\circ$ apart, for which the ideal correlator is negative. The ideal behaviour gives $6$, against $5$ at $G=2$. The second is a CHSH inequality between the bases $\theta=2,3$ and the measurements $y=1,3$,
\begin{equation}
    \label{eq:witness63b}
    C_{2,1}+C_{2,3}-C_{3,1}+C_{3,3}\leq G ,
\end{equation}
whose bound at $G=2$ is again the local bound $2$, while the ideal value is $5/2$.  The third one is
\begin{multline}
    \label{eq:witness63c}
    C_{1,1}+C_{1,2}-C_{1,3}-C_{3,1}+C_{3,2}+C_{3,3}\\
    +E_{(2,0),1}+E_{(2,0),2}+E_{(2,0),3}\leq1+2G ,
\end{multline}
which coincides with Eq.~\eqref{eq:witness63} on the bases $\theta=1,3$ and has the same values, $6$ for the ideal behaviour against $5$ at $G=2$, but a different dependence on $G$.

For $\prot33$, in which only one state of the $X$ basis is prepared, classical behaviours satisfy
\begin{multline}
    \label{eq:witness33}
    E_{(1,0),1}-E_{(1,1),1}+E_{(1,0),3}-E_{(1,1),2}\\
    +E_{(2,0),2}-E_{(2,0),3}\leq2G ,
\end{multline}
for which the ideal behaviour gives $2+2\sqrt2$, against $4$ at $G=2$, and
\begin{multline}
    \label{eq:witness33b}
    E_{(1,0),1}+E_{(1,0),2}+E_{(1,0),3}-E_{(1,1),2}\\
    +E_{(2,0),2}-E_{(2,0),3}\leq2G ,
\end{multline}
whose ideal value is $1+5/\sqrt2\approx4.54$.

We obtained these inequalities from the dual solutions of the linear programs of Sec.~\ref{sec:lp}. Their classical bounds are established in Appendix~\ref{app:classical}.

\subsection{Visibility and efficiency thresholds}
\label{sec:lp}
How far a behaviour lies from the classical set can be quantified in several ways. The classical information content $\Gmin(p)$ of a behaviour $p$ is the smallest $G$ such that $p\in\cset_G$. For the ideal behaviour, it measures how much information a classical strategy would need to reproduce the target statistics. 

At the honest leakage bound $G$, the robustness of non-classicality to noise can be measured by the classical visibility $\vcl$, the largest visibility $v$ at which the behaviour $v\,p +(1-v)\,p_{\mathrm{rand}}$ is still classical, where $p_{\mathrm{rand}}$ is the uniform behaviour with ${p_{\mathrm{rand}}(b|x,y)=1/2}$ for all $x$, $y$ and $b$. This noisy behaviour results from modelling imperfect preparation and channel noise by a depolarising channel of visibility $v$:
\begin{equation}
    \label{eq:depol}
    \rho_x \to v\,\rho_x + (1-v)\,\frac{\id}{2} .
\end{equation}
For the key basis the visibility is related to the quantum bit error
rate by $\mathrm{QBER}=(1-v)/2$.

We can also quantify the robustness of non-classicality to loss, represented  by an erasure channel of transmission $\eta\in[0,1]$, which appends a
no-click outcome $\varnothing$ to each of Bob's measurements,
\begin{equation}
    \label{eq:loss}
    M_{b|y}^{(\eta)} = \eta\, M_{b|y}, \qquad
    M_{\varnothing|y}^{(\eta)} = (1-\eta)\,\id .
\end{equation}
The parameter $\eta$ includes the detector efficiency and the channel transmission. If we bin the no-click outcome with the outcome $b=0$, the resulting behaviour is $\eta\,p(b|x,y)+(1-\eta)\,\delta_{b,0}$. The classical efficiency $\eta_{\mathrm{cl}}$ is the largest $\eta$ at which this behaviour is still classical. 

Since classical behaviours certify no key, $v>\vcl$ and $\eta>\eta_{\mathrm{cl}}$ are necessary conditions for a positive key rate at the honest leakage bound $G$. 

All three quantities $\Gmin$, $\vcl$, and $\eta_{\mathrm{cl}}$ can be determined by linear programs over $\cset_G$~\cite{Tavakoli2020,Tavakoli2022}. Table~\ref{tab:protocols} lists the results. The inequalities of Sec.~\ref{sec:facets} reproduce all of them. For instance, for $\prot43$ and $\prot44$, the CHSH inequality gives $\Gmin=2\sqrt2$, and, since the CHSH value becomes $2\sqrt2\,v$ under white noise and $2\sqrt2\,\eta$ under binned loss, $\vcl=\eta_{\mathrm{cl}}=1/\sqrt2$.

These classical thresholds are necessary conditions only; the key-rate thresholds of Sec.~\ref{sec:results} are higher, for instance $v_c=0.852$ and $\eta_c=0.791$ for $\prot43$.

\begin{table}[tb]
    \centering
    \begin{tabular}{lcccc}
        \toprule
        Protocol & $\Gmin$ & $\log_2\Gmin$ & $\vcl$ & $\eta_{\mathrm{cl}}$ \\
        \midrule
        $\prot43$ & $2\sqrt2$ & $3/2$ & $1/\sqrt2$ & $1/\sqrt2$ \\
        $\prot44$ & $2\sqrt2$ & $3/2$ & $1/\sqrt2$ & $1/\sqrt2$ \\
        $\prot63$ & $3$ & $\log_23$ & $4/5$ & $2/3$ \\
        $\prot33$ & $1+\sqrt2$ & $\log_2(1+\sqrt2)$ & $2\sqrt2-2$ & $(4+10\sqrt2)/23$ \\
        \bottomrule
    \end{tabular}
    \caption{Classical thresholds of the protocols of Table~\ref{tab:protocols_def}: the classical information content $\Gmin$ of the ideal behaviour, the corresponding number of bits $\log_2\Gmin$ that a classical strategy would need, against $\log_2G=1$ for the honest bound, and the white-noise visibility $\vcl$ and detection efficiency $\eta_{\mathrm{cl}}$ (with the no-click outcome binned) below which the behaviour is classical at the honest bound. All values were computed by linear programming (Appendix~\ref{app:classical}) and are reproduced by the inequalities of Sec.~\ref{sec:facets}.}
    \label{tab:protocols}
\end{table}

\subsection{Ideal implementations certify one bit of key}
\label{sec:ideal}
Beyond producing non-classical behaviours, the above protocols certify one bit of key per key round in the ideal implementation, at the honest leakage bound, as we now show.

\begin{proposition}
\label{prop:ideal}
For any of the protocols $\prot43$, $\prot44$, $\prot63$ and $\prot33$, assume the ideal target behaviour and the leakage bound $G=2$. Then the state $\rho_{AE\Theta}$ of Eq.~\eqref{eq:rho_AE} factorises as $\rho_{AE\Theta}=\rho_{A\Theta}\otimes\rho_E$, so that the key $A$ is independent of Eve's system $E$, i.e., $H(A|E\Theta)=1$. Since Alice's raw key symbol is uniformly random and perfectly correlated with Bob's outcome, i.e., $H(A|B)=0$, one bit of secret key per key round is certified.
\end{proposition}

\begin{proof}
The statement that $H(A|B)=0$ is immediate from the ideal correlations. To show that Eve has no information about the key, we show that, for the ideal behaviour, the protocol is equivalent to an entanglement-based protocol based on a purification of the state $\sigma$ of Eq.~\eqref{eq:sigma}. The ideal correlations then self-test a maximally entangled state, which decouples Eve.

Let $\sigma$ satisfy Eq.~\eqref{eq:sigma} with $G=2$, that is, $\sigma\geq\rho_x$ for every $x$ and $\Tr\sigma\leq2$. Bob's key measurement $y=\theta$ distinguishes the two states $\rho_{(\theta,0)}$ and $\rho_{(\theta,1)}$ of a key basis perfectly, $\Tr(M_{a|\theta}\rho_{(\theta,a)})=1$, so that
\begin{equation}
    2\geq\Tr\sigma=\sum_{a=0}^1\Tr(M_{a|\theta}\sigma)\geq\sum_{a=0}^1\Tr(M_{a|\theta}\rho_{(\theta,a)})=2 .
\end{equation}
All inequalities are thus equalities. For positive operators $P$ and $M$, $\Tr(MP)=\|\sqrt P\sqrt M\|^2=0$ implies $PM=0$, where $\|\cdot\|$ is the Hilbert--Schmidt norm. Applied to $\Tr\bigl[M_{a|\theta}(\sigma-\rho_{(\theta,a)})\bigr]=0$ and to $\Tr\bigl[(\id-M_{a|\theta})\rho_{(\theta,a)}\bigr]=0$, this gives $\sigma M_{a|\theta}=\rho_{(\theta,a)}M_{a|\theta}=\rho_{(\theta,a)}$, and summing over $a$,
\begin{equation}
    \label{eq:tau_sum}
    \sigma=\rho_{(\theta,0)}+\rho_{(\theta,1)}
\end{equation}
for every key basis $\theta$.

On the other hand, $\sigma-\rho_x$ is a positive operator of unit trace, hence a state, for every preparation $x$, and the state $\sigma/2$ thus admits the decompositions
\begin{equation}
    \label{eq:two_decompositions}
    \frac\sigma2=\tfrac12\rho_{(\theta,0)}+\tfrac12\rho_{(\theta,1)}=\tfrac12\rho_x+\tfrac12(\sigma-\rho_x)
\end{equation}
into two states with equal weights.

Let $\ket\Psi$ be a purification of $\sigma/2$ on $CBE$, where $C$ is an auxiliary system. Every decomposition of $\sigma/2$ into two states with weights $1/2$ is induced by a two-outcome measurement on $C$~\cite{Watrous2018}, which can be taken projective by enlarging $C$. Equation~\eqref{eq:two_decompositions} thus defines, for each key basis $\theta$, a $\pm1$ observable $A_\theta$ on $C$ whose outcome $(-1)^a$ leaves $BE$ in the state $\rho_{(\theta,a)}$, and, for each preparation $x$ outside the key bases, an observable $A'_x$ whose outcomes $+1$ and $-1$ leave $BE$ in the states $\rho_x$ and $\sigma-\rho_x$. The statistics of these virtual measurements with Bob's are fixed by the observed behaviour. Writing $\langle\cdot\rangle$ for expectation values in $\ket\Psi$ and $B_y=M_{0|y}-M_{1|y}$ for Bob's observables, with his measurements taken projective without loss of generality, Eq.~\eqref{eq:two_decompositions} gives
\begin{equation}
    \label{eq:virtual_corr}
    \begin{aligned}
    \langle A_\theta B_y\rangle&=\tfrac12\Tr\bigl[(\rho_{(\theta,0)}-\rho_{(\theta,1)})B_y\bigr]=C_{\theta,y},\\
    \langle A'_xB_y\rangle&=\Tr\bigl[(\rho_x-\tfrac12\sigma)B_y\bigr]\\&=E_{x,y}-\tfrac12\bigl(E_{(\theta,0),y}+E_{(\theta,1),y}\bigr),
    \end{aligned}
\end{equation}
with the biases and correlators of Eqs.~\eqref{eq:bias} and~\eqref{eq:correlator}, and $\langle A_\theta\rangle=\langle A'_x\rangle=0$.

The pure state $\ket\Psi$ on $CBE$, with Alice measuring $A_\theta$ or $A'_x$ on $C$, Bob measuring $B_y$ on $B$ and Eve holding $E$, defines an entanglement-based protocol in which Alice's and Bob's measurements act on separate systems of arbitrary dimension and Eve holds the purification of their joint state. It is equivalent to the P\&M protocol in two respects. First, its correlations between $C$ and $B$ are fixed by the observed behaviour, through Eq.~\eqref{eq:virtual_corr}. Second, it has the same key state. If Alice measures $A_\theta$ and $B$ is traced out, the classical--quantum state of her outcome and $E$ is $\tfrac12\sum_a\ket a\bra a\otimes\Tr_B(\rho_{(\theta,a)})$, and averaging over the key bases with the label $\theta$ in the register $\Theta$ gives $\rho_{AE\Theta}$ of Eq.~\eqref{eq:rho_AE}. Moreover, the marginal of $\ket\Psi$ on $E$ is the same whichever observable is measured on $C$. The factorisation $\rho_{AE\Theta}=\rho_{A\Theta}\otimes\rho_E$ therefore follows if we show that, for every key basis $\theta$, the state of $E$ conditioned on the outcome of $A_\theta$ does not depend on that outcome. This is what the ideal correlations establish.

For $\prot43$, consider the key observable $A\equiv A_1$ and the observable $A'\equiv A'_{(2,0)}$ of the preparation $\ket+$. At the ideal behaviour, Eq.~\eqref{eq:virtual_corr} gives $\langle AB_2\rangle=\langle AB_3\rangle=\langle A'B_2\rangle=-\langle A'B_3\rangle=1/\sqrt2$, hence the CHSH value
\begin{equation}
    \label{eq:virtual_chsh}
    S=\langle AB_2\rangle+\langle AB_3\rangle+\langle A'B_2\rangle-\langle A'B_3\rangle=2\sqrt2 .
\end{equation}
These are the ideal statistics of the CHSH-based DI protocol of Ref.~\cite{PhysRevLett.98.230501}, and the maximal CHSH value self-tests the maximally entangled state~\cite{PopescuRohrlich1992,MayersYao2004,McKague2012,Supic2020}. Specifically, for any pure state $\ket\Psi$ on $CBE$ and any $\pm1$ observables $A$, $A'$ on $C$ and $B_2$, $B_3$ on $B$ with $S=2\sqrt2$, there are local isometries on $C$ and on $B$ that map $\ket\Psi$ to $\ket{\Phi^+}_{C'B'}\otimes\ket\xi_{C''B''E}$, with $\ket{\Phi^+}$ a maximally entangled state of two qubits, and $A\ket\Psi$ to $(Z\otimes\id)\ket{\Phi^+}\otimes\ket\xi$. The outcome of $A$ is therefore uniformly random, and the state of $E$ conditioned on it is $\Tr_{C''B''}\ket\xi\bra\xi$ for both outcomes, which is the claim. The protocol $\prot33$ contains the same three preparations and the same measurements, so the same argument applies. In $\prot44$, both $A_1$ and $A_2$ are key observables, with the CHSH value $\langle A_1B_3\rangle+\langle A_1B_4\rangle+\langle A_2B_3\rangle-\langle A_2B_4\rangle=2\sqrt2$ on the tilted measurements. The self-test maps $A_1$ and $A_2$ to $Z$ and $X$ on $C'$, so the claim holds for both key bases.

For $\prot63$, all three bases are key bases, and the virtual observables $A_1$, $A_2$, $A_3$ reproduce with $B_1$, $B_2$, $B_3$ the statistics of a maximally entangled pair measured along three coplanar directions $60^\circ$ apart, $\langle A_\theta B_y\rangle=\cos\bigl((\theta-y)\pi/3\bigr)$. No pair of these observables violates the CHSH inequality maximally, the largest value being $5/2$, so the CHSH self-test cannot be invoked. Appendix~\ref{app:ideal} shows that these correlations nevertheless self-test a maximally entangled state in each key basis. The perfect correlations $A_\theta\ket\Psi=B_\theta\ket\Psi$ and the correlators between different bases force $A_1$ to anticommute on $\ket\Psi$ with an observable built from $A_2$ and $A_3$, which Bob reproduces as well. This is the structure from which the self-test of the maximal CHSH violation extracts $\ket{\Phi^+}$, obtained here from three bases instead of two, and it gives the claim for each $A_\theta$ in the same way.
\end{proof}

\section{Imperfect implementations}
\label{sec:imperfect}

QKD protocols are typically described, as in the previous section, by an honest and ideal implementation, free of noise, imperfections, and side channels. However, in a real implementation, the actual states $\rho_x$ and the actual measurements $M_{b|y}$ may differ from the target states and measurements. Such deviations have to be taken into account in the security analysis.

Our security analysis takes two inputs: the observed behaviour $p(b|x,y)$ and the leakage bound $G$. Together they certify that the behaviour is non-classical and, through the bound on $H(A|E\Theta)$, that key can be extracted. Imperfections affect these two inputs differently. They change the observed behaviour, and they may invalidate the leakage bound. The first effect does not compromise security. The analysis never assumes that the behaviour is the ideal one; it bounds the key rate for whatever behaviour is observed, and a degraded behaviour simply yields a lower rate, or none if it falls into the classical set, exactly as a reduced Bell violation lowers the rate of a DI protocol or brings it to zero. Deviations of Bob's measurements, which are uncharacterised, can only have this first effect. The second effect is the problematic one, and it concerns the source alone. The bound is an assumption, not an observation. If the actual source leaks more than $G$ allows, the certified rate is meaningless, and nothing in the observed data reveals it. The value of $G$ must therefore be chosen so that it holds for the actual source, including the expected, or modelled worst-case, imperfections. The rest of this section discusses how to do this. We do not provide a general recipe, but illustrate how the analysis can be applied to typical imperfections, and we give a few examples of how to model them.

Remember the setting of Fig.~\ref{fig:pnmqkd}. In each round, Alice's input $x$ is drawn inside her laboratory, and her source emits the system $Q$ into the channel. Eve controls the channel, and she may in addition hold a system $E_0$ that was correlated with the source before the round began. Everything that Eve learns about $x$, she learns from $QE_0$, since her final state, and with it the state $\rho_{AE\Theta}$ of Eq.~\eqref{eq:rho_AE} from which the key rate is computed, is obtained from $\rho_x^{QE_0}$ by her isometry $\mathcal U$. This is why the leakage bound~\eqref{eq:ri} is a statement about $\rho_x^{QE_0}$, and why, for the analysis to be valid, $Q$ and $E_0$ must comprise every system outside Alice's laboratory that depends on $x$ and that Eve may use: the intended signal, any unintended emission in other degrees of freedom, the light returned by a probe that Eve sends into the laboratory, and any record of the internal state of the source that Eve may hold. Conversely, nothing else needs to be counted. Information that never reaches Eve, because it stays inside Alice's laboratory, does not affect $\rho_{AE\Theta}$ and plays no role in the bound. The leakage bound is thus a bound on the total information about $x$ that is available to Eve, whatever its physical carrier.

In the honest implementation of the previous section, the states are modelled as
\begin{equation}
    \label{eq:honest_form}
    \rho_x^{QE_0}=\ket{\psi_x}\bra{\psi_x}\otimes\tau^{E_0},
\end{equation}
with $\tau^{E_0}$ independent of $x$, meaning that the source emits the target states, which are pure qubit states for the protocols of this work, and Eve's prior system is uncorrelated with them. The rescaled guessing probability $\Pgn$ of this ensemble is the honest bound, which we denote $G_{\mathrm{hon}}$; for our protocols $G_{\mathrm{hon}}=2$, one bit of information. It is the smallest bound that can be declared for these states, which contain orthogonal pairs. In a real implementation, imperfections modify the states~\eqref{eq:honest_form}. They are of two types: those that change the states without changing the leakage bound, and those that may increase it. 

\subsection{Imperfections that leave the bound valid}
\label{sec:imperfect_corr}

Many imperfections change the emitted states without increasing their distinguishability. As noted above, they change the observed behaviour and lower the key rate, but do not require a specific treatment in the security analysis. 

The simplest such imperfections are those that act on the emitted system independently of $x$: depolarisation or any other noise in the channel, loss, inefficient or mismatched detectors, and even active attacks on the detection devices. An $x$-independent channel cannot increase the guessing probability, so $\Pgn$, and with it the bound $G_{\mathrm{hon}}$, is unaffected. 

Imperfections that depend on $x$ may also leave the bound unchanged. For the protocols of this work, this follows from the dimension bound of Eq.~\eqref{eq:pg_dim}, according to which $G_{\mathrm{hon}}=2$ is the largest value of $\Pgn$ for any ensemble supported on a two-dimensional subspace, whatever the states within it, pure or mixed, aligned with the target states or not. Every imperfection that keeps the emission within the qubit subspace therefore leaves $G_{\mathrm{hon}}$ valid as it stands, even if it depends on $x$: a rotation of each prepared state away from its target by an $x$-dependent angle, a modulator that prepares $\cos\gamma_x\ket{\psi_x}+\sin\gamma_x\ket{\psi_x^\perp}$ instead of $\ket{\psi_x}$, a dephasing or depolarisation within the subspace whose strength depends on $x$, or a mixing of the target states among themselves. In all these cases the emitted states differ from the target states and the behaviour differs from the ideal one, but no measurement can guess $x$ with $\Pgn>2$.

The same holds for imperfections that fluctuate from round to round in a way that Eve may know, that is, for classical correlations between the source and Eve's prior system $E_0$. The security analysis itself does not care whether $E_0$ is correlated with the source, since it uses only the bound~\eqref{eq:ri_be} on the post-attack states; what matters is whether the correlations affect the value of the bound. 

The typical situation is a source with an internal state $\lambda$, such as the configuration of its components, the drift of a modulator or a fluctuation of the emitted mode, on which the emitted state $\rho_{x,\lambda}$ depends. To Alice and Bob, who do not know $\lambda$, the emission looks like the mixture $\bar\rho_x=\sum_\lambda p(\lambda)\rho_{x,\lambda}$. But Eve may hold a record of $\lambda$, for instance because she has measured the temperature of the laboratory, or because the source is simply a device of her making. Then $E_0$ carries $\lambda$, and $\rho_x^{QE_0}=\sum_\lambda p(\lambda)\ket\lambda\bra\lambda\otimes\rho_{x,\lambda}$, with $p(\lambda)$ independent of $x$. These states are block diagonal in $\lambda$, so an optimal measurement first reads $\lambda$, and
\begin{align}
    \Pgn\bigl(\{\rho_x^{QE_0}\}\bigr)&=\sum_\lambda p(\lambda)\,\Pgn\bigl(\{\rho_{x,\lambda}\}\bigr) \nonumber\\
    &\geq\Pgn\bigl(\{\bar\rho_x\}\bigr).
    \label{eq:pg_lambda}
\end{align}
Note first from the upper equality that the assumption bounds the average of the leakage conditioned on $\lambda$, not its worst case. The source may leak much more, or even reveal $x$ completely, for rare values of $\lambda$, provided that it leaks correspondingly less for the others; the multi-photon term below is an example, with $\lambda$ the photon number.

Second, the bottom inequality implies that a characterisation of the average states $\bar\rho_x$ does not suffice, however complete. For instance, if $\rho_{x,\lambda}=\ket{x+\lambda}\bra{x+\lambda}$ in dimension $n$, with addition modulo $n$ and $\lambda$ uniform, then $\bar\rho_x=\id/n$ for every $x$, whereas $\Pgn=n$: the source encrypts $x$ with a one-time pad whose key Eve holds. The characterisation must therefore hold conditionally on $\lambda$. 

Third, every bound that holds for each $\lambda$ carries over unchanged. In particular, if for every $\lambda$ the states $\rho_{x,\lambda}$ are supported on a qubit subspace with projector $\Pi_\lambda$, then $\sigma=\sum_\lambda p(\lambda)\ket\lambda\bra\lambda\otimes\Pi_\lambda$ satisfies Eq.~\eqref{eq:sigma} with $\Tr\sigma=2$. The variable $\lambda$ may thus change the encoded states, and even move the qubit subspace, at no cost, because the bound $G_{\mathrm{hon}}=2$ requires only that, for every internal state of the source, the emitted states be supported on some two-dimensional subspace, not that this subspace, or the states within it, be known. Fluctuations of this kind are therefore harmless, like the $x$-dependent imperfections within the qubit subspace discussed above, provided that the characterisation of the source holds conditionally on~$\lambda$.

Correlations of this type are known in the implementation-security literature as setting-choice-independent correlations~\cite{Mizutani2019,CurrasLorenzo2025}. In those analyses the characterisation must hold for every value of $\lambda$: the encoded phases and intensities must lie in given intervals~\cite{Mizutani2019}, or the fidelity to fixed pure reference states must exceed $1-\epsilon$~\cite{CurrasLorenzo2025}. Rounds that violate it can be tagged; their number must then be bounded, and they are assumed to reveal everything to Eve. The leakage bound needs only the average of the leakage over $\lambda$, not how it is distributed\footnote{This is true within the collective attacks considered here, where $\lambda$ is drawn independently in each round. Variables that drift over many rounds would need a more complex round-by-round formulation.}, and under the dimension argument it leaves the qubit states uncharacterised. 

\subsection{Imperfections that raise the bound}
\label{sec:raise}

An imperfection that takes the emission out of the qubit subspace in an $x$-dependent way can increase $\Pgn$ above $G_{\mathrm{hon}}$. Multi-photon emissions are an example: two copies of $\ket{\psi_x}$ are more distinguishable than one. So are side channels, that is, degrees of freedom of the emitted light other than the encoded one -- spectral, temporal or spatial modes -- that depend on $x$, and the light that a Trojan-horse probe sent by Eve into the laboratory brings back, modulated by the encoding device. In all these cases the emitted states are no longer supported on a two-dimensional subspace, and the dimension argument no longer applies.

Unlike the imperfections of the previous subsection, these do not necessarily leave a trace in the behaviour that Alice and Bob observe. A side channel, for instance, can reveal $x$ to Eve while Bob's statistics remain those of the target implementation. The bound must therefore be raised on the basis of a physical characterisation of the source. 

For instance, let the source emit a mixture of photon-number states with an $x$-independent distribution $\{p_k\}$, as for phase-randomised or heralded sources, and let $p_{\geq2}=\sum_{k\geq2}p_k$ be the multi-photon probability. Suppose, conservatively, that Eve learns the photon number and that the multi-photon components reveal $x$ completely. Then, by Eq.~\eqref{eq:pg_lambda} with $\lambda$ the photon number, $\Pgn\leq p_0+2p_1+n\,p_{\geq2}$, which exceeds $2$ by at most $(n-2)\,p_{\geq2}$.

As another example, consider side channels and Trojan-horse attacks. Write the single-photon component of the state as $\psi_x\otimes\tau_x$, where $\psi_x$ are arbitrary qubit states and $\tau_x$ collects the unintended degrees of freedom: spectral, temporal or spatial modes that depend on $x$, and the light returned in a Trojan-horse attack. If $F(\tau_x,\tau)\geq1-\epsilon$ for an $x$-independent state $\tau$, these add at most $n\sqrt\epsilon$ to $\Pgn$.
Indeed, a bound on the fidelity between the actual states and reference states of known guessing probability implies a leakage bound.

\begin{proposition}
\label{prop:fidelity}
Let $\rho_x$ be the states of the full emitted system and $\phi_x$ reference states on the same space with fidelity $F(\rho_x,\phi_x)=\|\sqrt{\rho_x}\sqrt{\phi_x}\|_1^2\geq1-\epsilon_x$. Then
\begin{equation}
    \label{eq:fidelity_bound}
    \Pgn\bigl(\{\rho_x\}\bigr)\leq\Pgn\bigl(\{\phi_x\}\bigr)+\sum_x\sqrt{\epsilon_x} .
\end{equation}
\end{proposition}

\begin{proof}
For $0\leq N_x\leq\id$, $\Tr(\rho_xN_x)-\Tr(\phi_xN_x)\leq\tfrac12\|\rho_x-\phi_x\|_1\leq\sqrt{1-F(\rho_x,\phi_x)}\leq\sqrt{\epsilon_x}$, the second step being the Fuchs--van de Graaf inequality~\cite{FuchsvandeGraaf1999}. Hence, for every POVM $\{N_x\}$,
\begin{align*}
    \sum_x\Tr(\rho_xN_x)&\leq\sum_x\Tr(\phi_xN_x)+\sum_x\sqrt{\epsilon_x}\\
    &\leq\Pgn\bigl(\{\phi_x\}\bigr)+\sum_x\sqrt{\epsilon_x} ,
\end{align*}
where the last step holds because $\{N_x\}$ is a valid, though in general not optimal, measurement for the reference states. The right-hand side does not depend on $\{N_x\}$, so maximising the left-hand side over POVMs gives the claim.
\end{proof}

Applied to $\rho_x=\psi_x\otimes\tau_x$ with the reference states $\psi_x\otimes\tau$, for which $\Pgn\leq2$ by Eq.~\eqref{eq:pg_dim} since the $\psi_x$ are qubit states, it gives $\Pgn(\{\rho_x\})\leq2+n\sqrt\epsilon$, whatever the qubit states $\psi_x$ are. If several unintended degrees of freedom contribute, $\tau_x=\bigotimes_j\tau_x^{(j)}$ with $F(\tau_x^{(j)},\tau^{(j)})\geq1-\epsilon_j$, then $\epsilon\leq\sum_j\epsilon_j$, because the fidelity of product states is the product of the fidelities. 

For a Trojan-horse attack the natural reference state $\tau$ is the vacuum, and the relevant parameter is the mean photon number $\mu_{\mathrm{out}}$ of the light that returns from the laboratory. A returned coherent pulse has $\epsilon_{\mathrm{THA}}=1-e^{-\mu_{\mathrm{out}}}\leq\mu_{\mathrm{out}}$, and more generally any returned state of mean photon number at most $\mu_{\mathrm{out}}$ has vacuum probability, hence fidelity to the vacuum, at least $1-\mu_{\mathrm{out}}$. Proposition~\ref{prop:fidelity} then gives $\Pgn\leq2+n\sqrt{\mu_{\mathrm{out}}}$. Bounds on $\mu_{\mathrm{out}}$, from the isolation of the source and the strongest probe that Eve can send into it, are provided, e.g., in Refs.~\cite{Lucamarini2015,Tamaki2016,Sixto2025}. 

Proposition~\ref{prop:fidelity} is not limited to side channels of product form. For the protocols of this work, Eq.~\eqref{eq:pg_dim} gives $\Pgn(\{\phi_x\})\leq2$ for any qubit states $\phi_x$, not only for the target states. The reference ensemble may therefore be taken to be the qubit ensemble closest to the actual states, so that $\epsilon$ quantifies only the part of the emitted states that lies outside the qubit subspace. A deviation $\epsilon_x\leq\epsilon$ from the qubit subspace, in this sense, translates into a relaxation $\delta G\leq n\sqrt\epsilon$ of the bound.

Note that the bound of Proposition~\ref{prop:fidelity} is tight. Consider the four states of $\prot43$ with a side channel in the state $\tau$ for three of them and in a state $\tau'$ with $|\braket{\tau|\tau'}|^2=1-\epsilon$ for the fourth, say $\ket+$. Eve can discriminate $\tau$ from $\tau'$ on the side channel, which succeeds with probability $\tfrac12(1+\sqrt\epsilon)$~\cite{Helstrom1969}, and then measure the qubit in the $X$ basis if she found $\tau'$ and in the $Z$ basis otherwise. This guesses $x$ with $\Pgn=2+\sqrt\epsilon$, which is the bound. When all four side-channel states deviate from $\tau$, with $\ket{\tau_x}=\sqrt{1-\epsilon}\,\ket\tau+\sqrt\epsilon\,\ket{e_x}$ and the $\ket{e_x}$ orthonormal and orthogonal to $\ket\tau$, one finds $\Pgn=2+2\sqrt{\epsilon(2-\epsilon)}$,\footnote{The square-root measurement is optimal here. The Gram matrix $\Gamma$ of the four states satisfies $(\Gamma-\id)^2=(1-\epsilon)^2\id$, so that $\sqrt\Gamma$ has the constant diagonal $\alpha=\tfrac12(\sqrt{2-\epsilon}+\sqrt\epsilon)$. The square-root measurement then guesses $x$ with $\Pgn=4\alpha^2$, and $\sigma=\alpha\,(\sum_x\rho_x)^{1/2}$ satisfies $\sigma\geq\rho_x$ with $\Tr\sigma=4\alpha^2$, so that $\Pgn=4\alpha^2$ by Eq.~\eqref{eq:duality}.} against the bound $2+4\sqrt\epsilon$. The square-root dependence on $\epsilon$ is thus unavoidable.

Note also that fidelity bounds also survive the classical correlations of Sec.~\ref{sec:imperfect_corr} when the reference states are pure and independent of $\lambda$. Applying Proposition~\ref{prop:fidelity} for each $\lambda$ and using the concavity of the square root yields Eq.~\eqref{eq:fidelity_bound} with $\epsilon_x=1-\bra{\phi_x}\bar\rho_x\ket{\phi_x}$, which is determined by the average states because the fidelity to a pure state is linear.

\subsection{Entanglement with Eve}
\label{sec:entanglement}

The honest form~\eqref{eq:honest_form} assumes that Eve's prior system $E_0$ is not entangled with the source. Entanglement is different from classical correlations. The dense-coding example of Sec.~\ref{sec:ri} shows that the dimension bound does not survive it, since a qubit that is maximally entangled with a qubit of Eve can carry two bits about $x$. The characterisation of the source must therefore either exclude entanglement with the outside or bound it.

A fidelity bound to pure reference states survives entanglement between the source and $E_0$, because a state close to a pure state can be only weakly entangled with any other system. Let $F(\rho_x^Q,\phi_x)\geq1-\epsilon_x$ with $\phi_x$ pure, purify $\rho_x^{QE_0}$ into $\ket{\Psi_x}$ on $QE_0R$, and let $\ket{a_x}=(\bra{\phi_x}\otimes\id)\ket{\Psi_x}$, so that $\braket{a_x|a_x}=\bra{\phi_x}\rho_x^Q\ket{\phi_x}\geq1-\epsilon_x$. The marginal $\omega_x$ of $\ket{\Psi_x}$ on $E_0R$ satisfies $\omega_x\geq\ket{a_x}\bra{a_x}$, hence $F(\Psi_x,\phi_x\otimes\omega_x)=\bra{a_x}\omega_x\ket{a_x}\geq(1-\epsilon_x)^2$. Tracing out $R$ cannot decrease the fidelity, and $\Tr_R\,\omega_x=\rho^{E_0}$, so
\begin{equation}
    \label{eq:fid_E0}
    F\bigl(\rho_x^{QE_0},\phi_x\otimes\rho^{E_0}\bigr)\geq(1-\epsilon_x)^2 .
\end{equation}
Since $\rho^{E_0}$ does not depend on $x$, the reference states $\phi_x\otimes\rho^{E_0}$ have the rescaled guessing probability $\Pgn(\{\phi_x\})$, and Proposition~\ref{prop:fidelity} applied to $QE_0$ gives
\begin{equation}
    \label{eq:fidelity_bound_E0}
    \Pgn\bigl(\{\rho_x^{QE_0}\}\bigr)\leq\Pgn\bigl(\{\phi_x\}\bigr)+\sum_x\sqrt{\epsilon_x(2-\epsilon_x)},
\end{equation}
whatever the correlations between the source and Eve. Compared with Eq.~\eqref{eq:fidelity_bound}, prior entanglement costs at most a factor $\sqrt2$ in the correction. For the Trojan-horse attack of Sec.~\ref{sec:raise}, with the vacuum as the pure reference state for the returned light, whose vacuum probability is at least $1-\mu_{\mathrm{out}}$ for any probe, this replaces $\epsilon_{\mathrm{THA}}$ by $1-(1-\mu_{\mathrm{out}})^2\leq2\mu_{\mathrm{out}}$ when the probe is entangled with a system of Eve.

As a corollary, a bound on the leakage from $Q$ alone, together with a bound on the purity of the emitted states, bounds the leakage of $QE_0$, whatever Eve's prior system. Let $1-\epsilon_x$ be the largest eigenvalue of $\rho_x^Q$ and $\phi_x$ the corresponding eigenvector; a bound $\Tr(\rho_x^Q)^2\geq1-\epsilon_x$ on the purity implies this, since $\Tr\rho^2\leq\lambda_{\max}\Tr\rho$. Then $F(\rho_x^Q,\phi_x)=\bra{\phi_x}\rho_x^Q\ket{\phi_x}=1-\epsilon_x$, and Eq.~\eqref{eq:fidelity_bound_E0} applies with the pure reference states $\phi_x$. Its right-hand side involves $\Pgn(\{\phi_x\})$, the leakage of the eigenvectors, which are not characterised; we replace it by the leakage of the actual states. Since $\phi_x$ is an eigenvector, $\rho_x^Q=(1-\epsilon_x)\phi_x+\epsilon_x\chi_x$ with $\chi_x$ a state orthogonal to $\phi_x$, so that $\tfrac12\|\rho_x^Q-\phi_x\|_1=\epsilon_x$ and $\Tr(\phi_xN_x)-\Tr(\rho_x^QN_x)\leq\epsilon_x$ for every $0\leq N_x\leq\id$. Summing over $x$ and maximising over POVMs, as in the proof of Proposition~\ref{prop:fidelity}, gives $\Pgn(\{\phi_x\})\leq\Pgn(\{\rho_x^Q\})+\sum_x\epsilon_x$, and Eq.~\eqref{eq:fidelity_bound_E0} becomes
\begin{align}
    \Pgn\bigl(\{\rho_x^{QE_0}\}\bigr)\leq{}&\Pgn\bigl(\{\rho_x^Q\}\bigr) \nonumber\\
    &+\sum_x\Bigl(\epsilon_x+\sqrt{\epsilon_x(2-\epsilon_x)}\Bigr).
    \label{eq:purity_bound}
\end{align}
At $\epsilon_x=0$ the correction vanishes, as it must, since a pure $\rho_x^Q$ forces $\rho_x^{QE_0}=\rho_x^Q\otimes\rho^{E_0}$, so that the leakage of $QE_0$ is that of $Q$. For small $\epsilon_x$, the square-root dependence cannot be avoided either. If Alice applies one of the four Pauli operators to a qubit prepared in the state $\sqrt{1-\epsilon}\ket{00}+\sqrt\epsilon\ket{11}$ with a qubit held by Eve, then $\Pgn(\{\rho_x^Q\})=2(1-\epsilon)$, whereas $\Pgn(\{\rho_x^{QE_0}\})=2+4\sqrt{\epsilon(1-\epsilon)}$, which interpolates between the qubit value $2$ at $\epsilon=0$ and dense coding, $\Pgn=4$, at $\epsilon=1/2$. The same example shows that Eq.~\eqref{eq:fidelity_bound_E0} is tight up to a factor $\sqrt2$ in the correction: with the eigenvectors $\phi_x\in\{\ket0,\ket1\}$ as reference states, it gives $2+4\sqrt{\epsilon(2-\epsilon)}$ against the actual $2+4\sqrt{\epsilon(1-\epsilon)}$. Only mixedness that may come from entanglement with Eve matters here. Indeed, noise applied after the encoding through an $x$-independent channel cannot increase $\Pgn$, so Eq.~\eqref{eq:purity_bound} may be applied to the states before such noise.

\subsection{Final bound on the leakage}
\label{sec:physical}

To illustrate how the bounds of this section combine, consider a source designed to emit qubit states, and write $\delta G=G-G_{\mathrm{hon}}$, with $G_{\mathrm{hon}}=2$, for the relaxation of the declared bound. By Sec.~\ref{sec:imperfect_corr}, imperfections within the qubit subspace, such as modulation errors or depolarisation, fluctuations of the encoding known to Eve, and noise that acts after the encoding through an $x$-independent channel, cost nothing, provided that the qubit is not entangled with Eve. Only what takes the emitted system out of a qubit subspace in an $x$-dependent way contributes, and the two examples of Sec.~\ref{sec:raise} bound the contributions separately: multi-photon emissions add at most $(n-2)\,p_{\geq2}$, and side channels add at most $n\sqrt{\epsilon_{\mathrm{modes}}+\mu_{\mathrm{out}}}$, where $\epsilon_{\mathrm{modes}}$ bounds the mismatch between the unintended modes emitted for different inputs and $\mu_{\mathrm{out}}$ the light returned by a Trojan-horse probe, the two fidelities being combined as in Sec.~\ref{sec:raise}. Altogether,
\begin{equation}
    \label{eq:budget}
    \delta G\leq (n-2)\,p_{\geq2}+n\sqrt{\epsilon_{\mathrm{modes}}+\mu_{\mathrm{out}}} .
\end{equation}
For illustration, take a multi-photon probability $p_{\geq2}=10^{-3}$, a mismatch $\epsilon_{\mathrm{modes}}=10^{-5}$ between the modes emitted for different inputs, and $\mu_{\mathrm{out}}=10^{-6}$. This gives $\delta G\leq2\cdot10^{-3}+1.3\cdot10^{-2}\approx1.5\cdot10^{-2}$ for the four-state protocols, which we round up to $\delta G=0.02$, and $\delta G\leq2.4\cdot10^{-2}$ for $\prot63$, which we round up to $0.03$. With an entangled probe, $\mu_{\mathrm{out}}$ is doubled and these values become $1.6\cdot10^{-2}$ and $2.5\cdot10^{-2}$. The square root makes the side channels the dominant term, so they must be characterised to high precision. Every term must be an upper bound obtained from a characterisation of the source; none can be estimated from the observed statistics. Where a characterisation is uncertain, the corresponding term may be chosen conservatively. The price is a lower certified key rate, not a loss of security.

\subsection{Relation to other descriptions of the source}
\label{sec:otherassumptions}

Several semi-device-independent and implementation-security analyses characterise the source by a few parameters, such as its dimension, energy or state overlaps~\cite{onewaysdiqkd,Vanhimbeeck,Ioannou_2022}, or fidelities to reference states~\cite{Pereira2019,Pereira2023,CurrasLorenzo2025,CurrasLorenzo2025num,Pereira2025}. Each of them implies a leakage bound and can therefore serve as an input to our analysis. Indeed, a model under which a protocol is secure must exclude $\Pg=1$, since perfect disclosure of $x$ would let Eve learn the key without disturbing Bob's statistics, and therefore implies a non-trivial leakage bound~\cite{Pauwels2025,BraskReview2026}. Proposition~\ref{prop:fidelity} makes the bound explicit for the fidelity-based descriptions, including the Trojan-horse bounds of Refs.~\cite{Lucamarini2015,Tamaki2016,Sixto2025}. The bounds implied by assumptions on the dimension, the energy, the vacuum component or the overlaps of the emitted states are derived systematically in Ref.~\cite{Pauwels2025}. Our analysis retains only the resulting bound and the observed statistics. It discards the rest of the model, which is what makes it insensitive to the details of the imperfections, possibly at the price of a lower noise tolerance. In exchange, the proof remains valid under deviations from the model, such as a violated fidelity bound, as long as $\Pgn\leq G$ holds; $G$ can be chosen conservatively to cover uncertainties in the characterisation. Furthermore, unlike in several of these analyses, Bob's device remains uncharacterised.

\section{Bounding the conditional entropy}
\label{sec:framework}

We now construct the lower bound on $H(A|E\Theta)$ necessary to compute the key rate~\eqref{eq:dw}. The quantity to compute is
\begin{equation}
    \label{eq:target}
    \min\; H(A|E\Theta)_{\rho_{AE\Theta}}
\end{equation}
over all states $\rho_x$ on $\mathcal H_B\otimes\mathcal H_E$ of unrestricted
dimension and all POVMs $\{M_{b|y}\}$ on $\mathcal H_B$ that reproduce the
observed behaviour $p(b|x,y)$ and satisfy the leakage
bound~\eqref{eq:ri_be}, with $\rho_{AE\Theta}$ given by Eq.~\eqref{eq:rho_AE}.

\subsection{A polynomial lower bound on the entropy}
\label{sec:entropy}

Write the conditional entropy as a relative entropy,
\begin{equation}
    H(A|E\Theta)_{\rho_{AE\Theta}} = -D\bigl(\rho_{AE\Theta}\,\|\,\id_A\otimes\rho_{E\Theta}\bigr),
\end{equation}
with $D(\rho\|\varsigma)=\Tr(\rho\log\rho)-\Tr(\rho\log\varsigma)$. We use the
variational bound of Koßmann and
Schwonnek~\cite{koßmann2025reliableentropyestimationobserved}, which is built
on an integral representation of the relative
entropy~\cite{Frenkel2023,kossmann2024_optimization}. For a partition $0<t_1<\dots<t_r=1$ of the integration domain, the bound of Ref.~\cite{koßmann2025reliableentropyestimationobserved} reads
\begin{equation}
    \label{eq:entropy_bound}
\begin{split}
    H(A|E\Theta)_{\rho_{AE\Theta}} \geq \frac{1}{\ln 2}\Bigl[
        \Tr\bigl(\id_A\otimes\rho_{E\Theta} - \rho_{AE\Theta}\bigr) \\
        {}+ \sum_{k=0}^r \inf_{P_k}
        \Tr\bigl(P_k(-\alpha_k\rho_{AE\Theta}-\beta_k\,\id_A\otimes\rho_{E\Theta})\bigr)
    \Bigr],
\end{split}
\end{equation}
where the infima are over projectors, $P_k^2=P_k=P_k^\dagger$, acting on $A$ together with Eve's system, and the real coefficients $\alpha_k,\beta_k$ depend only on the partition. They are given in Appendix~\ref{app:coefficients}. The bound becomes tight as the partition is refined.

Two features of Eq.~\eqref{eq:entropy_bound} make it suitable here. It is a
polynomial of low degree in $\rho_{AE\Theta}$ and the projectors $P_k$, and it
requires no assumption on the purity of the preparations. The projectors $P_k$ act on Eve's system and
therefore commute with Bob's measurement operators,
\begin{equation}
    \label{eq:commutation}
    [P_k, M_{b|y}] = 0 .
\end{equation}
We use Ref.~\cite{koßmann2025reliableentropyestimationobserved} rather than
the Brown--Fawzi--Fawzi construction~\cite{brown2024device} because its
variational objects are projectors rather than general non-Hermitian
operators, which yields smaller moment matrices at the same level. We are not
aware of a general statement on which of the two converges faster. 

Substituting Eq.~\eqref{eq:rho_AE} into Eq.~\eqref{eq:entropy_bound} and using
the classical--quantum structure of $\rho_{AE\Theta}$ reduces the objective to a
polynomial in the states $\rho_x$ and a set of
projectors on Eve's system alone. The reduction is carried out in
Appendix~\ref{app:hae}; the resulting bound is
\begin{equation}
\label{eq:objective}
    H(A|E\Theta) \;\geq\; \frac{1}{\ln 2}\Bigl[\, 1 \,+\,
    \sum_{k=0}^{r}\inf \sum_{\theta=1}^{\nu} \Tr\bigl(\Omega_k^\theta\bigr) \Bigr] ,
\end{equation}
where the infimum is over the projectors, the states and Bob's measurements, and
\begin{equation}
\label{eq:objective_term}
\begin{split}
    \Omega_k^\theta = {}& -\frac{\alpha_k}{2\nu}
        \bigl(P_{0,k}^\theta \rho_{(\theta,0)} + P_{1,k}^\theta \rho_{(\theta,1)}\bigr) \\
    & -\frac{\beta_k}{2\nu}\bigl(P_{0,k}^\theta+P_{1,k}^\theta\bigr)
        \bigl(\rho_{(\theta,0)}+\rho_{(\theta,1)}\bigr) .
\end{split}
\end{equation}
Here $\theta=1,\dots,\nu$ runs over the key bases and $\rho_{(\theta,a)}$ is the state of the key input $(\theta,a)$ (Sec.~\ref{sec:generic}). The projectors $P_{a,k}^\theta$ acting on Eve's system carry the basis index $\theta$ (reflecting the fact that Eve's information is made of a quantum part $E$ and a classical part $\Theta$).

\subsection{Moment relaxation}
\label{sec:tracial}

Collecting Eqs.~\eqref{eq:sigma}, \eqref{eq:objective}, \eqref{eq:objective_term} and~\eqref{eq:commutation}, the bound on the
entropy is the optimal value of the non-commutative polynomial program
\begin{align}
    \inf_{\{\rho_x\},\sigma,\{M_{b|y}\},\{P^\theta_{a,k}\}}
    & \Tr\bigl(\Omega(\rho_x, P^\theta_{a,k})\bigr), \nonumber\\
    \text{s.t.}\quad
    & \rho_x \geq 0, \quad \Tr\rho_x = 1, \nonumber\\
    & \sigma-\rho_x \geq 0, \quad \Tr\sigma \leq G, \nonumber\\
    & M_{b|y}M_{b'|y} = \delta_{b,b'}M_{b|y}, \quad \textstyle\sum_b M_{b|y}=\id, \nonumber\\
    & \Tr(\rho_x M_{b|y}) = p(b|x,y), \nonumber\\
    & (P^\theta_{a,k})^2 = P^\theta_{a,k}, \quad [P^\theta_{a,k}, M_{b|y}] = 0 ,
    \label{eq:ncpop}
\end{align}
where $\Omega=\sum_{k,\theta}\Omega_k^\theta$ is the polynomial of
Eqs.~\eqref{eq:objective}--\eqref{eq:objective_term}, which involves only the key states and Eve's projectors; Bob's measurements enter through the constraints alone.\footnote{Ref.~\cite{Tavakoli2022} imposes in addition the constraints $\rho_x-\rho_x^2\geq0$ and $G\id-\sigma\geq0$. They are redundant, since $\rho_x\geq0$ and $\Tr\rho_x=1$ imply $\rho_x\leq\id$, and $\sigma\geq0$ and $\Tr\sigma\leq G$ imply $\sigma\leq G\id$, but they make the feasible set Archimedean, the standard condition under which tracial moment hierarchies converge~\cite{burgdorf2016optimization}. We drop them because the relaxation below retains only the moments with a single trace-class factor (Proposition~\ref{lem:onefactor}), on which $\rho_x-\rho_x^2\geq0$ reduces to $\rho_x\geq0$ and $G\id-\sigma\geq0$ contributes nothing. Convergence plays no role in the security claims, since each level is used only as a lower bound.} Bob's measurements are taken projective without loss of generality.

Problem~\eqref{eq:ncpop} is an instance of \emph{tracial} non-commutative
polynomial optimisation~\cite{PNA2010,burgdorf2016optimization}: the objective and the
constraints involve traces of polynomials rather than expectation values in a
fixed state. It can be solved through semidefinite relaxations \cite{Tavakoli2022,glty-kkbp} built from the moments $L(w)=\Tr(w)$ of words $w$ in the operators
\begin{equation}
    \label{eq:operators}
    \mathcal O = \bigl\{\{\rho_x\}_x,\ \sigma,\
                  \{M_{b|y}\}_{b,y},\ \{P_{a,k}^\theta\}_{\theta,a,k}\bigr\}\,.
\end{equation}
Such SDP relaxation impose that the moment matrix $\Gamma_{u,v}=L(u^\dagger v)$ and the localising matrices $\Gamma^{(g)}_{u,v}=L(u^\dagger g\,v)$, for
$g\in\{\rho_x,\ \sigma-\rho_x\}$, be positive semidefinite, with rows and columns indexed by chosen sets of monomials $u,v$ in the variables \eqref{eq:operators}. The remaining conditions of Eq.~\eqref{eq:ncpop} become linear
identities among the moments. These are: $L(\rho_x)=1$, $L(\sigma)\leq G$, $L(\rho_xM_{b|y})=p(b|x,y)$; the identities that follow from
projectivity of $M_{b|y}$ and $P_{a,k}^\theta$, from completeness
$\sum_b M_{b|y}=\id$, from the commutation
relations~\eqref{eq:commutation}, and from cyclicity of the trace. Any realisation yields moments satisfying all of these conditions, so minimising the objective -- a linear
functional of the moments -- is a semidefinite program whose value lower-bounds Eq.~\eqref{eq:ncpop}.

\emph{Moments with $k$ trace-class factors.}  The operators $\rho_x$ and $\sigma$ are the only ones in Eq.~\eqref{eq:operators} with a finite trace; we call them the trace-class operators, and we classify the moments $L(w)$ by the number $k$ of trace-class factors in the word $w$.

Only the moments with $k=1$ carry the data of the problem. Indeed, the objective is a sum of terms $\Tr(P^\theta_{a,k}\rho_{(\theta,a)})$ and the linear conditions are $\Tr\rho_x=1$, $\Tr\sigma\leq G$ and $\Tr(\rho_xM_{b|y})=p(b|x,y)$, all of the form $\Tr(\rho_xw)$ or $\Tr(\sigma w)$ with $w$ a word in Bob's and Eve's projectors. Furthermore, the algebraic identities among the moments, which come from projectivity, completeness, the commutation relations~\eqref{eq:commutation} and cyclicity of the trace, never relate moments with different values of $k$. The moments with $k=0$ or $k\geq2$ thus appear only as entries of the moment and localising matrices, that is, only through positivity constraints, and are otherwise unrelated to the $k=1$ moments. This suggests simplifying the relaxation by dropping the rows and columns of the moment and localising matrices that are indexed by monomials with $k=0$ or $k\geq2$.

\emph{The one-factor relaxation.} Let $W$ be a set of words in Bob's and Eve's projectors, and consider the program whose variables are the moments $L(\zeta w)$, with $\zeta\in\{\rho_x,\sigma\}$ and $w$ a word in Bob's and Eve's projectors, subject to the conditions $L(\rho_x)=1$, $L(\sigma)\leq G$ and $L(\rho_xM_{b|y})=p(b|x,y)$, to the identities among the moments that follow from projectivity, completeness and commutation, and to
\begin{equation}
    \label{eq:swnpa}
    \Gamma^{(\rho_x)}\succeq 0, \qquad \Gamma^{(\sigma-\rho_x)}\succeq0 \qquad (x=1,\dots,n),
\end{equation}
with rows and columns indexed by $W$. Here $\Gamma^{(\zeta)}_{u,v}=L(\zeta u^\dagger v)$ is the usual NPA moment matrix~\cite{Navascues2007,Navascues2008} of the words $W$, with the unnormalised operator $\zeta$ in place of the state. It coincides with the localising matrix of $\zeta$ defined above by cyclicity of the trace (up to the relabelling $u\to u^\dagger$ of its rows and columns).

Moment matrices indexed by the preparations $\rho_x$ were introduced for contextuality scenarios in Ref.~\cite{Chaturvedi2021}, and the one-factor relaxation of informationally restricted correlations, including the operator $\sigma$, appears in Ref.~\cite{Manna2024}. Every entry is finite for trace-class $\zeta$, so the relaxation is manifestly valid in any dimension, and it is much smaller than the tracial relaxation on the same words, since no moment with two or more trace-class factors appears. The next proposition shows that such one-factor relaxations are as strong as the full tracial relaxations.

\begin{proposition}
\label{lem:onefactor}
Let $W$ be a set of words in Bob's and Eve's projectors, and consider a tracial relaxation of Eq.~\eqref{eq:ncpop} whose localising matrices of $\rho_x$ and $\sigma-\rho_x$ are indexed by a set of monomials $V$ whose words without a trace-class factor are exactly those of $W$. Then its certified dual bound is equal to that of the one-factor relaxation~\eqref{eq:swnpa} on $W$. If in addition the tracial relaxation has no duality gap, then its optimal value is also equal to that of the one-factor relaxation.
\end{proposition}
\begin{proof}
Let $\alpha$, $\beta$ be the primal and dual values of the full tracial relaxation and $\alpha_1,\beta_1$ the primal and dual values of the one-factor relaxation.  We show below that the dual bounds coincide, $\beta=\beta_1$.

Furthermore, the tracial relaxation is at least as strong as the one-factor relaxation, since every feasible point of the former yields a feasible point of the latter with the same objective. Hence $\alpha_1\leq\alpha$. If, in addition, the tracial relaxation has no duality gap, then $\alpha=\beta$, and thus $\alpha = \beta=\beta_1\leq \alpha_1$, where the last inequality follows from weak duality. We thus deduce that $\alpha=\alpha_1$, and the two relaxations have the same optimal value.

To show  $\beta=\beta_1$, we first make the dual of the tracial relaxation explicit. Its objective is $L(\Omega)$, with $\Omega=\sum_{k,\theta}\Omega^\theta_k$ the polynomial of Eqs.~\eqref{eq:objective}--\eqref{eq:objective_term}. Let $U$ be the set of monomials indexing the moment matrix $\Gamma$ and $V$ the set of monomials indexing the localising matrices $\Gamma^{(g)}$. 
For a vector of coefficients $c=(c_u)$ indexed by $U$ and the polynomial $p=\sum_uc_uu$, one has $c^\dagger\Gamma c=\sum_{u,v}c_u^*L(u^\dagger v)c_v=L(p^\dagger p)$. Likewise, for coefficients indexed by $V$, $c^\dagger\Gamma^{(g)}c=L(p^\dagger gp)$. The positivity constraints of the relaxation thus state that $L(p^\dagger p)\geq0$ and $L(p^\dagger gp)\geq0$ for all such polynomials. A dual certificate of a bound $\beta$ consists of real multipliers $\lambda_x$, $\nu_{bxy}$ and $\mu\geq0$, and of polynomials $p_i$ with monomials in $U$ and $a_{xi}$, $b_{xi}$ with monomials in $V$, such that the polynomial
\begin{equation}\label{eq:certificate}
    h=\Omega-\sum_x\lambda_x\rho_x-\sum_{b,x,y}\nu_{bxy}\,\rho_xM_{b|y}+\mu\,\sigma
\end{equation}
satisfies
\begin{equation}\label{eq:certificate2}
    h\equiv\sum_ip_i^\dagger p_i+\sum_{x,i}a_{xi}^\dagger\rho_xa_{xi}+\sum_{x,i}b_{xi}^\dagger(\sigma-\rho_x)b_{xi},
\end{equation}
where $\equiv$ denotes equality of the two polynomials after reduction by the algebraic relations (projectivity, completeness and commutation) and cyclic reordering, that is, equality of $L$ on both sides for every $L$ that satisfies the identities of the relaxation. Applying a feasible $L$, the right-hand side of \eqref{eq:certificate2} is non-negative, thus $L(h)\geq0$, and since $L(\rho_x)=1$, $L(\rho_xM_{b|y})=p(b|x,y)$ and $\mu L(\sigma)\leq\mu G$, \eqref{eq:certificate} gives
\begin{equation*}
    L(\Omega)\;\geq\;\sum_x\lambda_x+\sum_{b,x,y}\nu_{bxy}\,p(b|x,y)-\mu G\;=:\;\beta .
\end{equation*}
The certified bound of the relaxation is the largest $\beta$ of this form. It is what the solver returns, and it is a valid lower bound on $L(\Omega)$, hence on the entropy, whether or not it equals the optimal value of the relaxation.

Every term of $h$ has exactly one trace-class factor, and the identities among the moments never change this number, so the identity holds separately for the terms with zero, one, and two or more trace-class factors.

We now show that the moments without a trace-class factor, such as $\Tr(\id)$, $\Tr(M_{b|y})$ or $\Tr(M_{b|y}P^\theta_{a,k})$, play no role in the certificate. These moments are infinite in general, since the Hilbert space is unrestricted; they are left as free variables of the relaxation, so that in particular $\Tr\id$ is not normalised. Write $p_i=p_{i,0}+p_{i,+}$, where $p_{i,0}$ collects the monomials of $p_i$ without a trace-class factor. The terms without a trace-class factor on the right-hand side are $\sum_ip_{i,0}^\dagger p_{i,0}$, and there are none on the left, so $\sum_ip_{i,0}^\dagger p_{i,0}\equiv0$ modulo the identities among the words in Bob's and Eve's projectors. 

This forces $p_{i,0}=0$, as we show by exhibiting a trace on the algebra of Bob's and Eve's projectors that vanishes only for the zero polynomial. This algebra is generated by the operators $M_y = M_{0|y} - M_{1|y}$ and $Q^\theta_{a,k} = 2 P^\theta_{a,k}-\id$, which are Hermitian and square to the identity, and satisfy the commutation relations $[M_y,Q^\theta_{a,k}]=0$. All relations satisfied by Bob's and Eve's operators are consequences of these, so  the polynomials in Bob's and Eve's projectors are exactly the elements $b=\sum_gb_g\,g$ of the group algebra of the group $\mathcal G$ generated by these operators, and the reduced words are the elements $g\in\mathcal G$.

 On this group algebra, let $\tau(b)=b_e$ be the coefficient of the identity element $e$. It is a trace, since $gh=e$ if and only if $hg=e$, so it is well defined modulo cyclic reordering. Furthermore, $b^\dagger b=\sum_{g,h}b_g^*b_h\,g^{-\dagger}h$ has identity coefficient $\sum_g|b_g|^2$ (since $g^{-\dagger}h=e$ if and only if $g=h$), so $\tau(b^\dagger b) = \sum_g|b_g|^2$ vanishes only for $b=0$. Applying $\tau$ to $\sum_ip_{i,0}^\dagger p_{i,0}\equiv0$ gives $\sum_i\tau(p_{i,0}^\dagger p_{i,0})=0$, a sum of non-negative terms, hence $p_{i,0}=0$ for every $i$. Every certificate thus vanishes on the rows and columns of the moment matrix indexed by monomials without a trace-class factor, and these rows can be deleted without changing the certified bound.

Since every monomial of the $p_i$ now contains a trace-class factor, every term $p_i^\dagger p_i$ has at least two. Collecting the terms with exactly one gives
\begin{equation*}
    h\equiv\sum_{x,i}a_{xi,0}^\dagger\rho_xa_{xi,0}+\sum_{x,i}b_{xi,0}^\dagger(\sigma-\rho_x)b_{xi,0},
\end{equation*}
where $a_{xi,0}$ and $b_{xi,0}$ are the parts of $a_{xi}$ and $b_{xi}$ without a trace-class factor, which lie in the span of $W$. But this is a certificate of the same bound $\beta$ for the one-factor relaxation~\eqref{eq:swnpa}, so that $\beta_1\geq\beta$. Conversely, a certificate for Eq.~\eqref{eq:swnpa} involves only polynomials in the span of $W\subseteq V$ and is therefore a certificate for the tracial relaxation, so that $\beta\geq\beta_1$. The two dual bounds thus coincide, $\beta=\beta_1$.
\end{proof}

The proposition relies on the states entering only linearly in the problem. If they were required to be pure, $\rho_x^2=\rho_x$, the moments with several trace-class factors would be tied to those with one, and the full tracial hierarchy would be needed; under the leakage bound, however, purity cannot be assumed (Sec.~\ref{sec:ri}). As in the hierarchy of Ref.~\cite{Tavakoli2022}, we do not claim convergence of the hierarchy to the exact value; each finite level is used only as a relaxation, which is sufficient to prove security.

\section{Results}
\label{sec:results}

The results of this section are computed with the reduced relaxation~\eqref{eq:swnpa} of Proposition~\ref{lem:onefactor}. For $\prot43$ and $\prot33$, $W$ consists of all words of length at most three in Bob's and Eve's projectors; for $\prot44$ and $\prot63$, of all words of length at most two together with the products $M_{0|y}M_{0|y'}P^\theta_{a,k}$. The entropy bound of Eq.~\eqref{eq:objective} is evaluated on a partition with $26$ nodes, and the term of each node is minimised separately, which is a further relaxation~\cite{brown2024device}. The semidefinite programs are solved with MOSEK, and the dual solution is converted into a certified lower bound. Every reported value is therefore a lower bound on the key rate. The longer words and finer partitions we tested lower the thresholds by at most $3\cdot10^{-3}$. Appendix~\ref{app:implementation} gives the details of the implementation, and Appendix~\ref{app:checks} its validation and further results.


\subsection{Noise model}
\label{sec:noise}

Imperfect state preparation and channel noise are described by the depolarising channel of visibility $v$, Eq.~\eqref{eq:depol}, and loss by the erasure channel of transmission $\eta$, Eq.~\eqref{eq:loss}, which appends a no-click outcome $\varnothing$ to each of Bob's measurements. 

The no-click outcome is treated differently in the two terms of
Eq.~\eqref{eq:dw}. In the reconciliation term $H(A|B)$ we keep $\varnothing$
as a separate outcome, since it is part of Bob's actual data and he uses it when correcting his string. Since Bob measures every key basis exactly, this gives $H(A|B)=1-\eta+\eta\,h\bigl((1-v)/2\bigr)$, where $h$ is the binary entropy. In the entropy term $H(A|E\Theta)$ we bin $\varnothing$ into
the outcome $b=0$ of each measurement, as in Sec.~\ref{sec:lp}, i.e.\ we use the two-outcome POVMs $\{\eta M_{0|y}+(1-\eta)\id,\ \eta M_{1|y}\}$. This removes one operator from the
semidefinite program. For the protocols of Table~\ref{tab:protocols}, binning leaves $\eta_{\mathrm{cl}}$ unchanged.

We keep the honest bound $G$ at the value of the noiseless target states, as appropriate when the noise arises in the channel. If the noise arises in the source, the honest value for the depolarised target states is smaller, $\Pgn=1+v$, and a correspondingly tighter bound could be declared. Unlike $G=2$, which holds for any qubit states by Eq.~\eqref{eq:pg_dim}, this tighter bound relies on the precise form of the states. This makes a large difference, as the visibility thresholds with $G=1+v$ become $v_c=0.727$ for $\prot43$ and $0.715$ for $\prot44$, instead of $0.852$ and $0.838$ at fixed $G$ (Appendix~\ref{app:checks}). These values rest on a stronger assumption, however.

\subsection{Key rates}
\label{sec:keyrates}

Figure~\ref{fig:keyrates} shows the key rate per key round, $K_\infty$, as a function of visibility at $\eta=1$ and as a function of detection efficiency at $v=1$, for $\prot43$, $\prot44$, $\prot63$ and $\prot33$ at the honest $G$. We plot $K_\infty$ rather than the rate per signal $R=\pkey K_\infty$. 
Figure~\ref{fig:dG} shows the thresholds when the declared bound is relaxed to $G+\delta G$, as functions of $\delta G$. Table~\ref{tab:thresholds} collects these thresholds for some values of $\delta G$. 

\begin{figure*}[tb]
    \centering
    \includegraphics[width=\textwidth]{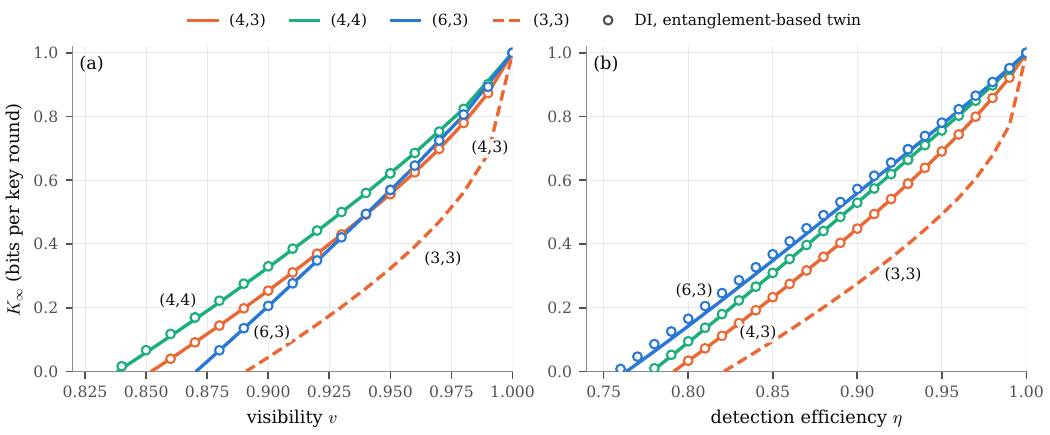}
    \caption{Certified key rate per key round, $K_\infty$, (a) versus visibility $v$ at unit detection efficiency and (b) versus detection efficiency $\eta$ at unit visibility, at the honest bound $G=2$. The curves are lower bounds under the leakage bound, computed as described in Appendix~\ref{app:implementation}; each ends at the threshold of Table~\ref{tab:thresholds}. Open circles are the device-independent rates of the entanglement-based twins of $\prot43$, $\prot44$ and $\prot63$, computed from the same statistics (Sec.~\ref{sec:comparison}). The exact DI rate upper-bounds what can be certified under the leakage bound. The circles are finite NPA relaxations of it, and where one lies below a curve, as for $\prot44$ and $\prot63$ near the ideal point, that relaxation has not converged.}
    \label{fig:keyrates}
\end{figure*}

\begin{table}[tb]
    \centering
    \begin{tabular}{lccccc}
        \toprule
        & \multicolumn{4}{c}{LB, $\delta G=$} & \\
        \cmidrule(lr){2-5}
        Protocol & $0$ & $0.04$ & $0.08$ & $0.12$ & DI twin \\
        \midrule
        \multicolumn{6}{l}{\textit{visibility} $v_c$ (at $\eta=1$)} \\
        $\prot43$ & $0.852$ & $0.865$ & $0.879$ & $0.891$ & $0.852$ \\
        $\prot44$ & $0.838$ & $0.852$ & $0.864$ & $0.875$ & $0.837$ \\
        $\prot63$ & $0.870$ & $0.880$ & $0.889$ & $0.897$ & $0.870$ \\
        $\prot33$ & $0.891$ & $0.905$ & $0.919$ & $0.930$ & --- \\
        \midrule
        \multicolumn{6}{l}{\textit{detection efficiency} $\eta_c$ (at $v=1$)} \\
        $\prot43$ & $0.791$ & $0.810$ & $0.827$ & $0.843$ & $0.791$ \\
        $\prot44$ & $0.779$ & $0.794$ & $0.808$ & $0.821$ & $0.778$ \\
        $\prot63$ & $0.764$ & $0.785$ & $0.803$ & $0.820$ & $0.758$ \\
        $\prot33$ & $0.821$ & $0.841$ & $0.860$ & $0.878$ & --- \\
        \bottomrule
    \end{tabular}
    \caption{Thresholds for a positive key rate: the visibility $v_c$ (upper block) and detection efficiency $\eta_c$ (lower block) below which the certified rate vanishes, when the declared bound is relaxed from the honest value $G$ to $G+\delta G$. At $v=\eta=1$ and $\delta G=0$ every protocol certifies $K_\infty\geq0.997$ bits per key round. The last column gives the thresholds of the device-independent bound for the entanglement-based twin (Sec.~\ref{sec:comparison}); no twin was computed for $\prot33$.}
    \label{tab:thresholds}
\end{table}

\begin{figure*}[tb]
    \centering
    \includegraphics[width=\textwidth]{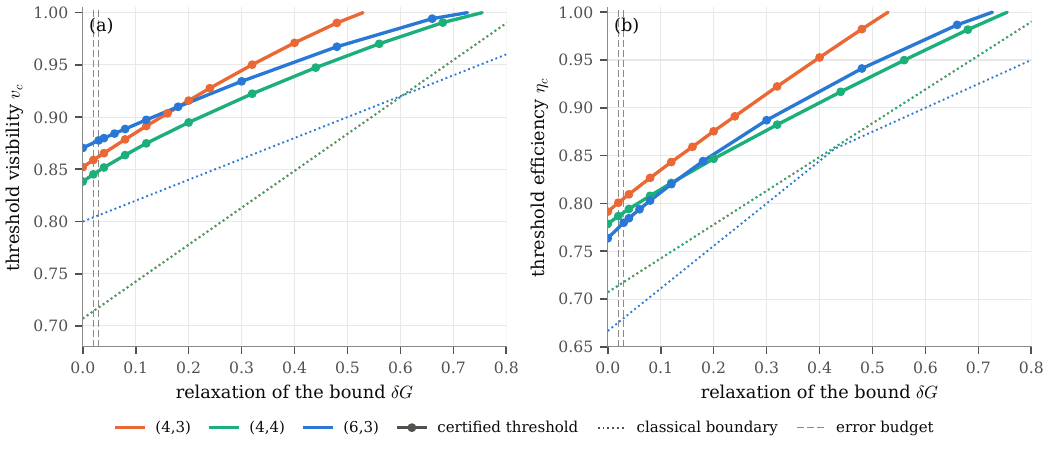}
    \caption{Thresholds for a positive key rate as functions of the relaxation $\delta G$ of the declared bound, $G=2+\delta G$: (a) visibility $v_c$ at $\eta=1$ and (b) detection efficiency $\eta_c$ at $v=1$, computed with the relaxations of Table~\ref{tab:thresholds}. Each curve ends at the largest $\delta G$ at which the ideal behaviour still certifies key. Dotted lines are the classical boundaries $\vcl$ and $\eta_{\mathrm{cl}}$, below which the behaviour is classical at $G=2+\delta G$; those of $\prot43$ and $\prot44$ coincide. The dashed lines mark the error budgets of Sec.~\ref{sec:physical}, $\delta G=0.02$ for the four-state protocols and $0.03$ for $\prot63$.}
    \label{fig:dG}
\end{figure*}

The following features stand out.

\emph{Ideal rate.} At $v=\eta=1$ and the honest $G$, every protocol certifies $K_\infty\geq0.997$, against the exact value of one bit established in Proposition~\ref{prop:ideal}. The small deficit is numerical.

\emph{Visibility.} The $\prot43$ protocol reaches $v_c=0.852$. The $\prot44$ protocol, which also generates key in the $X$ basis, reaches $v_c=0.838$, the best visibility threshold of our protocols. The two tilted measurements reproduce the random-access-code configuration, whose statistics are furthest from the classical set; a single tilted measurement is not enough (see below). Across our protocols, the visibility thresholds follow the ordering of the classical thresholds $\vcl$ of Table~\ref{tab:protocols}. Within the family $\cos\gamma\,Z\pm\sin\gamma\,X$, the choice $\gamma=\pi/4$ is optimal for $\prot43$ among the angles we tested (Appendix~\ref{app:checks}).

\emph{Loss.} The visibility threshold of the $(6,3)$ protocol, $v_c=0.870$, is higher than those of $\prot43$ and $\prot44$, but its efficiency threshold, $\eta_c=0.764$, is the lowest of all our protocols. The symmetric arrangement of the bases matters: the bases of $Z$, $X$ and $(Z+X)/\sqrt2$ give $v_c=0.892$ and $\eta_c=0.795$.

\emph{Three preparations.} The $\prot33$ protocol, which drops the state $\ket-$ from $\prot43$ and keeps its measurements, reaches $v_c=0.891$ and $\eta_c=0.821$, only $0.03$--$0.04$ above the thresholds of $\prot43$. Its rate, the dashed curve of Fig.~\ref{fig:keyrates}, however, falls much faster below the ideal point, and at $v=0.95$ or $\eta=0.95$ it certifies $0.32$ and $0.49$ bits, against $0.56$ and $0.69$ for $\prot43$. Three preparations thus suffice to generate a key, but the fourth state matters away from the ideal point.

\emph{Sensitivity to $G$.}  Relaxing the bound from one bit to $\log_2(2+\delta G)\approx1.08$ bits, i.e.\ $\delta G=0.12$, costs $0.039$, $0.037$ and $0.027$ in threshold visibility and $0.052$, $0.042$ and $0.056$ in threshold efficiency for $\prot43$, $\prot44$ and $\prot63$. The efficiency threshold of $\prot63$ rises fastest, so that its advantage in loss tolerance is lost beyond $\delta G\approx 0.13$, where $\prot44$ takes over. By Proposition~\ref{prop:fidelity}, a fidelity imperfection $\epsilon$ outside the qubit subspace gives $\delta G=n\sqrt\epsilon$, so that at equal $\epsilon$ the relaxation is larger for $\prot63$ than for the four-state protocols by the factor $3/2$. At the error budget of Sec.~\ref{sec:physical}, $\delta G=0.02$ for the four-state protocols and $0.03$ for $\prot63$, the thresholds are $v_c=0.859$, $0.845$ and $0.878$ and $\eta_c=0.801$, $0.787$ and $0.780$ for $\prot43$, $\prot44$ and $\prot63$, i.e.\ about $0.007$ above their values at $\delta G=0$ in visibility and $0.008$--$0.016$ in efficiency. A well-characterised source therefore costs little. Since $G$ must be justified by a trusted characterisation of the source rather than inferred from the observed data, the experimental requirement lies in the precision of that characterisation, in particular of the side channels, which enter through $n\sqrt\epsilon$. At $v=\eta=1$, key is certified up to $\delta G=0.53$, $0.76$ and $0.73$, where the curves of Fig.~\ref{fig:dG} end (Appendix~\ref{app:checks}), and the min-entropy bound vanishes at the same points. 

\emph{Loss is the main practical limitation.} The loss thresholds, $\eta_c\geq0.764$, correspond to at most about $1.2\,\mathrm{dB}$ of total loss, including detector inefficiency. As in DI QKD, this reflects the absence of any assumption on the receiver. Below $\eta_{\mathrm{cl}}$ the lossy behaviour admits a classical model (Table~\ref{tab:protocols}), and more generally, below a critical detection efficiency, the lossy measurements of an untrusted receiver become partially jointly measurable, which gives Eve a direct attack with classical side information~\cite{Masini2024,Lobo2026}. 

\emph{Other protocols.} We also computed the key rates of several variants (Appendix~\ref{app:checks}, Table~\ref{tab:others}). With a single tilted measurement instead of two, the three-state protocol only reaches $v_c=0.939$ and $\eta_c=0.906$. Likewise, BB84 supplemented by the single measurement $(Z+X)/\sqrt2$, keeping both BB84 key bases, gives only $v_c=0.919$ and $\eta_c=0.867$, against $0.838$ and $0.779$ for $\prot44$, which adds both tilted measurements. Robustness thus originates from the pair of tilted test measurements.

\subsection{Comparison with device-independent bounds}
\label{sec:comparison}

Each of our four-state and six-state protocols has an entanglement-based twin, in which Alice measures her half of a maximally entangled state in the bases whose eigenstates she prepares in the P\&M version, and Bob measures as before. For $\prot43$ the twin is the CHSH-based DI protocol of Ref.~\cite{PhysRevLett.98.230501}.

A DI attack on the twin induces a P\&M attack on the original protocol that respects the leakage bound. Let $\rho_{(\theta,a)}$ be the state of Bob's and Eve's systems conditioned on Alice's setting $\theta$ and outcome $a$. The observed statistics fix Alice's outcomes to be uniform, so $\tfrac12\sum_a\rho_{(\theta,a)}=\rho_{BE}$ for every $\theta$, and $\sigma=2\rho_{BE}$ satisfies Eq.~\eqref{eq:sigma} with $\Tr\sigma=2=G$. The induced P\&M attack thus reproduces the same statistics and the same key state at the honest $G$. Consequently, the device-independent key rate of the twin, computed from the same full statistics, is an upper bound on any rate that can be certified from those statistics under the leakage bound.

Section~\ref{sec:ideal} provides the converse at the ideal point. When a key basis is observed exactly, every attack under the honest bound is an attack on a virtual entanglement-based protocol with the same correlators and the same key state. When the key basis is not observed exactly, this reduction is not available. The leakage bound then allows a larger set of attacks than the twin, since it only requires $\sigma\geq\rho_x$ for each preparation, whereas an attack on the twin satisfies $\sigma=\rho_{(\theta,0)}+\rho_{(\theta,1)}$ for every basis $\theta$. We addressed numerically whether this additional freedom helps Eve.

We computed the DI rates of the twins with the same entropy bound, decoupled over the same nodes, and an NPA relaxation~\cite{Navascues2007,Navascues2008} (Appendix~\ref{app:checks}). They are shown as open circles in Fig.~\ref{fig:keyrates}, and their thresholds in the last column of Table~\ref{tab:thresholds}. For $\prot43$, the DI and LB thresholds coincide to within $5\cdot10^{-4}$, and the two rates agree to within $2\cdot10^{-5}$ bits along the whole visibility curve, and to within $10^{-3}$ bits along the efficiency curve. For $\prot44$, the DI thresholds, $v_c=0.837$ and $\eta_c=0.778$, lie $1.5\cdot10^{-3}$ and $10^{-3}$ below ours. Adding the products $M_{0|y}M_{0|y'}M_{0|y''}$ to our relaxation gives $v_c=0.8373$ and $\eta_c=0.7783$, which reduces the gaps to $5\cdot10^{-4}$ and $7\cdot10^{-4}$, so they mostly reflect the convergence of our relaxation. For $\prot63$, the visibility thresholds coincide, $v_c=0.870$, while the DI efficiency threshold, $\eta_c=0.758$, lies $0.006$ below ours. The products $M_{0|y}M_{0|y'}P^\theta_{a,k}$ already lowered our value by $0.019$, and we expect longer words to close part of this gap. Near the ideal point, the DI relaxations we used for $\prot44$ and $\prot63$ lie below our bounds, by up to $0.010$ and $0.009$ bits at $v=0.99$. A DI relaxation only lower-bounds the DI rate, so there it has not converged.

Thus, everywhere along the curves of Fig.~\ref{fig:keyrates}, for these protocols a bound of one bit on the information that leaves Alice's laboratory certifies, within the precision of our relaxations, as much key as a device-independent analysis of the entanglement-based twin (the largest remaining difference, $0.006$ in the efficiency threshold of $\prot63$, is at least partly due to the convergence of our relaxation). The coincidence of LB and DI bounds mirrors the observation of Ref.~\cite{Tavakoli2022} for randomness generation, where the one-bit leakage bound and the qubit bound certify the same randomness.

\subsection{Entropy as a function of the CHSH value}
\label{sec:chsh}

In the DI twin of $\prot43$, it is possible to certify entropy from the CHSH value alone, instead of the full statistics~\cite{PhysRevLett.98.230501,Pironio2009}:
\begin{equation}
    \label{eq:di_chsh}
    H(A|E)\geq1-h\Bigl(\frac12+\frac12\sqrt{(S/2)^2-1}\Bigr),\qquad S\geq2,
\end{equation}
where $h$ is the binary entropy. 

The same is true in the P\&M protocol under the leakage bound. We define
\begin{equation}
    \label{eq:FSG}
    F(S,G) = \inf\bigl\{H(A|E) : S_{\mathrm{CHSH}}\geq S,\ \Pgn\leq G\bigr\} ,
\end{equation}
where $S_{\mathrm{CHSH}}$ is the P\&M CHSH value of Eq.~\eqref{eq:chsh}, and where the infimum runs over all states, measurements and attacks that satisfy the two constraints. Bob's key measurement and all other statistics are left free. Since $\nu=1$ for $\prot43$, the basis register is trivial and $H(A|E\Theta)=H(A|E)$; we write the latter here. The key rate of $\prot43$ is then at least $F(S_{\mathrm{CHSH}},G)-H(A|B)$, which requires only the CHSH value and the statistics of the key rounds. We computed lower bounds on $F$ with the relaxation~\eqref{eq:swnpa} on the words of length at most four, in steps of $0.02$ in $S$ (Appendix~\ref{app:checks}). They are shown in Fig.~\ref{fig:chsh}.

\begin{figure*}[tb]
    \centering
    \includegraphics[width=\textwidth]{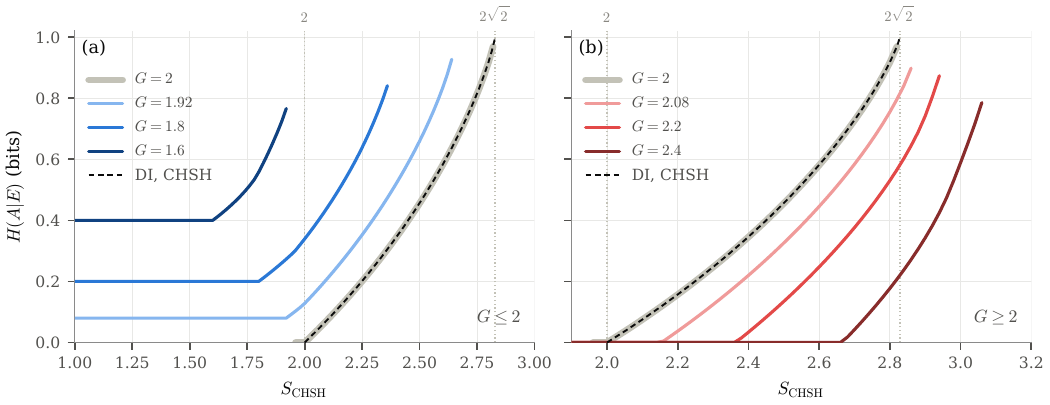}
    \caption{Lower bounds on $F(S,G)$, the conditional entropy  $H(A|E)$ certified by the CHSH value $S_{\mathrm{CHSH}}\geq S$ of Eq.~\eqref{eq:chsh} and the bound $\Pgn\leq G$, (a) for $G\leq2$ and (b) for $G\geq2$. The thick gray curve is $G=2$, and the dashed curve is the device-independent CHSH bound, Eq.~\eqref{eq:di_chsh}. Each curve ends at the largest value of $S$ compatible with $G$. The curves are computed with the relaxation~\eqref{eq:swnpa} on the words of length at most four and the $26$-node partition. Longer words and a finer partition change them by at most $1.4\cdot10^{-3}$ bits for $S\leq2\sqrt2$ (Appendix~\ref{app:checks}).}
    \label{fig:chsh}
\end{figure*}

At $G=2$ the computed bound coincides with the device-independent bound. For $2.1\leq S\leq2\sqrt2$, $F(S,2)$ lies below the right-hand side of Eq.~\eqref{eq:di_chsh} by at most $1.2\cdot10^{-3}$ bits with the $26$-node partition of Fig.~\ref{fig:chsh}, and by at most $5\cdot10^{-5}$ bits when the partition is refined to $120$ nodes (Appendix~\ref{app:checks}). Eve's guessing probability of the key bit likewise attains the device-independent value $\tfrac12+\tfrac12\sqrt{2-S^2/4}$~\cite{Pironio2010}, to within $5\cdot10^{-4}$. For the behaviours where the correlations in the key bases are perfect, this is a consequence of the reduction of Sec.~\ref{sec:ideal}, which shows that every attack under the leakage bound is an attack on a virtual entanglement-based protocol with the same correlators and the same key state. The numerics show that it remains true when the key basis is noisy. 

For $G\neq2$ the curves are not simple rescalings of Eq.~\eqref{eq:di_chsh}. Below $G=2$, entropy is certified even without a violation of the CHSH inequality, $S\leq2$. This is because the leakage bound alone limits Eve's knowledge of the key bit. Indeed, a guess of the key bit, together with the guess that the basis is $Z$, is a guess of $x$, so that $\Pg(A|E)\leq G/2<1$ whatever the statistics. In exchange, the largest CHSH value compatible with the bound decreases as $G$ is lowered, and falls below $2$ at $G=1.6$. Above $G=2$ the curves lie below the DI bound, and a CHSH violation no longer suffices, since the entropy vanishes for $S$ up to a threshold that grows with $G$ and, at $G=2.6$, already exceeds the value $2\sqrt2$ of the target behaviour.

Used as a key-rate formula for $\prot43$, $F$ loses little compared with the full statistics. At $\delta G=0.08$ it gives $v_c=0.883$ and $\eta_c=0.830$, against $0.879$ and $0.827$ in Table~\ref{tab:thresholds}.
\section{Conclusion}
\label{sec:conclusion}

We have shown that a bound on the information that the source emits about Alice's input is enough to prove security of suitably designed P\&M protocols. The leakage bound, a single bound on the probability of guessing Alice's input from everything that leaves her laboratory, yields security against collective attacks with an uncharacterised receiver and without any assumption on the dimension, energy or form of the emitted states. 

The best protocols we studied tolerate visibilities down to $v=0.84$ and detection efficiencies down to $\eta=0.76$, against $v=0.78$ for BB84 with a characterised receiver. Relative to device independence, little is lost as our bounds certify, within the precision of our relaxations, as much as the DI analysis along the whole curves, provided leakage is bounded by its honest value.

Several directions remain open. The analysis is asymptotic and restricted to collective attacks. The dual certificates of the semidefinite programs are the natural input to entropy accumulation~\cite{ArnonFriedman2018,Metger2024}, which in the DI setting extends security to coherent attacks and finite keys, but whether it applies here depends on how the leakage bound is formulated for a sequential protocol, where the emission of a round may depend on the past. Such a formulation is also the natural setting for sources with memory and correlations between rounds~\cite{Mizutani2019,Pereira2020}. On the experimental side, the entire analysis rests on the value of $G$, which calls for sources whose leakage is certified, for instance through the fidelity and Trojan-horse bounds of Refs.~\cite{Lucamarini2015,CurrasLorenzo2025}, and for optimised states and measurements. Finally, loss tolerance is the main obstacle to long-distance operation; combining the leakage bound with some structure on the measurement device, or with a measurement-device-independent architecture, may improve it.

\vspace{\baselineskip}\noindent\emph{Note added.} The results of this work were first presented at the Veriqtas workshop (talk, 25 March 2026) and at the QKD Summer School 2026 (poster, 16--21 August 2026). After this work was essentially completed, we became aware of Ref.~\cite{Chaturvedi2026}, which independently analyses the $\prot43$ protocol under the same bound, $\Pg\leq1/2$, and under several other operational source assumptions, and obtains min-entropy and von Neumann entropy bounds with a P\&M version of the method of Ref.~\cite{brown2024device}, using moment matrices weighted by the preparations as in Eq.~\eqref{eq:swnpa}. At $G=2$, Ref.~\cite{Chaturvedi2026} reports $v_c=0.8895$ from the min-entropy and $v_c=0.874$ from the von Neumann entropy; our min-entropy threshold coincides with the former, and our von Neumann threshold, $v_c=0.852$, is lower and equals the device-independent value. The difference may come from the entropy bound as well as from the relaxation. 

\vspace{\baselineskip}\noindent\emph{Use of AI tools.} The core analytical and numerical results of this work were obtained without the use of AI tools. AI tools were used in later stages to improve the manuscript, to write code for additional numerical examples and for the figures, and to explore how the leakage bound can be certified in practice (Sec.~\ref{sec:imperfect}). The authors have independently verified the numerical code and take full responsibility for all claims made in this work.

\vspace{\baselineskip}\noindent\emph{Acknowledgements.}
S.P., A.M. and M.B.J. acknowledge funding from the ToDiQT project within the QuantERA III Programme that has received funding from the European Union’s Horizon  research and innovation programme and the F.R.S-FNRS Pint-Multi programme under Grant Agreement R.8011.26, from the European Union’s Horizon Europe research and innovation programme under the project "Quantum Security Networks Partnership" (QSNP, grant agreement No 101114043), 
from the F.R.S-FNRS through the PDR T.0171.22,
from the FWO and F.R.S.-FNRS under the Excellence of Science (EOS) programme project 40007526,
from the FWO through the BeQuNet SBO project S008323N. A.T. is supported by the Swedish Research Council under Contract No. 2023-03498, the Swedish Foundation for Strategic Research, and the Knut and Alice Wallenberg Foundation through the Wallenberg Center for Quantum Technology (WACQT).

S.P. is a Research Director of the Fonds de la Recherche Scientifique – FNRS.

Funded by the European Union. Views and opinions expressed are however those of the authors only and do not necessarily reflect those of the European Union. Neither the European Union nor the granting authority can be held responsible for them.

\vspace{\baselineskip}\noindent\emph{Data availability}---The code and data that support the findings of this article are openly available~\cite{RIQKD} which uses PCPOP.jl~\cite{BermejoMoran2026PCPOP} in backend for building and solving the SDP relaxations.

\bibliography{Refs}

\begin{thebibliography}{81}%
\makeatletter
\providecommand \@ifxundefined [1]{%
 \@ifx{#1\undefined}
}%
\providecommand \@ifnum [1]{%
 \ifnum #1\expandafter \@firstoftwo
 \else \expandafter \@secondoftwo
 \fi
}%
\providecommand \@ifx [1]{%
 \ifx #1\expandafter \@firstoftwo
 \else \expandafter \@secondoftwo
 \fi
}%
\providecommand \natexlab [1]{#1}%
\providecommand \enquote  [1]{``#1''}%
\providecommand \bibnamefont  [1]{#1}%
\providecommand \bibfnamefont [1]{#1}%
\providecommand \citenamefont [1]{#1}%
\providecommand \href@noop [0]{\@secondoftwo}%
\providecommand \href [0]{\begingroup \@sanitize@url \@href}%
\providecommand \@href[1]{\@@startlink{#1}\@@href}%
\providecommand \@@href[1]{\endgroup#1\@@endlink}%
\providecommand \@sanitize@url [0]{\catcode `\\12\catcode `\$12\catcode `\&12\catcode `\#12\catcode `\^12\catcode `\_12\catcode `\%12\relax}%
\providecommand \@@startlink[1]{}%
\providecommand \@@endlink[0]{}%
\providecommand \url  [0]{\begingroup\@sanitize@url \@url }%
\providecommand \@url [1]{\endgroup\@href {#1}{\urlprefix }}%
\providecommand \urlprefix  [0]{URL }%
\providecommand \Eprint [0]{\href }%
\providecommand \doibase [0]{https://doi.org/}%
\providecommand \selectlanguage [0]{\@gobble}%
\providecommand \bibinfo  [0]{\@secondoftwo}%
\providecommand \bibfield  [0]{\@secondoftwo}%
\providecommand \translation [1]{[#1]}%
\providecommand \BibitemOpen [0]{}%
\providecommand \bibitemStop [0]{}%
\providecommand \bibitemNoStop [0]{.\EOS\space}%
\providecommand \EOS [0]{\spacefactor3000\relax}%
\providecommand \BibitemShut  [1]{\csname bibitem#1\endcsname}%
\let\auto@bib@innerbib\@empty
\bibitem [{\citenamefont {Bennett}\ and\ \citenamefont {Brassard}(1984)}]{BB84}%
  \BibitemOpen
  \bibfield  {author} {\bibinfo {author} {\bibfnamefont {C.~H.}\ \bibnamefont {Bennett}}\ and\ \bibinfo {author} {\bibfnamefont {G.}~\bibnamefont {Brassard}},\ }in\ \href@noop {} {\emph {\bibinfo {booktitle} {Proceedings of IEEE International Conference on Computers, Systems and Signal Processing}}}\ (\bibinfo {address} {Bangalore, India},\ \bibinfo {year} {1984})\ pp.\ \bibinfo {pages} {175--179}\BibitemShut {NoStop}%
\bibitem [{\citenamefont {Ekert}(1991)}]{Ekert91}%
  \BibitemOpen
  \bibfield  {author} {\bibinfo {author} {\bibfnamefont {A.~K.}\ \bibnamefont {Ekert}},\ }\href {https://doi.org/10.1103/PhysRevLett.67.661} {\bibfield  {journal} {\bibinfo  {journal} {Phys. Rev. Lett.}\ }\textbf {\bibinfo {volume} {67}},\ \bibinfo {pages} {661} (\bibinfo {year} {1991})}\BibitemShut {NoStop}%
\bibitem [{\citenamefont {Brassard}\ \emph {et~al.}(2000)\citenamefont {Brassard}, \citenamefont {L\"utkenhaus}, \citenamefont {Mor},\ and\ \citenamefont {Sanders}}]{PhysRevLett.85.1330}%
  \BibitemOpen
  \bibfield  {author} {\bibinfo {author} {\bibfnamefont {G.}~\bibnamefont {Brassard}}, \bibinfo {author} {\bibfnamefont {N.}~\bibnamefont {L\"utkenhaus}}, \bibinfo {author} {\bibfnamefont {T.}~\bibnamefont {Mor}},\ and\ \bibinfo {author} {\bibfnamefont {B.~C.}\ \bibnamefont {Sanders}},\ }\href {https://doi.org/10.1103/PhysRevLett.85.1330} {\bibfield  {journal} {\bibinfo  {journal} {Phys. Rev. Lett.}\ }\textbf {\bibinfo {volume} {85}},\ \bibinfo {pages} {1330} (\bibinfo {year} {2000})}\BibitemShut {NoStop}%
\bibitem [{\citenamefont {Lydersen}\ \emph {et~al.}(2010)\citenamefont {Lydersen}, \citenamefont {Wiechers}, \citenamefont {Wittmann}, \citenamefont {Elser}, \citenamefont {Skaar},\ and\ \citenamefont {Makarov}}]{Lydersen2010}%
  \BibitemOpen
  \bibfield  {author} {\bibinfo {author} {\bibfnamefont {L.}~\bibnamefont {Lydersen}}, \bibinfo {author} {\bibfnamefont {C.}~\bibnamefont {Wiechers}}, \bibinfo {author} {\bibfnamefont {C.}~\bibnamefont {Wittmann}}, \bibinfo {author} {\bibfnamefont {D.}~\bibnamefont {Elser}}, \bibinfo {author} {\bibfnamefont {J.}~\bibnamefont {Skaar}},\ and\ \bibinfo {author} {\bibfnamefont {V.}~\bibnamefont {Makarov}},\ }\href {https://doi.org/10.1038/nphoton.2010.214} {\bibfield  {journal} {\bibinfo  {journal} {Nature Photonics}\ }\textbf {\bibinfo {volume} {4}},\ \bibinfo {pages} {686} (\bibinfo {year} {2010})}\BibitemShut {NoStop}%
\bibitem [{\citenamefont {Xu}\ \emph {et~al.}(2020)\citenamefont {Xu}, \citenamefont {Ma}, \citenamefont {Zhang}, \citenamefont {Lo},\ and\ \citenamefont {Pan}}]{Xu2020}%
  \BibitemOpen
  \bibfield  {author} {\bibinfo {author} {\bibfnamefont {F.}~\bibnamefont {Xu}}, \bibinfo {author} {\bibfnamefont {X.}~\bibnamefont {Ma}}, \bibinfo {author} {\bibfnamefont {Q.}~\bibnamefont {Zhang}}, \bibinfo {author} {\bibfnamefont {H.-K.}\ \bibnamefont {Lo}},\ and\ \bibinfo {author} {\bibfnamefont {J.-W.}\ \bibnamefont {Pan}},\ }\href {https://doi.org/10.1103/RevModPhys.92.025002} {\bibfield  {journal} {\bibinfo  {journal} {Rev. Mod. Phys.}\ }\textbf {\bibinfo {volume} {92}},\ \bibinfo {pages} {025002} (\bibinfo {year} {2020})}\BibitemShut {NoStop}%
\bibitem [{\citenamefont {Zapatero}\ \emph {et~al.}(2025)\citenamefont {Zapatero}, \citenamefont {Navarrete},\ and\ \citenamefont {Curty}}]{Zapatero2025}%
  \BibitemOpen
  \bibfield  {author} {\bibinfo {author} {\bibfnamefont {V.}~\bibnamefont {Zapatero}}, \bibinfo {author} {\bibfnamefont {{\'A}.}~\bibnamefont {Navarrete}},\ and\ \bibinfo {author} {\bibfnamefont {M.}~\bibnamefont {Curty}},\ }\href {https://doi.org/10.1002/qute.202300380} {\bibfield  {journal} {\bibinfo  {journal} {Adv. Quantum Technol.}\ }\textbf {\bibinfo {volume} {8}},\ \bibinfo {pages} {2300380} (\bibinfo {year} {2025})}\BibitemShut {NoStop}%
\bibitem [{\citenamefont {Gottesman}\ \emph {et~al.}(2004)\citenamefont {Gottesman}, \citenamefont {Lo}, \citenamefont {Lütkenhaus},\ and\ \citenamefont {Preskill}}]{gottesman2004securityquantumkeydistribution}%
  \BibitemOpen
  \bibfield  {author} {\bibinfo {author} {\bibfnamefont {D.}~\bibnamefont {Gottesman}}, \bibinfo {author} {\bibfnamefont {H.-K.}\ \bibnamefont {Lo}}, \bibinfo {author} {\bibfnamefont {N.}~\bibnamefont {Lütkenhaus}},\ and\ \bibinfo {author} {\bibfnamefont {J.}~\bibnamefont {Preskill}},\ }\href {https://arxiv.org/abs/quant-ph/0212066} {\bibinfo {title} {Security of quantum key distribution with imperfect devices}} (\bibinfo {year} {2004}),\ \Eprint {https://arxiv.org/abs/quant-ph/0212066} {arXiv:quant-ph/0212066 [quant-ph]} \BibitemShut {NoStop}%
\bibitem [{\citenamefont {Koashi}\ and\ \citenamefont {Preskill}(2003)}]{KoashiPreskill2003}%
  \BibitemOpen
  \bibfield  {author} {\bibinfo {author} {\bibfnamefont {M.}~\bibnamefont {Koashi}}\ and\ \bibinfo {author} {\bibfnamefont {J.}~\bibnamefont {Preskill}},\ }\href {https://doi.org/10.1103/PhysRevLett.90.057902} {\bibfield  {journal} {\bibinfo  {journal} {Phys. Rev. Lett.}\ }\textbf {\bibinfo {volume} {90}},\ \bibinfo {pages} {057902} (\bibinfo {year} {2003})}\BibitemShut {NoStop}%
\bibitem [{\citenamefont {Lo}\ and\ \citenamefont {Preskill}(2007)}]{LoPreskill2007}%
  \BibitemOpen
  \bibfield  {author} {\bibinfo {author} {\bibfnamefont {H.-K.}\ \bibnamefont {Lo}}\ and\ \bibinfo {author} {\bibfnamefont {J.}~\bibnamefont {Preskill}},\ }\href@noop {} {\bibfield  {journal} {\bibinfo  {journal} {Quantum Inf. Comput.}\ }\textbf {\bibinfo {volume} {7}},\ \bibinfo {pages} {431} (\bibinfo {year} {2007})}\BibitemShut {NoStop}%
\bibitem [{\citenamefont {Mar{\o}y}\ \emph {et~al.}(2010)\citenamefont {Mar{\o}y}, \citenamefont {Lydersen},\ and\ \citenamefont {Skaar}}]{Maroy2010}%
  \BibitemOpen
  \bibfield  {author} {\bibinfo {author} {\bibfnamefont {{\O}.}~\bibnamefont {Mar{\o}y}}, \bibinfo {author} {\bibfnamefont {L.}~\bibnamefont {Lydersen}},\ and\ \bibinfo {author} {\bibfnamefont {J.}~\bibnamefont {Skaar}},\ }\href {https://doi.org/10.1103/PhysRevA.82.032337} {\bibfield  {journal} {\bibinfo  {journal} {Phys. Rev. A}\ }\textbf {\bibinfo {volume} {82}},\ \bibinfo {pages} {032337} (\bibinfo {year} {2010})}\BibitemShut {NoStop}%
\bibitem [{\citenamefont {Tamaki}\ \emph {et~al.}(2014)\citenamefont {Tamaki}, \citenamefont {Curty}, \citenamefont {Kato}, \citenamefont {Lo},\ and\ \citenamefont {Azuma}}]{Tamaki2014}%
  \BibitemOpen
  \bibfield  {author} {\bibinfo {author} {\bibfnamefont {K.}~\bibnamefont {Tamaki}}, \bibinfo {author} {\bibfnamefont {M.}~\bibnamefont {Curty}}, \bibinfo {author} {\bibfnamefont {G.}~\bibnamefont {Kato}}, \bibinfo {author} {\bibfnamefont {H.-K.}\ \bibnamefont {Lo}},\ and\ \bibinfo {author} {\bibfnamefont {K.}~\bibnamefont {Azuma}},\ }\href {https://doi.org/10.1103/PhysRevA.90.052314} {\bibfield  {journal} {\bibinfo  {journal} {Phys. Rev. A}\ }\textbf {\bibinfo {volume} {90}},\ \bibinfo {pages} {052314} (\bibinfo {year} {2014})}\BibitemShut {NoStop}%
\bibitem [{\citenamefont {Lucamarini}\ \emph {et~al.}(2015)\citenamefont {Lucamarini}, \citenamefont {Choi}, \citenamefont {Ward}, \citenamefont {Dynes}, \citenamefont {Yuan},\ and\ \citenamefont {Shields}}]{Lucamarini2015}%
  \BibitemOpen
  \bibfield  {author} {\bibinfo {author} {\bibfnamefont {M.}~\bibnamefont {Lucamarini}}, \bibinfo {author} {\bibfnamefont {I.}~\bibnamefont {Choi}}, \bibinfo {author} {\bibfnamefont {M.~B.}\ \bibnamefont {Ward}}, \bibinfo {author} {\bibfnamefont {J.~F.}\ \bibnamefont {Dynes}}, \bibinfo {author} {\bibfnamefont {Z.~L.}\ \bibnamefont {Yuan}},\ and\ \bibinfo {author} {\bibfnamefont {A.~J.}\ \bibnamefont {Shields}},\ }\href {https://doi.org/10.1103/PhysRevX.5.031030} {\bibfield  {journal} {\bibinfo  {journal} {Phys. Rev. X}\ }\textbf {\bibinfo {volume} {5}},\ \bibinfo {pages} {031030} (\bibinfo {year} {2015})}\BibitemShut {NoStop}%
\bibitem [{\citenamefont {Pereira}\ \emph {et~al.}(2019)\citenamefont {Pereira}, \citenamefont {Curty},\ and\ \citenamefont {Tamaki}}]{Pereira2019}%
  \BibitemOpen
  \bibfield  {author} {\bibinfo {author} {\bibfnamefont {M.}~\bibnamefont {Pereira}}, \bibinfo {author} {\bibfnamefont {M.}~\bibnamefont {Curty}},\ and\ \bibinfo {author} {\bibfnamefont {K.}~\bibnamefont {Tamaki}},\ }\href {https://doi.org/10.1038/s41534-019-0180-9} {\bibfield  {journal} {\bibinfo  {journal} {npj Quantum Inf.}\ }\textbf {\bibinfo {volume} {5}},\ \bibinfo {pages} {62} (\bibinfo {year} {2019})}\BibitemShut {NoStop}%
\bibitem [{\citenamefont {Pereira}\ \emph {et~al.}(2020)\citenamefont {Pereira}, \citenamefont {Kato}, \citenamefont {Mizutani}, \citenamefont {Curty},\ and\ \citenamefont {Tamaki}}]{Pereira2020}%
  \BibitemOpen
  \bibfield  {author} {\bibinfo {author} {\bibfnamefont {M.}~\bibnamefont {Pereira}}, \bibinfo {author} {\bibfnamefont {G.}~\bibnamefont {Kato}}, \bibinfo {author} {\bibfnamefont {A.}~\bibnamefont {Mizutani}}, \bibinfo {author} {\bibfnamefont {M.}~\bibnamefont {Curty}},\ and\ \bibinfo {author} {\bibfnamefont {K.}~\bibnamefont {Tamaki}},\ }\href {https://doi.org/10.1126/sciadv.aaz4487} {\bibfield  {journal} {\bibinfo  {journal} {Sci. Adv.}\ }\textbf {\bibinfo {volume} {6}},\ \bibinfo {pages} {eaaz4487} (\bibinfo {year} {2020})}\BibitemShut {NoStop}%
\bibitem [{\citenamefont {Zhang}\ \emph {et~al.}(2021)\citenamefont {Zhang}, \citenamefont {Coles}, \citenamefont {Winick}, \citenamefont {Lin},\ and\ \citenamefont {L{\"u}tkenhaus}}]{Zhang2021}%
  \BibitemOpen
  \bibfield  {author} {\bibinfo {author} {\bibfnamefont {Y.}~\bibnamefont {Zhang}}, \bibinfo {author} {\bibfnamefont {P.~J.}\ \bibnamefont {Coles}}, \bibinfo {author} {\bibfnamefont {A.}~\bibnamefont {Winick}}, \bibinfo {author} {\bibfnamefont {J.}~\bibnamefont {Lin}},\ and\ \bibinfo {author} {\bibfnamefont {N.}~\bibnamefont {L{\"u}tkenhaus}},\ }\href {https://doi.org/10.1103/PhysRevResearch.3.013076} {\bibfield  {journal} {\bibinfo  {journal} {Phys. Rev. Research}\ }\textbf {\bibinfo {volume} {3}},\ \bibinfo {pages} {013076} (\bibinfo {year} {2021})}\BibitemShut {NoStop}%
\bibitem [{\citenamefont {Nahar}\ \emph {et~al.}(2023)\citenamefont {Nahar}, \citenamefont {Upadhyaya},\ and\ \citenamefont {L{\"u}tkenhaus}}]{Nahar2023}%
  \BibitemOpen
  \bibfield  {author} {\bibinfo {author} {\bibfnamefont {S.}~\bibnamefont {Nahar}}, \bibinfo {author} {\bibfnamefont {T.}~\bibnamefont {Upadhyaya}},\ and\ \bibinfo {author} {\bibfnamefont {N.}~\bibnamefont {L{\"u}tkenhaus}},\ }\href {https://doi.org/10.1103/PhysRevApplied.20.064031} {\bibfield  {journal} {\bibinfo  {journal} {Phys. Rev. Applied}\ }\textbf {\bibinfo {volume} {20}},\ \bibinfo {pages} {064031} (\bibinfo {year} {2023})}\BibitemShut {NoStop}%
\bibitem [{\citenamefont {Curr{\'a}s-Lorenzo}\ \emph {et~al.}(2025{\natexlab{a}})\citenamefont {Curr{\'a}s-Lorenzo}, \citenamefont {Pereira}, \citenamefont {Kato}, \citenamefont {Curty},\ and\ \citenamefont {Tamaki}}]{CurrasLorenzo2025}%
  \BibitemOpen
  \bibfield  {author} {\bibinfo {author} {\bibfnamefont {G.}~\bibnamefont {Curr{\'a}s-Lorenzo}}, \bibinfo {author} {\bibfnamefont {M.}~\bibnamefont {Pereira}}, \bibinfo {author} {\bibfnamefont {G.}~\bibnamefont {Kato}}, \bibinfo {author} {\bibfnamefont {M.}~\bibnamefont {Curty}},\ and\ \bibinfo {author} {\bibfnamefont {K.}~\bibnamefont {Tamaki}},\ }\href {https://doi.org/10.1364/OPTICAQ.569424} {\bibfield  {journal} {\bibinfo  {journal} {Optica Quantum}\ }\textbf {\bibinfo {volume} {3}},\ \bibinfo {pages} {525} (\bibinfo {year} {2025}{\natexlab{a}})}\BibitemShut {NoStop}%
\bibitem [{\citenamefont {Ac\'{\i}n}\ \emph {et~al.}(2007)\citenamefont {Ac\'{\i}n}, \citenamefont {Brunner}, \citenamefont {Gisin}, \citenamefont {Massar}, \citenamefont {Pironio},\ and\ \citenamefont {Scarani}}]{PhysRevLett.98.230501}%
  \BibitemOpen
  \bibfield  {author} {\bibinfo {author} {\bibfnamefont {A.}~\bibnamefont {Ac\'{\i}n}}, \bibinfo {author} {\bibfnamefont {N.}~\bibnamefont {Brunner}}, \bibinfo {author} {\bibfnamefont {N.}~\bibnamefont {Gisin}}, \bibinfo {author} {\bibfnamefont {S.}~\bibnamefont {Massar}}, \bibinfo {author} {\bibfnamefont {S.}~\bibnamefont {Pironio}},\ and\ \bibinfo {author} {\bibfnamefont {V.}~\bibnamefont {Scarani}},\ }\href {https://doi.org/10.1103/PhysRevLett.98.230501} {\bibfield  {journal} {\bibinfo  {journal} {Phys. Rev. Lett.}\ }\textbf {\bibinfo {volume} {98}},\ \bibinfo {pages} {230501} (\bibinfo {year} {2007})}\BibitemShut {NoStop}%
\bibitem [{\citenamefont {Pironio}\ \emph {et~al.}(2009)\citenamefont {Pironio}, \citenamefont {Ac\'{\i}n}, \citenamefont {Brunner}, \citenamefont {Gisin}, \citenamefont {Massar},\ and\ \citenamefont {Scarani}}]{Pironio2009}%
  \BibitemOpen
  \bibfield  {author} {\bibinfo {author} {\bibfnamefont {S.}~\bibnamefont {Pironio}}, \bibinfo {author} {\bibfnamefont {A.}~\bibnamefont {Ac\'{\i}n}}, \bibinfo {author} {\bibfnamefont {N.}~\bibnamefont {Brunner}}, \bibinfo {author} {\bibfnamefont {N.}~\bibnamefont {Gisin}}, \bibinfo {author} {\bibfnamefont {S.}~\bibnamefont {Massar}},\ and\ \bibinfo {author} {\bibfnamefont {V.}~\bibnamefont {Scarani}},\ }\href {https://doi.org/10.1088/1367-2630/11/4/045021} {\bibfield  {journal} {\bibinfo  {journal} {New J. Phys.}\ }\textbf {\bibinfo {volume} {11}},\ \bibinfo {pages} {045021} (\bibinfo {year} {2009})}\BibitemShut {NoStop}%
\bibitem [{\citenamefont {Vazirani}\ and\ \citenamefont {Vidick}(2014)}]{PhysRevLett.113.140501}%
  \BibitemOpen
  \bibfield  {author} {\bibinfo {author} {\bibfnamefont {U.}~\bibnamefont {Vazirani}}\ and\ \bibinfo {author} {\bibfnamefont {T.}~\bibnamefont {Vidick}},\ }\href {https://doi.org/10.1103/PhysRevLett.113.140501} {\bibfield  {journal} {\bibinfo  {journal} {Phys. Rev. Lett.}\ }\textbf {\bibinfo {volume} {113}},\ \bibinfo {pages} {140501} (\bibinfo {year} {2014})}\BibitemShut {NoStop}%
\bibitem [{\citenamefont {Arnon-Friedman}\ \emph {et~al.}(2018)\citenamefont {Arnon-Friedman}, \citenamefont {Dupuis}, \citenamefont {Fawzi}, \citenamefont {Renner},\ and\ \citenamefont {Vidick}}]{ArnonFriedman2018}%
  \BibitemOpen
  \bibfield  {author} {\bibinfo {author} {\bibfnamefont {R.}~\bibnamefont {Arnon-Friedman}}, \bibinfo {author} {\bibfnamefont {F.}~\bibnamefont {Dupuis}}, \bibinfo {author} {\bibfnamefont {O.}~\bibnamefont {Fawzi}}, \bibinfo {author} {\bibfnamefont {R.}~\bibnamefont {Renner}},\ and\ \bibinfo {author} {\bibfnamefont {T.}~\bibnamefont {Vidick}},\ }\href {https://doi.org/10.1038/s41467-017-02307-4} {\bibfield  {journal} {\bibinfo  {journal} {Nat. Commun.}\ }\textbf {\bibinfo {volume} {9}},\ \bibinfo {pages} {459} (\bibinfo {year} {2018})}\BibitemShut {NoStop}%
\bibitem [{\citenamefont {Zapatero}\ \emph {et~al.}(2023)\citenamefont {Zapatero}, \citenamefont {van Leent}, \citenamefont {Arnon-Friedman}, \citenamefont {Liu}, \citenamefont {Zhang}, \citenamefont {Weinfurter},\ and\ \citenamefont {Curty}}]{Zapatero2023}%
  \BibitemOpen
  \bibfield  {author} {\bibinfo {author} {\bibfnamefont {V.}~\bibnamefont {Zapatero}}, \bibinfo {author} {\bibfnamefont {T.}~\bibnamefont {van Leent}}, \bibinfo {author} {\bibfnamefont {R.}~\bibnamefont {Arnon-Friedman}}, \bibinfo {author} {\bibfnamefont {W.-Z.}\ \bibnamefont {Liu}}, \bibinfo {author} {\bibfnamefont {Q.}~\bibnamefont {Zhang}}, \bibinfo {author} {\bibfnamefont {H.}~\bibnamefont {Weinfurter}},\ and\ \bibinfo {author} {\bibfnamefont {M.}~\bibnamefont {Curty}},\ }\href {https://doi.org/10.1038/s41534-023-00684-x} {\bibfield  {journal} {\bibinfo  {journal} {npj Quantum Information}\ }\textbf {\bibinfo {volume} {9}},\ \bibinfo {pages} {10} (\bibinfo {year} {2023})}\BibitemShut {NoStop}%
\bibitem [{\citenamefont {Paw\l{}owski}\ and\ \citenamefont {Brunner}(2011)}]{onewaysdiqkd}%
  \BibitemOpen
  \bibfield  {author} {\bibinfo {author} {\bibfnamefont {M.}~\bibnamefont {Paw\l{}owski}}\ and\ \bibinfo {author} {\bibfnamefont {N.}~\bibnamefont {Brunner}},\ }\href {https://doi.org/10.1103/PhysRevA.84.010302} {\bibfield  {journal} {\bibinfo  {journal} {Phys. Rev. A}\ }\textbf {\bibinfo {volume} {84}},\ \bibinfo {pages} {010302} (\bibinfo {year} {2011})}\BibitemShut {NoStop}%
\bibitem [{\citenamefont {Van~Himbeeck}\ \emph {et~al.}(2017)\citenamefont {Van~Himbeeck}, \citenamefont {Woodhead}, \citenamefont {Cerf}, \citenamefont {García-Patrón},\ and\ \citenamefont {Pironio}}]{Vanhimbeeck}%
  \BibitemOpen
  \bibfield  {author} {\bibinfo {author} {\bibfnamefont {T.}~\bibnamefont {Van~Himbeeck}}, \bibinfo {author} {\bibfnamefont {E.}~\bibnamefont {Woodhead}}, \bibinfo {author} {\bibfnamefont {N.~J.}\ \bibnamefont {Cerf}}, \bibinfo {author} {\bibfnamefont {R.}~\bibnamefont {García-Patrón}},\ and\ \bibinfo {author} {\bibfnamefont {S.}~\bibnamefont {Pironio}},\ }\href {https://doi.org/10.22331/q-2017-11-18-33} {\bibfield  {journal} {\bibinfo  {journal} {Quantum}\ }\textbf {\bibinfo {volume} {1}},\ \bibinfo {pages} {33} (\bibinfo {year} {2017})}\BibitemShut {NoStop}%
\bibitem [{\citenamefont {Ioannou}\ \emph {et~al.}(2022)\citenamefont {Ioannou}, \citenamefont {Sekatski}, \citenamefont {Abbott}, \citenamefont {Rosset}, \citenamefont {Bancal},\ and\ \citenamefont {Brunner}}]{Ioannou_2022}%
  \BibitemOpen
  \bibfield  {author} {\bibinfo {author} {\bibfnamefont {M.}~\bibnamefont {Ioannou}}, \bibinfo {author} {\bibfnamefont {P.}~\bibnamefont {Sekatski}}, \bibinfo {author} {\bibfnamefont {A.~A.}\ \bibnamefont {Abbott}}, \bibinfo {author} {\bibfnamefont {D.}~\bibnamefont {Rosset}}, \bibinfo {author} {\bibfnamefont {J.-D.}\ \bibnamefont {Bancal}},\ and\ \bibinfo {author} {\bibfnamefont {N.}~\bibnamefont {Brunner}},\ }\href {https://doi.org/10.1088/1367-2630/ac71bc} {\bibfield  {journal} {\bibinfo  {journal} {New Journal of Physics}\ }\textbf {\bibinfo {volume} {24}},\ \bibinfo {pages} {063006} (\bibinfo {year} {2022})}\BibitemShut {NoStop}%
\bibitem [{\citenamefont {Barrett}\ \emph {et~al.}(2013)\citenamefont {Barrett}, \citenamefont {Colbeck},\ and\ \citenamefont {Kent}}]{BarrettColbeckKent2013}%
  \BibitemOpen
  \bibfield  {author} {\bibinfo {author} {\bibfnamefont {J.}~\bibnamefont {Barrett}}, \bibinfo {author} {\bibfnamefont {R.}~\bibnamefont {Colbeck}},\ and\ \bibinfo {author} {\bibfnamefont {A.}~\bibnamefont {Kent}},\ }\href {https://doi.org/10.1103/PhysRevLett.110.010503} {\bibfield  {journal} {\bibinfo  {journal} {Phys. Rev. Lett.}\ }\textbf {\bibinfo {volume} {110}},\ \bibinfo {pages} {010503} (\bibinfo {year} {2013})}\BibitemShut {NoStop}%
\bibitem [{\citenamefont {Curty}\ and\ \citenamefont {Lo}(2019)}]{CurtyLo2019}%
  \BibitemOpen
  \bibfield  {author} {\bibinfo {author} {\bibfnamefont {M.}~\bibnamefont {Curty}}\ and\ \bibinfo {author} {\bibfnamefont {H.-K.}\ \bibnamefont {Lo}},\ }\href {https://doi.org/10.1038/s41534-019-0131-5} {\bibfield  {journal} {\bibinfo  {journal} {npj Quantum Inf.}\ }\textbf {\bibinfo {volume} {5}},\ \bibinfo {pages} {14} (\bibinfo {year} {2019})}\BibitemShut {NoStop}%
\bibitem [{\citenamefont {Tan}(2023)}]{Tan2023}%
  \BibitemOpen
  \bibfield  {author} {\bibinfo {author} {\bibfnamefont {E.~Y.-Z.}\ \bibnamefont {Tan}},\ }\href@noop {} {\bibinfo {title} {Robustness of implemented device-independent protocols against constrained leakage}} (\bibinfo {year} {2023}),\ \Eprint {https://arxiv.org/abs/2302.13928} {arXiv:2302.13928 [quant-ph]} \BibitemShut {NoStop}%
\bibitem [{\citenamefont {Zapatero}\ and\ \citenamefont {Curty}(2026)}]{ZapateroCurty2026}%
  \BibitemOpen
  \bibfield  {author} {\bibinfo {author} {\bibfnamefont {V.}~\bibnamefont {Zapatero}}\ and\ \bibinfo {author} {\bibfnamefont {M.}~\bibnamefont {Curty}},\ }\href@noop {} {\bibinfo {title} {Device-independent quantum cryptography with input leakage}} (\bibinfo {year} {2026}),\ \Eprint {https://arxiv.org/abs/2604.20573} {arXiv:2604.20573 [quant-ph]} \BibitemShut {NoStop}%
\bibitem [{\citenamefont {Tavakoli}\ \emph {et~al.}(2020)\citenamefont {Tavakoli}, \citenamefont {Zambrini~Cruzeiro}, \citenamefont {Brask}, \citenamefont {Gisin},\ and\ \citenamefont {Brunner}}]{Tavakoli2020}%
  \BibitemOpen
  \bibfield  {author} {\bibinfo {author} {\bibfnamefont {A.}~\bibnamefont {Tavakoli}}, \bibinfo {author} {\bibfnamefont {E.}~\bibnamefont {Zambrini~Cruzeiro}}, \bibinfo {author} {\bibfnamefont {J.~B.}\ \bibnamefont {Brask}}, \bibinfo {author} {\bibfnamefont {N.}~\bibnamefont {Gisin}},\ and\ \bibinfo {author} {\bibfnamefont {N.}~\bibnamefont {Brunner}},\ }\href {https://doi.org/10.22331/q-2020-09-24-332} {\bibfield  {journal} {\bibinfo  {journal} {Quantum}\ }\textbf {\bibinfo {volume} {4}},\ \bibinfo {pages} {332} (\bibinfo {year} {2020})}\BibitemShut {NoStop}%
\bibitem [{\citenamefont {Tavakoli}\ \emph {et~al.}(2022)\citenamefont {Tavakoli}, \citenamefont {Zambrini~Cruzeiro}, \citenamefont {Woodhead},\ and\ \citenamefont {Pironio}}]{Tavakoli2022}%
  \BibitemOpen
  \bibfield  {author} {\bibinfo {author} {\bibfnamefont {A.}~\bibnamefont {Tavakoli}}, \bibinfo {author} {\bibfnamefont {E.}~\bibnamefont {Zambrini~Cruzeiro}}, \bibinfo {author} {\bibfnamefont {E.}~\bibnamefont {Woodhead}},\ and\ \bibinfo {author} {\bibfnamefont {S.}~\bibnamefont {Pironio}},\ }\href {https://doi.org/10.22331/q-2022-01-05-620} {\bibfield  {journal} {\bibinfo  {journal} {{Q}uantum}\ }\textbf {\bibinfo {volume} {6}},\ \bibinfo {pages} {620} (\bibinfo {year} {2022})}\BibitemShut {NoStop}%
\bibitem [{\citenamefont {Acín}\ \emph {et~al.}(2006)\citenamefont {Acín}, \citenamefont {Gisin},\ and\ \citenamefont {Masanes}}]{acin2006from}%
  \BibitemOpen
  \bibfield  {author} {\bibinfo {author} {\bibfnamefont {A.}~\bibnamefont {Acín}}, \bibinfo {author} {\bibfnamefont {N.}~\bibnamefont {Gisin}},\ and\ \bibinfo {author} {\bibfnamefont {L.}~\bibnamefont {Masanes}},\ }\bibfield  {journal} {\bibinfo  {journal} {Physical Review Letters}\ }\textbf {\bibinfo {volume} {97}},\ \href {https://doi.org/10.1103/physrevlett.97.120405} {10.1103/physrevlett.97.120405} (\bibinfo {year} {2006})\BibitemShut {NoStop}%
\bibitem [{\citenamefont {Bennett}(1992)}]{B92}%
  \BibitemOpen
  \bibfield  {author} {\bibinfo {author} {\bibfnamefont {C.~H.}\ \bibnamefont {Bennett}},\ }\href {https://doi.org/10.1103/PhysRevLett.68.3121} {\bibfield  {journal} {\bibinfo  {journal} {Phys. Rev. Lett.}\ }\textbf {\bibinfo {volume} {68}},\ \bibinfo {pages} {3121} (\bibinfo {year} {1992})}\BibitemShut {NoStop}%
\bibitem [{\citenamefont {Koßmann}\ and\ \citenamefont {Schwonnek}(2025)}]{koßmann2025reliableentropyestimationobserved}%
  \BibitemOpen
  \bibfield  {author} {\bibinfo {author} {\bibfnamefont {G.}~\bibnamefont {Koßmann}}\ and\ \bibinfo {author} {\bibfnamefont {R.}~\bibnamefont {Schwonnek}},\ }\href {https://arxiv.org/abs/2411.04858} {\bibinfo {title} {Reliable entropy estimation from observed statistics for device-independent quantum cryptography}} (\bibinfo {year} {2025}),\ \Eprint {https://arxiv.org/abs/2411.04858} {arXiv:2411.04858 [quant-ph]} \BibitemShut {NoStop}%
\bibitem [{\citenamefont {Ko{\ss}mann}\ and\ \citenamefont {Schwonnek}(2024)}]{kossmann2024_optimization}%
  \BibitemOpen
  \bibfield  {author} {\bibinfo {author} {\bibfnamefont {G.}~\bibnamefont {Ko{\ss}mann}}\ and\ \bibinfo {author} {\bibfnamefont {R.}~\bibnamefont {Schwonnek}},\ }\href@noop {} {\bibfield  {journal} {\bibinfo  {journal} {arXiv:2404.17016}\ } (\bibinfo {year} {2024})}\BibitemShut {NoStop}%
\bibitem [{\citenamefont {Frenkel}(2023)}]{Frenkel2023}%
  \BibitemOpen
  \bibfield  {author} {\bibinfo {author} {\bibfnamefont {P.~E.}\ \bibnamefont {Frenkel}},\ }\href@noop {} {\bibfield  {journal} {\bibinfo  {journal} {Quantum}\ }\textbf {\bibinfo {volume} {7}},\ \bibinfo {pages} {1102} (\bibinfo {year} {2023})}\BibitemShut {NoStop}%
\bibitem [{\citenamefont {Chaves}\ \emph {et~al.}(2015)\citenamefont {Chaves}, \citenamefont {Brask},\ and\ \citenamefont {Brunner}}]{PhysRevLett.115.110501}%
  \BibitemOpen
  \bibfield  {author} {\bibinfo {author} {\bibfnamefont {R.}~\bibnamefont {Chaves}}, \bibinfo {author} {\bibfnamefont {J.~B.}\ \bibnamefont {Brask}},\ and\ \bibinfo {author} {\bibfnamefont {N.}~\bibnamefont {Brunner}},\ }\href {https://doi.org/10.1103/PhysRevLett.115.110501} {\bibfield  {journal} {\bibinfo  {journal} {Phys. Rev. Lett.}\ }\textbf {\bibinfo {volume} {115}},\ \bibinfo {pages} {110501} (\bibinfo {year} {2015})}\BibitemShut {NoStop}%
\bibitem [{\citenamefont {Tavakoli}\ \emph {et~al.}(2021)\citenamefont {Tavakoli}, \citenamefont {Pauwels}, \citenamefont {Woodhead},\ and\ \citenamefont {Pironio}}]{tavakoli2021correlations}%
  \BibitemOpen
  \bibfield  {author} {\bibinfo {author} {\bibfnamefont {A.}~\bibnamefont {Tavakoli}}, \bibinfo {author} {\bibfnamefont {J.}~\bibnamefont {Pauwels}}, \bibinfo {author} {\bibfnamefont {E.}~\bibnamefont {Woodhead}},\ and\ \bibinfo {author} {\bibfnamefont {S.}~\bibnamefont {Pironio}},\ }\href@noop {} {\bibfield  {journal} {\bibinfo  {journal} {PRX Quantum}\ }\textbf {\bibinfo {volume} {2}},\ \bibinfo {pages} {040357} (\bibinfo {year} {2021})}\BibitemShut {NoStop}%
\bibitem [{\citenamefont {Holevo}(1973)}]{Holevo1973}%
  \BibitemOpen
  \bibfield  {author} {\bibinfo {author} {\bibfnamefont {A.~S.}\ \bibnamefont {Holevo}},\ }\href@noop {} {\bibfield  {journal} {\bibinfo  {journal} {Journal of Multivariate Analysis}\ }\textbf {\bibinfo {volume} {3}},\ \bibinfo {pages} {337} (\bibinfo {year} {1973})}\BibitemShut {NoStop}%
\bibitem [{\citenamefont {Nielsen}\ and\ \citenamefont {Chuang}(2010)}]{Nielsen_Chuang_2010}%
  \BibitemOpen
  \bibfield  {author} {\bibinfo {author} {\bibfnamefont {M.~A.}\ \bibnamefont {Nielsen}}\ and\ \bibinfo {author} {\bibfnamefont {I.~L.}\ \bibnamefont {Chuang}},\ }\href@noop {} {\emph {\bibinfo {title} {Quantum Computation and Quantum Information: 10th Anniversary Edition}}}\ (\bibinfo  {publisher} {Cambridge University Press},\ \bibinfo {year} {2010})\BibitemShut {NoStop}%
\bibitem [{\citenamefont {Wang}\ \emph {et~al.}(2019)\citenamefont {Wang}, \citenamefont {Primaatmaja}, \citenamefont {Lavie}, \citenamefont {Varvitsiotis},\ and\ \citenamefont {Lim}}]{Wang2019}%
  \BibitemOpen
  \bibfield  {author} {\bibinfo {author} {\bibfnamefont {Y.}~\bibnamefont {Wang}}, \bibinfo {author} {\bibfnamefont {I.~W.}\ \bibnamefont {Primaatmaja}}, \bibinfo {author} {\bibfnamefont {E.}~\bibnamefont {Lavie}}, \bibinfo {author} {\bibfnamefont {A.}~\bibnamefont {Varvitsiotis}},\ and\ \bibinfo {author} {\bibfnamefont {C.~C.~W.}\ \bibnamefont {Lim}},\ }\href {https://doi.org/10.1038/s41534-019-0133-3} {\bibfield  {journal} {\bibinfo  {journal} {npj Quantum Information}\ }\textbf {\bibinfo {volume} {5}},\ \bibinfo {pages} {17} (\bibinfo {year} {2019})}\BibitemShut {NoStop}%
\bibitem [{\citenamefont {Yuen}\ \emph {et~al.}(1975)\citenamefont {Yuen}, \citenamefont {Kennedy},\ and\ \citenamefont {Lax}}]{YuenKennedyLax1975}%
  \BibitemOpen
  \bibfield  {author} {\bibinfo {author} {\bibfnamefont {H.}~\bibnamefont {Yuen}}, \bibinfo {author} {\bibfnamefont {R.}~\bibnamefont {Kennedy}},\ and\ \bibinfo {author} {\bibfnamefont {M.}~\bibnamefont {Lax}},\ }\href@noop {} {\bibfield  {journal} {\bibinfo  {journal} {IEEE Transactions on Information Theory}\ }\textbf {\bibinfo {volume} {21}},\ \bibinfo {pages} {125} (\bibinfo {year} {1975})}\BibitemShut {NoStop}%
\bibitem [{\citenamefont {Watrous}(2018)}]{Watrous2018}%
  \BibitemOpen
  \bibfield  {author} {\bibinfo {author} {\bibfnamefont {J.}~\bibnamefont {Watrous}},\ }\href@noop {} {\emph {\bibinfo {title} {The Theory of Quantum Information}}}\ (\bibinfo  {publisher} {Cambridge University Press},\ \bibinfo {year} {2018})\BibitemShut {NoStop}%
\bibitem [{\citenamefont {Pauwels}\ \emph {et~al.}(2025)\citenamefont {Pauwels}, \citenamefont {Pironio},\ and\ \citenamefont {Tavakoli}}]{Pauwels2025}%
  \BibitemOpen
  \bibfield  {author} {\bibinfo {author} {\bibfnamefont {J.}~\bibnamefont {Pauwels}}, \bibinfo {author} {\bibfnamefont {S.}~\bibnamefont {Pironio}},\ and\ \bibinfo {author} {\bibfnamefont {A.}~\bibnamefont {Tavakoli}},\ }\href {https://doi.org/10.22331/q-2025-02-18-1637} {\bibfield  {journal} {\bibinfo  {journal} {Quantum}\ }\textbf {\bibinfo {volume} {9}},\ \bibinfo {pages} {1637} (\bibinfo {year} {2025})}\BibitemShut {NoStop}%
\bibitem [{\citenamefont {Wolf}(2021)}]{wolf2021quantum}%
  \BibitemOpen
  \bibfield  {author} {\bibinfo {author} {\bibfnamefont {R.}~\bibnamefont {Wolf}},\ }\href {https://doi.org/10.1007/978-3-030-73991-1} {\emph {\bibinfo {title} {Quantum Key Distribution: An Introduction with Exercises}}},\ \bibinfo {series} {Lecture Notes in Physics}, Vol.\ \bibinfo {volume} {988}\ (\bibinfo  {publisher} {Springer Cham},\ \bibinfo {year} {2021})\BibitemShut {NoStop}%
\bibitem [{\citenamefont {Devetak}\ and\ \citenamefont {Winter}(2005)}]{DevetakWinter2005}%
  \BibitemOpen
  \bibfield  {author} {\bibinfo {author} {\bibfnamefont {I.}~\bibnamefont {Devetak}}\ and\ \bibinfo {author} {\bibfnamefont {A.}~\bibnamefont {Winter}},\ }\href {https://doi.org/10.1098/rspa.2004.1372} {\bibfield  {journal} {\bibinfo  {journal} {Proceedings of the Royal Society A: Mathematical, Physical and Engineering Sciences}\ }\textbf {\bibinfo {volume} {461}},\ \bibinfo {pages} {207} (\bibinfo {year} {2005})}\BibitemShut {NoStop}%
\bibitem [{\citenamefont {Curr{\'a}s-Lorenzo}\ \emph {et~al.}(2025{\natexlab{b}})\citenamefont {Curr{\'a}s-Lorenzo}, \citenamefont {Navarrete}, \citenamefont {N{\'u}{\~n}ez-Bon}, \citenamefont {Pereira},\ and\ \citenamefont {Curty}}]{CurrasLorenzo2025num}%
  \BibitemOpen
  \bibfield  {author} {\bibinfo {author} {\bibfnamefont {G.}~\bibnamefont {Curr{\'a}s-Lorenzo}}, \bibinfo {author} {\bibfnamefont {{\'A}.}~\bibnamefont {Navarrete}}, \bibinfo {author} {\bibfnamefont {J.}~\bibnamefont {N{\'u}{\~n}ez-Bon}}, \bibinfo {author} {\bibfnamefont {M.}~\bibnamefont {Pereira}},\ and\ \bibinfo {author} {\bibfnamefont {M.}~\bibnamefont {Curty}},\ }\href {https://doi.org/10.1088/2058-9565/addae1} {\bibfield  {journal} {\bibinfo  {journal} {Quantum Sci. Technol.}\ }\textbf {\bibinfo {volume} {10}},\ \bibinfo {pages} {035031} (\bibinfo {year} {2025}{\natexlab{b}})}\BibitemShut {NoStop}%
\bibitem [{\citenamefont {Pereira}\ \emph {et~al.}(2025)\citenamefont {Pereira}, \citenamefont {Curr{\'a}s-Lorenzo},\ and\ \citenamefont {Ara{\'u}jo}}]{Pereira2025}%
  \BibitemOpen
  \bibfield  {author} {\bibinfo {author} {\bibfnamefont {M.}~\bibnamefont {Pereira}}, \bibinfo {author} {\bibfnamefont {G.}~\bibnamefont {Curr{\'a}s-Lorenzo}},\ and\ \bibinfo {author} {\bibfnamefont {M.}~\bibnamefont {Ara{\'u}jo}},\ }\href@noop {} {\bibinfo {title} {Optimal key rates for quantum key distribution with partial source characterization}} (\bibinfo {year} {2025}),\ \Eprint {https://arxiv.org/abs/2510.13085} {arXiv:2510.13085 [quant-ph]} \BibitemShut {NoStop}%
\bibitem [{\citenamefont {Maurer}\ and\ \citenamefont {Wolf}(1999)}]{MaurerWolf1999}%
  \BibitemOpen
  \bibfield  {author} {\bibinfo {author} {\bibfnamefont {U.~M.}\ \bibnamefont {Maurer}}\ and\ \bibinfo {author} {\bibfnamefont {S.}~\bibnamefont {Wolf}},\ }\href {https://doi.org/10.1109/18.748999} {\bibfield  {journal} {\bibinfo  {journal} {IEEE Trans. Inf. Theory}\ }\textbf {\bibinfo {volume} {45}},\ \bibinfo {pages} {499} (\bibinfo {year} {1999})}\BibitemShut {NoStop}%
\bibitem [{\citenamefont {Clauser}\ \emph {et~al.}(1969)\citenamefont {Clauser}, \citenamefont {Horne}, \citenamefont {Shimony},\ and\ \citenamefont {Holt}}]{Clauser1969}%
  \BibitemOpen
  \bibfield  {author} {\bibinfo {author} {\bibfnamefont {J.~F.}\ \bibnamefont {Clauser}}, \bibinfo {author} {\bibfnamefont {M.~A.}\ \bibnamefont {Horne}}, \bibinfo {author} {\bibfnamefont {A.}~\bibnamefont {Shimony}},\ and\ \bibinfo {author} {\bibfnamefont {R.~A.}\ \bibnamefont {Holt}},\ }\href {https://doi.org/10.1103/PhysRevLett.23.880} {\bibfield  {journal} {\bibinfo  {journal} {Phys. Rev. Lett.}\ }\textbf {\bibinfo {volume} {23}},\ \bibinfo {pages} {880} (\bibinfo {year} {1969})}\BibitemShut {NoStop}%
\bibitem [{\citenamefont {Nayak}(1999)}]{814608}%
  \BibitemOpen
  \bibfield  {author} {\bibinfo {author} {\bibfnamefont {A.}~\bibnamefont {Nayak}},\ }in\ \href {https://doi.org/10.1109/SFFCS.1999.814608} {\emph {\bibinfo {booktitle} {40th Annual Symposium on Foundations of Computer Science (Cat. No.99CB37039)}}}\ (\bibinfo {year} {1999})\ pp.\ \bibinfo {pages} {369--376}\BibitemShut {NoStop}%
\bibitem [{\citenamefont {Tavakoli}\ \emph {et~al.}(2015)\citenamefont {Tavakoli}, \citenamefont {Hameedi}, \citenamefont {Marques},\ and\ \citenamefont {Bourennane}}]{PhysRevLett.114.170502}%
  \BibitemOpen
  \bibfield  {author} {\bibinfo {author} {\bibfnamefont {A.}~\bibnamefont {Tavakoli}}, \bibinfo {author} {\bibfnamefont {A.}~\bibnamefont {Hameedi}}, \bibinfo {author} {\bibfnamefont {B.}~\bibnamefont {Marques}},\ and\ \bibinfo {author} {\bibfnamefont {M.}~\bibnamefont {Bourennane}},\ }\href {https://doi.org/10.1103/PhysRevLett.114.170502} {\bibfield  {journal} {\bibinfo  {journal} {Phys. Rev. Lett.}\ }\textbf {\bibinfo {volume} {114}},\ \bibinfo {pages} {170502} (\bibinfo {year} {2015})}\BibitemShut {NoStop}%
\bibitem [{\citenamefont {Woodhead}\ \emph {et~al.}(2013)\citenamefont {Woodhead}, \citenamefont {Lim},\ and\ \citenamefont {Pironio}}]{erik_2013}%
  \BibitemOpen
  \bibfield  {author} {\bibinfo {author} {\bibfnamefont {E.}~\bibnamefont {Woodhead}}, \bibinfo {author} {\bibfnamefont {C.~C.~W.}\ \bibnamefont {Lim}},\ and\ \bibinfo {author} {\bibfnamefont {S.}~\bibnamefont {Pironio}},\ }in\ \href@noop {} {\emph {\bibinfo {booktitle} {Theory of Quantum Computation, Communication, and Cryptography}}},\ \bibinfo {editor} {edited by\ \bibinfo {editor} {\bibfnamefont {K.}~\bibnamefont {Iwama}}, \bibinfo {editor} {\bibfnamefont {Y.}~\bibnamefont {Kawano}},\ and\ \bibinfo {editor} {\bibfnamefont {M.}~\bibnamefont {Murao}}}\ (\bibinfo  {publisher} {Springer Berlin Heidelberg},\ \bibinfo {address} {Berlin, Heidelberg},\ \bibinfo {year} {2013})\ pp.\ \bibinfo {pages} {107--115}\BibitemShut {NoStop}%
\bibitem [{\citenamefont {Woodhead}\ and\ \citenamefont {Pironio}(2015)}]{PhysRevLett.115.150501}%
  \BibitemOpen
  \bibfield  {author} {\bibinfo {author} {\bibfnamefont {E.}~\bibnamefont {Woodhead}}\ and\ \bibinfo {author} {\bibfnamefont {S.}~\bibnamefont {Pironio}},\ }\href {https://doi.org/10.1103/PhysRevLett.115.150501} {\bibfield  {journal} {\bibinfo  {journal} {Phys. Rev. Lett.}\ }\textbf {\bibinfo {volume} {115}},\ \bibinfo {pages} {150501} (\bibinfo {year} {2015})}\BibitemShut {NoStop}%
\bibitem [{\citenamefont {Manna}\ \emph {et~al.}(2024)\citenamefont {Manna}, \citenamefont {Chaturvedi},\ and\ \citenamefont {Saha}}]{Manna2024}%
  \BibitemOpen
  \bibfield  {author} {\bibinfo {author} {\bibfnamefont {S.}~\bibnamefont {Manna}}, \bibinfo {author} {\bibfnamefont {A.}~\bibnamefont {Chaturvedi}},\ and\ \bibinfo {author} {\bibfnamefont {D.}~\bibnamefont {Saha}},\ }\href {https://doi.org/10.1103/PhysRevResearch.6.043269} {\bibfield  {journal} {\bibinfo  {journal} {Phys. Rev. Res.}\ }\textbf {\bibinfo {volume} {6}},\ \bibinfo {pages} {043269} (\bibinfo {year} {2024})}\BibitemShut {NoStop}%
\bibitem [{\citenamefont {Woodhead}(2016)}]{Woodhead_2016}%
  \BibitemOpen
  \bibfield  {author} {\bibinfo {author} {\bibfnamefont {E.}~\bibnamefont {Woodhead}},\ }\href {https://doi.org/10.1088/1367-2630/18/5/055010} {\bibfield  {journal} {\bibinfo  {journal} {New Journal of Physics}\ }\textbf {\bibinfo {volume} {18}},\ \bibinfo {pages} {055010} (\bibinfo {year} {2016})}\BibitemShut {NoStop}%
\bibitem [{\citenamefont {Popescu}\ and\ \citenamefont {Rohrlich}(1992)}]{PopescuRohrlich1992}%
  \BibitemOpen
  \bibfield  {author} {\bibinfo {author} {\bibfnamefont {S.}~\bibnamefont {Popescu}}\ and\ \bibinfo {author} {\bibfnamefont {D.}~\bibnamefont {Rohrlich}},\ }\href {https://doi.org/10.1016/0375-9601(92)90819-8} {\bibfield  {journal} {\bibinfo  {journal} {Phys. Lett. A}\ }\textbf {\bibinfo {volume} {169}},\ \bibinfo {pages} {411} (\bibinfo {year} {1992})}\BibitemShut {NoStop}%
\bibitem [{\citenamefont {Mayers}\ and\ \citenamefont {Yao}(2004)}]{MayersYao2004}%
  \BibitemOpen
  \bibfield  {author} {\bibinfo {author} {\bibfnamefont {D.}~\bibnamefont {Mayers}}\ and\ \bibinfo {author} {\bibfnamefont {A.}~\bibnamefont {Yao}},\ }\href {https://doi.org/10.26421/QIC4.4-3} {\bibfield  {journal} {\bibinfo  {journal} {Quantum Inf. Comput.}\ }\textbf {\bibinfo {volume} {4}},\ \bibinfo {pages} {273} (\bibinfo {year} {2004})}\BibitemShut {NoStop}%
\bibitem [{\citenamefont {McKague}\ \emph {et~al.}(2012)\citenamefont {McKague}, \citenamefont {Yang},\ and\ \citenamefont {Scarani}}]{McKague2012}%
  \BibitemOpen
  \bibfield  {author} {\bibinfo {author} {\bibfnamefont {M.}~\bibnamefont {McKague}}, \bibinfo {author} {\bibfnamefont {T.~H.}\ \bibnamefont {Yang}},\ and\ \bibinfo {author} {\bibfnamefont {V.}~\bibnamefont {Scarani}},\ }\href {https://doi.org/10.1088/1751-8113/45/45/455304} {\bibfield  {journal} {\bibinfo  {journal} {J. Phys. A: Math. Theor.}\ }\textbf {\bibinfo {volume} {45}},\ \bibinfo {pages} {455304} (\bibinfo {year} {2012})}\BibitemShut {NoStop}%
\bibitem [{\citenamefont {{\v{S}}upi{\'c}}\ and\ \citenamefont {Bowles}(2020)}]{Supic2020}%
  \BibitemOpen
  \bibfield  {author} {\bibinfo {author} {\bibfnamefont {I.}~\bibnamefont {{\v{S}}upi{\'c}}}\ and\ \bibinfo {author} {\bibfnamefont {J.}~\bibnamefont {Bowles}},\ }\href {https://doi.org/10.22331/q-2020-09-30-337} {\bibfield  {journal} {\bibinfo  {journal} {Quantum}\ }\textbf {\bibinfo {volume} {4}},\ \bibinfo {pages} {337} (\bibinfo {year} {2020})}\BibitemShut {NoStop}%
\bibitem [{\citenamefont {Mizutani}\ \emph {et~al.}(2019)\citenamefont {Mizutani}, \citenamefont {Kato}, \citenamefont {Azuma}, \citenamefont {Curty}, \citenamefont {Ikuta}, \citenamefont {Yamamoto}, \citenamefont {Imoto}, \citenamefont {Lo},\ and\ \citenamefont {Tamaki}}]{Mizutani2019}%
  \BibitemOpen
  \bibfield  {author} {\bibinfo {author} {\bibfnamefont {A.}~\bibnamefont {Mizutani}}, \bibinfo {author} {\bibfnamefont {G.}~\bibnamefont {Kato}}, \bibinfo {author} {\bibfnamefont {K.}~\bibnamefont {Azuma}}, \bibinfo {author} {\bibfnamefont {M.}~\bibnamefont {Curty}}, \bibinfo {author} {\bibfnamefont {R.}~\bibnamefont {Ikuta}}, \bibinfo {author} {\bibfnamefont {T.}~\bibnamefont {Yamamoto}}, \bibinfo {author} {\bibfnamefont {N.}~\bibnamefont {Imoto}}, \bibinfo {author} {\bibfnamefont {H.-K.}\ \bibnamefont {Lo}},\ and\ \bibinfo {author} {\bibfnamefont {K.}~\bibnamefont {Tamaki}},\ }\href {https://doi.org/10.1038/s41534-018-0122-y} {\bibfield  {journal} {\bibinfo  {journal} {npj Quantum Inf.}\ }\textbf {\bibinfo {volume} {5}},\ \bibinfo {pages} {8} (\bibinfo {year} {2019})}\BibitemShut {NoStop}%
\bibitem [{\citenamefont {Fuchs}\ and\ \citenamefont {van~de Graaf}(1999)}]{FuchsvandeGraaf1999}%
  \BibitemOpen
  \bibfield  {author} {\bibinfo {author} {\bibfnamefont {C.~A.}\ \bibnamefont {Fuchs}}\ and\ \bibinfo {author} {\bibfnamefont {J.}~\bibnamefont {van~de Graaf}},\ }\href {https://doi.org/10.1109/18.761271} {\bibfield  {journal} {\bibinfo  {journal} {IEEE Trans. Inf. Theory}\ }\textbf {\bibinfo {volume} {45}},\ \bibinfo {pages} {1216} (\bibinfo {year} {1999})}\BibitemShut {NoStop}%
\bibitem [{\citenamefont {Tamaki}\ \emph {et~al.}(2016)\citenamefont {Tamaki}, \citenamefont {Curty},\ and\ \citenamefont {Lucamarini}}]{Tamaki2016}%
  \BibitemOpen
  \bibfield  {author} {\bibinfo {author} {\bibfnamefont {K.}~\bibnamefont {Tamaki}}, \bibinfo {author} {\bibfnamefont {M.}~\bibnamefont {Curty}},\ and\ \bibinfo {author} {\bibfnamefont {M.}~\bibnamefont {Lucamarini}},\ }\href {https://doi.org/10.1088/1367-2630/18/6/065008} {\bibfield  {journal} {\bibinfo  {journal} {New J. Phys.}\ }\textbf {\bibinfo {volume} {18}},\ \bibinfo {pages} {065008} (\bibinfo {year} {2016})}\BibitemShut {NoStop}%
\bibitem [{\citenamefont {Sixto}\ \emph {et~al.}(2025)\citenamefont {Sixto}, \citenamefont {Navarrete}, \citenamefont {Pereira}, \citenamefont {Curr{\'a}s-Lorenzo}, \citenamefont {Tamaki},\ and\ \citenamefont {Curty}}]{Sixto2025}%
  \BibitemOpen
  \bibfield  {author} {\bibinfo {author} {\bibfnamefont {X.}~\bibnamefont {Sixto}}, \bibinfo {author} {\bibfnamefont {{\'A}.}~\bibnamefont {Navarrete}}, \bibinfo {author} {\bibfnamefont {M.}~\bibnamefont {Pereira}}, \bibinfo {author} {\bibfnamefont {G.}~\bibnamefont {Curr{\'a}s-Lorenzo}}, \bibinfo {author} {\bibfnamefont {K.}~\bibnamefont {Tamaki}},\ and\ \bibinfo {author} {\bibfnamefont {M.}~\bibnamefont {Curty}},\ }\href {https://doi.org/10.1088/2058-9565/addb6e} {\bibfield  {journal} {\bibinfo  {journal} {Quantum Science and Technology}\ }\textbf {\bibinfo {volume} {10}},\ \bibinfo {pages} {035034} (\bibinfo {year} {2025})}\BibitemShut {NoStop}%
\bibitem [{\citenamefont {Helstrom}(1969)}]{Helstrom1969}%
  \BibitemOpen
  \bibfield  {author} {\bibinfo {author} {\bibfnamefont {C.~W.}\ \bibnamefont {Helstrom}},\ }\href@noop {} {\bibfield  {journal} {\bibinfo  {journal} {Journal of Statistical Physics}\ }\textbf {\bibinfo {volume} {1}},\ \bibinfo {pages} {231} (\bibinfo {year} {1969})}\BibitemShut {NoStop}%
\bibitem [{\citenamefont {Pereira}\ \emph {et~al.}(2023)\citenamefont {Pereira}, \citenamefont {Currás-Lorenzo}, \citenamefont {Navarrete}, \citenamefont {Mizutani}, \citenamefont {Kato}, \citenamefont {Curty},\ and\ \citenamefont {Tamaki}}]{Pereira2023}%
  \BibitemOpen
  \bibfield  {author} {\bibinfo {author} {\bibfnamefont {M.}~\bibnamefont {Pereira}}, \bibinfo {author} {\bibfnamefont {G.}~\bibnamefont {Currás-Lorenzo}}, \bibinfo {author} {\bibfnamefont {{\'A}.}~\bibnamefont {Navarrete}}, \bibinfo {author} {\bibfnamefont {A.}~\bibnamefont {Mizutani}}, \bibinfo {author} {\bibfnamefont {G.}~\bibnamefont {Kato}}, \bibinfo {author} {\bibfnamefont {M.}~\bibnamefont {Curty}},\ and\ \bibinfo {author} {\bibfnamefont {K.}~\bibnamefont {Tamaki}},\ }\href {https://doi.org/10.1103/PhysRevResearch.5.023065} {\bibfield  {journal} {\bibinfo  {journal} {Phys. Rev. Res.}\ }\textbf {\bibinfo {volume} {5}},\ \bibinfo {pages} {023065} (\bibinfo {year} {2023})}\BibitemShut {NoStop}%
\bibitem [{\citenamefont {Brask}\ \emph {et~al.}(2026)\citenamefont {Brask}, \citenamefont {Brunner}, \citenamefont {Pauwels}, \citenamefont {Rusca},\ and\ \citenamefont {Tavakoli}}]{BraskReview2026}%
  \BibitemOpen
  \bibfield  {author} {\bibinfo {author} {\bibfnamefont {J.~B.}\ \bibnamefont {Brask}}, \bibinfo {author} {\bibfnamefont {N.}~\bibnamefont {Brunner}}, \bibinfo {author} {\bibfnamefont {J.}~\bibnamefont {Pauwels}}, \bibinfo {author} {\bibfnamefont {D.}~\bibnamefont {Rusca}},\ and\ \bibinfo {author} {\bibfnamefont {A.}~\bibnamefont {Tavakoli}},\ }\href@noop {} {\bibinfo {title} {Quantum correlations in prepare-and-measure scenarios and their semi-device-independent applications}} (\bibinfo {year} {2026}),\ \Eprint {https://arxiv.org/abs/2603.23604} {arXiv:2603.23604 [quant-ph]} \BibitemShut {NoStop}%
\bibitem [{\citenamefont {Brown}\ \emph {et~al.}(2024)\citenamefont {Brown}, \citenamefont {Fawzi},\ and\ \citenamefont {Fawzi}}]{brown2024device}%
  \BibitemOpen
  \bibfield  {author} {\bibinfo {author} {\bibfnamefont {P.}~\bibnamefont {Brown}}, \bibinfo {author} {\bibfnamefont {H.}~\bibnamefont {Fawzi}},\ and\ \bibinfo {author} {\bibfnamefont {O.}~\bibnamefont {Fawzi}},\ }\href {https://doi.org/10.22331/q-2024-08-27-1445} {\bibfield  {journal} {\bibinfo  {journal} {{Quantum}}\ }\textbf {\bibinfo {volume} {8}},\ \bibinfo {pages} {1445} (\bibinfo {year} {2024})}\BibitemShut {NoStop}%
\bibitem [{\citenamefont {Burgdorf}\ \emph {et~al.}(2016)\citenamefont {Burgdorf}, \citenamefont {Klep},\ and\ \citenamefont {Povh}}]{burgdorf2016optimization}%
  \BibitemOpen
  \bibfield  {author} {\bibinfo {author} {\bibfnamefont {S.}~\bibnamefont {Burgdorf}}, \bibinfo {author} {\bibfnamefont {I.}~\bibnamefont {Klep}},\ and\ \bibinfo {author} {\bibfnamefont {J.}~\bibnamefont {Povh}},\ }\href {https://doi.org/10.1007/978-3-319-33338-0} {\bibfield  {journal} {\bibinfo  {journal} {SpringerLink}\ } (\bibinfo {year} {2016})}\BibitemShut {NoStop}%
\bibitem [{\citenamefont {Pironio}\ \emph {et~al.}(2010{\natexlab{a}})\citenamefont {Pironio}, \citenamefont {Navascu{\'e}s},\ and\ \citenamefont {Ac{\'\i}n}}]{PNA2010}%
  \BibitemOpen
  \bibfield  {author} {\bibinfo {author} {\bibfnamefont {S.}~\bibnamefont {Pironio}}, \bibinfo {author} {\bibfnamefont {M.}~\bibnamefont {Navascu{\'e}s}},\ and\ \bibinfo {author} {\bibfnamefont {A.}~\bibnamefont {Ac{\'\i}n}},\ }\href {https://doi.org/10.1137/090760155} {\bibfield  {journal} {\bibinfo  {journal} {SIAM J. Optim.}\ }\textbf {\bibinfo {volume} {20}},\ \bibinfo {pages} {2157} (\bibinfo {year} {2010}{\natexlab{a}})}\BibitemShut {NoStop}%
\bibitem [{\citenamefont {Roch~i Carceller}\ \emph {et~al.}(2025)\citenamefont {Roch~i Carceller}, \citenamefont {Pauwels}, \citenamefont {Pironio},\ and\ \citenamefont {Tavakoli}}]{glty-kkbp}%
  \BibitemOpen
  \bibfield  {author} {\bibinfo {author} {\bibfnamefont {C.}~\bibnamefont {Roch~i Carceller}}, \bibinfo {author} {\bibfnamefont {J.}~\bibnamefont {Pauwels}}, \bibinfo {author} {\bibfnamefont {S.}~\bibnamefont {Pironio}},\ and\ \bibinfo {author} {\bibfnamefont {A.}~\bibnamefont {Tavakoli}},\ }\href {https://doi.org/10.1103/glty-kkbp} {\bibfield  {journal} {\bibinfo  {journal} {Phys. Rev. Lett.}\ }\textbf {\bibinfo {volume} {135}},\ \bibinfo {pages} {140802} (\bibinfo {year} {2025})}\BibitemShut {NoStop}%
\bibitem [{\citenamefont {Navascu{\'e}s}\ \emph {et~al.}(2007)\citenamefont {Navascu{\'e}s}, \citenamefont {Pironio},\ and\ \citenamefont {Ac{\'\i}n}}]{Navascues2007}%
  \BibitemOpen
  \bibfield  {author} {\bibinfo {author} {\bibfnamefont {M.}~\bibnamefont {Navascu{\'e}s}}, \bibinfo {author} {\bibfnamefont {S.}~\bibnamefont {Pironio}},\ and\ \bibinfo {author} {\bibfnamefont {A.}~\bibnamefont {Ac{\'\i}n}},\ }\href {https://doi.org/10.1103/PhysRevLett.98.010401} {\bibfield  {journal} {\bibinfo  {journal} {Phys. Rev. Lett.}\ }\textbf {\bibinfo {volume} {98}},\ \bibinfo {pages} {010401} (\bibinfo {year} {2007})}\BibitemShut {NoStop}%
\bibitem [{\citenamefont {Navascués}\ \emph {et~al.}(2008)\citenamefont {Navascués}, \citenamefont {Pironio},\ and\ \citenamefont {Acín}}]{Navascues2008}%
  \BibitemOpen
  \bibfield  {author} {\bibinfo {author} {\bibfnamefont {M.}~\bibnamefont {Navascués}}, \bibinfo {author} {\bibfnamefont {S.}~\bibnamefont {Pironio}},\ and\ \bibinfo {author} {\bibfnamefont {A.}~\bibnamefont {Acín}},\ }\href {https://doi.org/10.1088/1367-2630/10/7/073013} {\bibfield  {journal} {\bibinfo  {journal} {New Journal of Physics}\ }\textbf {\bibinfo {volume} {10}},\ \bibinfo {pages} {073013} (\bibinfo {year} {2008})}\BibitemShut {NoStop}%
\bibitem [{\citenamefont {Chaturvedi}\ \emph {et~al.}(2021)\citenamefont {Chaturvedi}, \citenamefont {Farkas},\ and\ \citenamefont {Wright}}]{Chaturvedi2021}%
  \BibitemOpen
  \bibfield  {author} {\bibinfo {author} {\bibfnamefont {A.}~\bibnamefont {Chaturvedi}}, \bibinfo {author} {\bibfnamefont {M.}~\bibnamefont {Farkas}},\ and\ \bibinfo {author} {\bibfnamefont {V.~J.}\ \bibnamefont {Wright}},\ }\href {https://doi.org/10.22331/q-2021-06-29-484} {\bibfield  {journal} {\bibinfo  {journal} {Quantum}\ }\textbf {\bibinfo {volume} {5}},\ \bibinfo {pages} {484} (\bibinfo {year} {2021})}\BibitemShut {NoStop}%
\bibitem [{\citenamefont {Masini}\ \emph {et~al.}(2024)\citenamefont {Masini}, \citenamefont {Ioannou}, \citenamefont {Brunner}, \citenamefont {Pironio},\ and\ \citenamefont {Sekatski}}]{Masini2024}%
  \BibitemOpen
  \bibfield  {author} {\bibinfo {author} {\bibfnamefont {M.}~\bibnamefont {Masini}}, \bibinfo {author} {\bibfnamefont {M.}~\bibnamefont {Ioannou}}, \bibinfo {author} {\bibfnamefont {N.}~\bibnamefont {Brunner}}, \bibinfo {author} {\bibfnamefont {S.}~\bibnamefont {Pironio}},\ and\ \bibinfo {author} {\bibfnamefont {P.}~\bibnamefont {Sekatski}},\ }\href {https://doi.org/10.22331/q-2024-12-23-1574} {\bibfield  {journal} {\bibinfo  {journal} {Quantum}\ }\textbf {\bibinfo {volume} {8}},\ \bibinfo {pages} {1574} (\bibinfo {year} {2024})}\BibitemShut {NoStop}%
\bibitem [{\citenamefont {Lobo}\ \emph {et~al.}(2026)\citenamefont {Lobo}, \citenamefont {Balanz{\'o}-Juand{\'o}},\ and\ \citenamefont {Pironio}}]{Lobo2026}%
  \BibitemOpen
  \bibfield  {author} {\bibinfo {author} {\bibfnamefont {E.~P.}\ \bibnamefont {Lobo}}, \bibinfo {author} {\bibfnamefont {M.}~\bibnamefont {Balanz{\'o}-Juand{\'o}}},\ and\ \bibinfo {author} {\bibfnamefont {S.}~\bibnamefont {Pironio}},\ }\href@noop {} {\bibinfo {title} {Generalized measurement incompatibility}} (\bibinfo {year} {2026}),\ \Eprint {https://arxiv.org/abs/2605.16151} {arXiv:2605.16151 [quant-ph]} \BibitemShut {NoStop}%
\bibitem [{\citenamefont {Pironio}\ \emph {et~al.}(2010{\natexlab{b}})\citenamefont {Pironio}, \citenamefont {Ac\'{\i}n}, \citenamefont {Massar}, \citenamefont {Boyer de~la Giroday}, \citenamefont {Matsukevich}, \citenamefont {Maunz}, \citenamefont {Olmschenk}, \citenamefont {Hayes}, \citenamefont {Luo}, \citenamefont {Manning},\ and\ \citenamefont {Monroe}}]{Pironio2010}%
  \BibitemOpen
  \bibfield  {author} {\bibinfo {author} {\bibfnamefont {S.}~\bibnamefont {Pironio}}, \bibinfo {author} {\bibfnamefont {A.}~\bibnamefont {Ac\'{\i}n}}, \bibinfo {author} {\bibfnamefont {S.}~\bibnamefont {Massar}}, \bibinfo {author} {\bibfnamefont {A.}~\bibnamefont {Boyer de~la Giroday}}, \bibinfo {author} {\bibfnamefont {D.~N.}\ \bibnamefont {Matsukevich}}, \bibinfo {author} {\bibfnamefont {P.}~\bibnamefont {Maunz}}, \bibinfo {author} {\bibfnamefont {S.}~\bibnamefont {Olmschenk}}, \bibinfo {author} {\bibfnamefont {D.}~\bibnamefont {Hayes}}, \bibinfo {author} {\bibfnamefont {L.}~\bibnamefont {Luo}}, \bibinfo {author} {\bibfnamefont {T.~A.}\ \bibnamefont {Manning}},\ and\ \bibinfo {author} {\bibfnamefont {C.}~\bibnamefont {Monroe}},\ }\href {https://doi.org/10.1038/nature09008} {\bibfield  {journal} {\bibinfo  {journal} {Nature}\ }\textbf {\bibinfo {volume} {464}},\ \bibinfo {pages} {1021} (\bibinfo {year} {2010}{\natexlab{b}})}\BibitemShut {NoStop}%
\bibitem [{\citenamefont {Metger}\ \emph {et~al.}(2024)\citenamefont {Metger}, \citenamefont {Fawzi}, \citenamefont {Sutter},\ and\ \citenamefont {Renner}}]{Metger2024}%
  \BibitemOpen
  \bibfield  {author} {\bibinfo {author} {\bibfnamefont {T.}~\bibnamefont {Metger}}, \bibinfo {author} {\bibfnamefont {O.}~\bibnamefont {Fawzi}}, \bibinfo {author} {\bibfnamefont {D.}~\bibnamefont {Sutter}},\ and\ \bibinfo {author} {\bibfnamefont {R.}~\bibnamefont {Renner}},\ }\href {https://doi.org/10.1007/s00220-024-05121-4} {\bibfield  {journal} {\bibinfo  {journal} {Commun. Math. Phys.}\ }\textbf {\bibinfo {volume} {405}},\ \bibinfo {pages} {261} (\bibinfo {year} {2024})}\BibitemShut {NoStop}%
\bibitem [{\citenamefont {Chaturvedi}\ \emph {et~al.}(2026)\citenamefont {Chaturvedi}, \citenamefont {Viola}, \citenamefont {Panwar}, \citenamefont {Prasad},\ and\ \citenamefont {Saha}}]{Chaturvedi2026}%
  \BibitemOpen
  \bibfield  {author} {\bibinfo {author} {\bibfnamefont {A.}~\bibnamefont {Chaturvedi}}, \bibinfo {author} {\bibfnamefont {G.}~\bibnamefont {Viola}}, \bibinfo {author} {\bibfnamefont {E.}~\bibnamefont {Panwar}}, \bibinfo {author} {\bibfnamefont {T.}~\bibnamefont {Prasad}},\ and\ \bibinfo {author} {\bibfnamefont {D.}~\bibnamefont {Saha}},\ }\href@noop {} {\bibinfo {title} {Semi-device-independent quantum key distribution from operational assumptions}} (\bibinfo {year} {2026}),\ \Eprint {https://arxiv.org/abs/2607.06682} {arXiv:2607.06682 [quant-ph]} \BibitemShut {NoStop}%
\bibitem [{\citenamefont {Mishra}(2026)}]{RIQKD}%
  \BibitemOpen
  \bibfield  {author} {\bibinfo {author} {\bibfnamefont {A.}~\bibnamefont {Mishra}},\ }\href@noop {} {\bibinfo {title} {{RIQKD}: Code and data for prepare-and-measure quantum key distribution from a bound on information leakage}},\ \bibinfo {howpublished} {\url{https://github.com/abh1mishra/RIQKD}} (\bibinfo {year} {2026})\BibitemShut {NoStop}%
\bibitem [{\citenamefont {{Bermejo Mor{\'a}n}}\ and\ \citenamefont {Mishra}(2026)}]{BermejoMoran2026PCPOP}%
  \BibitemOpen
  \bibfield  {author} {\bibinfo {author} {\bibfnamefont {M.}~\bibnamefont {{Bermejo Mor{\'a}n}}}\ and\ \bibinfo {author} {\bibfnamefont {A.}~\bibnamefont {Mishra}},\ }\href {https://arxiv.org/abs/2607.09339} {\bibinfo {title} {{PCPOP.jl}: A {J}ulia package for partially commutative polynomial optimization}} (\bibinfo {year} {2026}),\ \bibinfo {note} {code available at \url{https://github.com/abh1mishra/PCPOP.jl}},\ \Eprint {https://arxiv.org/abs/2607.09339} {arXiv:2607.09339 [quant-ph]} \BibitemShut {NoStop}%
\end{thebibliography}%

\onecolumngrid
\appendix

\section{Classical sets and witnesses}
\label{app:classical}

This appendix describes how the classical quantities of Sec.~\ref{sec:lp} were computed and proves the classical bounds of the inequalities of Sec.~\ref{sec:facets}.

\emph{Deterministic responses.} The classical models~\eqref{eq:classical_model}--\eqref{eq:classical_pg} can be reduced to deterministic responses of Bob's device.

\begin{lemma}
\label{lem:deterministic}
Every $p\in\cset_G$ admits a classical model~\eqref{eq:classical_model}--\eqref{eq:classical_pg} in which $p(b|z,y)\in\{0,1\}$.
\end{lemma}

\begin{proof}
Write $p(b|z,y)=\sum_\mu p(\mu)\,\delta_{b,f(z,y,\mu)}$ with $\mu$ independent of $x$, and let Alice's device send the message $z'=(z,\mu)$ with $p(z'|x)=p(z|x)p(\mu)$. Bob's response $b=f(z,y,\mu)$ is then deterministic, the behaviour is unchanged, and $\sum_{z,\mu}\max_x p(z|x)p(\mu)=\sum_z\max_x p(z|x)$.
\end{proof}

\emph{Linear programs.} Once Bob's responses are deterministic, the messages can be identified with his response functions $f:\mathcal Y\to\{0,1\}$, and the conditions~\eqref{eq:classical_model}--\eqref{eq:classical_pg} become linear in the distributions $p(f|x)$, in $G$ and in auxiliary variables $t_f\geq p(f|x)$ that replace $\max_xp(f|x)$~\cite{Tavakoli2020,Tavakoli2022}. The quantities of Table~\ref{tab:protocols} are the optimal values of the resulting linear programs, and the dual of the program for $\Gmin$ yields the inequalities of Sec.~\ref{sec:facets}.

\emph{Classical bounds of the inequalities.} By Lemma~\ref{lem:deterministic}, Bob's response to a message $z$ may be taken deterministic, $b=f_z(y)$, and we write $\beta_y(z)=(-1)^{f_z(y)}$. For the CHSH value~\eqref{eq:chsh} of $\prot43$,
\begin{equation*}
    S_{\mathrm{CHSH}}=\tfrac12\sum_z\Bigl\{\bigl[p(z|(1,0))-p(z|(1,1))\bigr]\bigl[\beta_2(z)+\beta_3(z)\bigr]+\bigl[p(z|(2,0))-p(z|(2,1))\bigr]\bigl[\beta_2(z)-\beta_3(z)\bigr]\Bigr\} .
\end{equation*}
For each $z$, exactly one of $\beta_2(z)\pm\beta_3(z)$ is non-zero, and it equals $\pm2$, so that the term of $z$ is $p(z|x_+)-p(z|x_-)$ for two distinct inputs $x_\pm$ that depend on $z$. Write $M_z=\max_xp(z|x)$, so that $\sum_zM_z\leq G$ by Eq.~\eqref{eq:classical_pg}. Bounding the term of $z$ by $M_z$ gives $S_{\mathrm{CHSH}}\leq\sum_zM_z\leq G$, which is Eq.~\eqref{eq:chsh_bound}. Bounding it instead by $M_z-p(z|x_-)\leq\sum_x\bigl(M_z-p(z|x)\bigr)$ and using $\sum_{z,x}p(z|x)=4$ gives $S_{\mathrm{CHSH}}\leq4\sum_zM_z-4\leq4G-4$, which is stronger for $G<4/3$. Linear programming confirms that $\min\{G,4G-4\}$ is the largest classical value of $S_{\mathrm{CHSH}}$ for every $G\in[1,4]$. The bounds of Eqs.~\eqref{eq:witness63}, \eqref{eq:witness63b} and~\eqref{eq:witness33} follow from the same counting argument. Write a witness as $\sum_{x,y}w_{x,y}E_{x,y}$; for the inequalities written with the correlators $C$, the non-zero weights are $\pm\tfrac12$. Let $c_x(z)=\sum_yw_{x,y}\beta_y(z)$ be the value that the message $z$ contributes for preparation $x$, so that the witness equals $\sum_z\sum_xp(z|x)\,c_x(z)$. Whenever, for every $z$, the positive parts of the $c_x(z)$ sum to at most $c$, for instance because $c_x(z)$ is at most $c$ for one preparation and at most $0$ for the others, the witness is bounded by $c\sum_z\max_xp(z|x)\leq cG$. Since $\sum_zp(z|x)=1$, one may also add constants $\mu_x$ to the $c_x(z)$ before applying this criterion, at the cost of the constant term $\sum_x\mu_x$ in the bound. For Eq.~\eqref{eq:witness63b}, $c=1$, because the four preparations involved carry the four sign patterns on the measurements $y=1,3$, as in Eq.~\eqref{eq:chsh_bound}; this gives $G$. For Eq.~\eqref{eq:witness33}, $c_{(1,0)}=\beta_1+\beta_3$, $c_{(2,0)}=\beta_2-\beta_3$ and $c_{(1,1)}=-\beta_1-\beta_2$, so that at most one of them equals $2$ and the others are then at most $0$; this gives $c=2$ and the bound $2G$. For Eq.~\eqref{eq:witness63}, $2c_x(z)\in\{\pm1,\pm3\}$, and the value $3$ is reached for at most one preparation because the six rows of signs are distinct; hence $2\sum_xp(z|x)\,c_x(z)\leq\sum_xp(z|x)+2\max_xp(z|x)$, and summing over $z$ gives $3+G$. For Eq.~\eqref{eq:witness33b}, adding $\mu=(-1,1,0)$ to $c_{(1,0)}=\beta_1+\beta_2+\beta_3$, $c_{(1,1)}=-\beta_2$ and $c_{(2,0)}=\beta_2-\beta_3$ gives, for each of the eight sign patterns, one value $2$ and two values at most $0$, hence $c=2$ and the bound $2G$. For Eq.~\eqref{eq:witness63c}, with the weights $\pm\tfrac12$ on the bases $\theta=1,3$ and $1$ on the three biases of $(2,0)$, adding $\mu=(-\tfrac12,\tfrac12,-1,0,-\tfrac12,\tfrac12)$ makes the positive parts of the six shifted coefficients sum to at most $2$ for every sign pattern, which gives $2G$, plus the constant $-\sum_x\mu_x=1$: the bound $1+2G$. Linear programming confirms that both inequalities are tight at the honest bound.

\section{Proof of Proposition~\ref{prop:ideal} for $\prot63$}
\label{app:ideal}

We use the virtual entanglement-based protocol of Sec.~\ref{sec:ideal}: the pure state $\ket\Psi$ on $CBE$, the observables $A_1$, $A_2$, $A_3$ on $C$ that select the three key bases, with outcome $+1$ for the preparations $(\theta,0)$, the $+1$ eigenstates of the observables $O_{(\theta-1)\pi/3}$ of Sec.~\ref{sec:protchsh}, and Bob's observables $B_1$, $B_2$, $B_3$ on $B$. These observables square to the identity, commute across systems, and commute with every operator on $E$. For the ideal behaviour, Eq.~\eqref{eq:virtual_corr} gives
\begin{equation}
    \label{eq:63_corr}
    \langle A_\theta B_y\rangle=\cos\bigl((\theta-y)\pi/3\bigr),
\end{equation}
that is, $1$ for $\theta=y$, $\tfrac12$ for $|\theta-y|=1$ and $-\tfrac12$ for $|\theta-y|=2$. We show that Eq.~\eqref{eq:63_corr} alone implies that a local isometry maps $\ket\Psi$ to $\ket{\Phi^+}_{C'B'}\otimes\ket\xi_{C''B''E}$ and $A_1\ket\Psi$ to $(Z\otimes\id)\ket{\Phi^+}\otimes\ket\xi$, which decouples Eve from the first key as in Sec.~\ref{sec:ideal}. The relabelling $(A_1,A_2,A_3)\to(A_2,A_3,-A_1)$, applied to Bob's observables as well, leaves Eq.~\eqref{eq:63_corr} invariant, so the same holds for the other two keys.

\emph{Perfect correlations.} For $\theta=y$, $\|(A_\theta-B_\theta)\ket\Psi\|^2=2-2\langle A_\theta B_\theta\rangle=0$, so that
\begin{equation}
    \label{eq:63_perfect}
    A_\theta\ket\Psi=B_\theta\ket\Psi,\qquad \theta=1,2,3 .
\end{equation}
The correlations between Alice's observables are therefore fixed, $\langle A_\theta A_{\theta'}\rangle=\langle A_\theta B_{\theta'}\rangle$, and $\|(A_1-A_2+A_3)\ket\Psi\|^2=3-2\langle A_1A_2\rangle+2\langle A_1A_3\rangle-2\langle A_2A_3\rangle=0$. The same holds on Bob's side, so
\begin{equation}
    \label{eq:63_linear}
    A_1\ket\Psi=(A_2-A_3)\ket\Psi,\qquad B_1\ket\Psi=(B_2-B_3)\ket\Psi .
\end{equation}

\emph{An anticommuting pair.} Let $K=(A_2+A_3)/\sqrt3$ and $D=(B_2+B_3)/\sqrt3$. By Eq.~\eqref{eq:63_perfect}, $K\ket\Psi=D\ket\Psi$. Moreover, $(A_2-A_3)A_1\ket\Psi=(A_2-A_3)B_1\ket\Psi=B_1(A_2-A_3)\ket\Psi=B_1A_1\ket\Psi=\ket\Psi$, while Eq.~\eqref{eq:63_linear} gives $(A_2-A_3)A_1\ket\Psi=(A_2-A_3)^2\ket\Psi=(2-\{A_2,A_3\})\ket\Psi$. Hence $\{A_2,A_3\}\ket\Psi=\ket\Psi$, which gives $K^2\ket\Psi=\ket\Psi$ and, since $D$ commutes with $K$, $D^2\ket\Psi=DK\ket\Psi=K^2\ket\Psi=\ket\Psi$. Finally, $A_1$ and $K$ anticommute on $\ket\Psi$:
\begin{equation}
    \label{eq:63_anticomm}
    \sqrt3\,\{A_1,K\}\ket\Psi=(B_2+B_3)B_1\ket\Psi+(A_2+A_3)A_1\ket\Psi=[B_3,B_2]\ket\Psi+[A_3,A_2]\ket\Psi=0 .
\end{equation}
The first equality uses $K\ket\Psi=D\ket\Psi$ and $A_1\ket\Psi=B_1\ket\Psi$ together with the commutation of $A_1$ with $D$, the second uses Eq.~\eqref{eq:63_linear}, and the third $[B_3,B_2]\ket\Psi=B_3A_2\ket\Psi-B_2A_3\ket\Psi=A_2A_3\ket\Psi-A_3A_2\ket\Psi$, by Eq.~\eqref{eq:63_perfect} and the commutation of Alice's and Bob's observables. In the virtual protocol, Alice thus holds a pair of observables, $A_1$ and $K$, that anticommute on $\ket\Psi$ and that Bob reproduces, $A_1\ket\Psi=B_1\ket\Psi$ and $K\ket\Psi=D\ket\Psi$. This is the structure from which the self-test of the maximal CHSH violation extracts a maximally entangled state~\cite{McKague2012,Supic2020}; here it is obtained from three bases, none of which anticommutes with another.

\emph{Extraction of $\ket{\Phi^+}$.} The operators $K$ and $D$ are Hermitian but not $\pm1$ observables. However, $K^2\ket\Psi=\ket\Psi$ and $K^2A_1\ket\Psi=K^2B_1\ket\Psi=A_1\ket\Psi$ show that $\ket\Psi$ and $A_1\ket\Psi$ lie in the span of the eigenspaces of $K$ with eigenvalues $\pm1$. Let $\tilde K$ be the $\pm1$ observable that coincides with $K$ on these eigenspaces and with the identity on their orthogonal complement, so that $\tilde K\ket\Psi=K\ket\Psi$ and $\tilde KA_1\ket\Psi=KA_1\ket\Psi$, and define $\tilde D$ from $D$ in the same way. The $\pm1$ observables $A_1$, $\tilde K$ on $C$ and $B_1$, $\tilde D$ on $B$ then satisfy $A_1\ket\Psi=B_1\ket\Psi$, $\tilde K\ket\Psi=\tilde D\ket\Psi$ and $\{A_1,\tilde K\}\ket\Psi=0$. These are the conditions under which the swap isometry extracts a maximally entangled state~\cite{McKague2012,Supic2020}. Explicitly, with $P_a=\tfrac12(\id+(-1)^aA_1)$ and $P'_b=\tfrac12(\id+(-1)^bB_1)$, the local isometry
\begin{equation}
    \Phi\ket\Psi=\sum_{a,b=0}^1\ket{ab}_{C'B'}\otimes\tilde K^a\tilde D^bP_aP'_b\ket\Psi
\end{equation}
gives, since $P'_b\ket\Psi=P_b\ket\Psi$ and $P_aP_b=\delta_{ab}P_a$, only the terms $a=b$. The term $a=b=1$ is $\ket{11}\otimes\tilde K\tilde DP_1\ket\Psi=\ket{11}\otimes\tilde KP_1\tilde K\ket\Psi$, since $\tilde D$ commutes with $P_1$ and $\tilde D\ket\Psi=\tilde K\ket\Psi$, and $\tilde KP_1\tilde K\ket\Psi=\tfrac12(\tilde K^2-\tilde KA_1\tilde K)\ket\Psi=\tfrac12(\id+A_1)\ket\Psi=P_0\ket\Psi$, by $\tilde KA_1\tilde K\ket\Psi=-\tilde K^2A_1\ket\Psi$. Hence
\begin{equation}
    \Phi\ket\Psi=\bigl(\ket{00}+\ket{11}\bigr)\otimes P_0\ket\Psi=\ket{\Phi^+}\otimes\ket\xi ,
\end{equation}
with $\ket\xi=\sqrt2P_0\ket\Psi$, which is normalised because $\Phi$ is an isometry. Since $P_aA_1=(-1)^aP_a$, the same computation gives $\Phi A_1\ket\Psi=(Z\otimes\id)\ket{\Phi^+}\otimes\ket\xi$. The key observable $A_1$ thus acts as $Z$ on the maximally entangled pair, its outcome is uniformly random, and the state of $E$ does not depend on it.

\section{Coefficients of the entropy relaxation}
\label{app:coefficients}

For completeness we record the coefficients $\alpha_k,\beta_k$ of
Eq.~\eqref{eq:entropy_bound}, following
Ref.~\cite{koßmann2025reliableentropyestimationobserved}. They depend only on
the chosen partition and can be precomputed once.

For a classical--quantum state, Frenkel's integral representation gives~\cite{Frenkel2023,koßmann2025reliableentropyestimationobserved}
\begin{equation}
    \label{eq:app_frenkel}
    D(\rho_{AE\Theta}\|\id_A\otimes\rho_{E\Theta})=\frac{1}{\ln2}\Bigl[\Tr(\rho_{AE\Theta}-\id_A\otimes\rho_{E\Theta})+\int_0^1\frac{du}{u}\,f(u)\Bigr],\qquad f(u)=\sup_{0\leq P\leq\id}\Tr\bigl[P\,(u\,\id_A\otimes\rho_{E\Theta}-\rho_{AE\Theta})\bigr],
\end{equation}
where $f$ is convex, $f(0)=0$, and $f(u)=0$ for $u\leq\mu$ whenever $\mu\,\id_A\otimes\rho_{E\Theta}\leq\rho_{AE\Theta}$. Fix a partition $0<t_1<\dots<t_r$ of the integration interval with $t_r=1$. On each interval $[t_k,t_{k+1}]$, $f$ lies below its chord, which gives $\int_{t_1}^1 f(u)\,du/u\leq\sum_{k=1}^r c_kf(t_k)$ with
\begin{equation}
    \label{eq:app_ck}
    c_k = \begin{cases}
      \Bigl(1+\dfrac{t_1}{t_2-t_1}\Bigr)\ln\dfrac{t_2}{t_1} - 1,
        & k=1, \\[10pt]
      1 - \dfrac{t_{r-1}}{t_r-t_{r-1}}\ln\dfrac{t_r}{t_{r-1}},
        & k=r, \\[10pt]
      \Bigl(1+\dfrac{t_k}{t_{k+1}-t_k}\Bigr)\ln\dfrac{t_{k+1}}{t_k}
        - \dfrac{t_{k-1}}{t_k-t_{k-1}}\ln\dfrac{t_k}{t_{k-1}},
        & \text{otherwise.}
    \end{cases}
\end{equation}
On $[0,t_1]$, convexity and $f(0)=0$ give $f(u)\leq(u/t_1)f(t_1)$, hence $\int_0^{t_1}f(u)\,du/u\leq f(t_1)$. Writing $f(t_k)=-\inf_P\Tr[P(\rho_{AE\Theta}-t_k\id_A\otimes\rho_{E\Theta})]$ yields Eq.~\eqref{eq:entropy_bound} with
\begin{equation}
    \label{eq:app_alphabeta}
    \alpha_k = -c_k , \qquad \beta_k = c_k\, t_k \quad(1\leq k\leq r), \qquad \alpha_0=-1,\qquad \beta_0=t_1 .
\end{equation}
The $k=0$ term therefore involves the same operator as the $k=1$ term, and the two can be merged by replacing $c_1$ with $c_1+1$. It is generally non-zero: for instance, if Eve knows the key perfectly, $\rho_{AE\Theta}$ has support orthogonal to part of that of $\id_A\otimes\rho_{E\Theta}$ and $f(t_1)=t_1$. Only if $t_1$ is itself a valid lower limit, $t_1\,\id_A\otimes\rho_{E\Theta}\leq\rho_{AE\Theta}$, does the term vanish. Such a limit cannot be guaranteed for an adversarial state, and setting $\beta_0=0$ therefore overestimates the entropy by up to $t_1/\ln2$. The bound of
Eq.~\eqref{eq:entropy_bound} becomes tight as the partition is refined. For the ideal behaviour, where $\rho_{AE\Theta}=\tfrac12\id_A\otimes\rho_{E\Theta}$, $f$ vanishes on $[0,1/2]$ and is linear on $[1/2,1]$. Since $t=1/2$ is a node, the discretisation, including the $k=0$ term, is then exact.

\section{The entropy objective for each protocol}
\label{app:hae}

This appendix derives Eq.~\eqref{eq:objective} from
Eq.~\eqref{eq:entropy_bound} and specialises it to the protocols of
Sec.~\ref{sec:protocols}.

\subsection{Reduction of the objective}

With the basis register $\Theta$ written out, the state~\eqref{eq:rho_AE} and its marginal read
\begin{align}
    \label{eq:app_cq}
    \rho_{AE\Theta} &= \frac{1}{2\nu}\sum_{\theta=1}^\nu \sum_{a=0}^1
        \ket a\bra a_A \otimes
        \Tr_B(\rho_{(\theta,a)})\otimes\ket{\theta}\bra{\theta}_{\Theta} , \\
    \label{eq:app_rhoe}
    \id_A\otimes\rho_{E\Theta} &= \frac{1}{2\nu}\,\id_A \otimes
        \sum_{\theta=1}^\nu\sum_{a=0}^1
        \Tr_B(\rho_{(\theta,a)})\otimes\ket{\theta}\bra{\theta}_{\Theta} .
\end{align}
Since $|A|=2$, the first term of
Eq.~\eqref{eq:entropy_bound} is
$\Tr(\id_A\otimes\rho_{E\Theta}-\rho_{AE\Theta})=|A|-1=1$ for every protocol.

The operator $-\alpha_k\rho_{AE\Theta}-\beta_k\,\id_A\otimes\rho_{E\Theta}$ is block diagonal
with respect to both the key register $A$ and the basis register $\Theta$. The infimum of $\Tr(P_kX)$ over projectors is attained by the projector onto the negative part of $X$, which inherits the same block structure, so we may write, without loss of generality,
\begin{equation}
    \label{eq:app_pk}
    P_k = \sum_{\theta=1}^\nu\sum_{a=0}^1 \ket a\bra a_A \otimes
          P_{a,k}^\theta\otimes\ket{\theta}\bra{\theta}_{\Theta} ,
\end{equation}
with $P_{a,k}^\theta$ projectors on Eve's quantum memory. For different $a$ or $\theta$ they need not commute. They carry the index $\theta$
because the basis register belongs to Eve.
The two traces appearing in Eq.~\eqref{eq:entropy_bound} become
\begin{align}
    \label{eq:app_trace1}
    \Tr(P_k\rho_{AE\Theta})
    &= \frac{1}{2\nu}\sum_{\theta=1}^\nu
       \Tr\bigl[P_{0,k}^\theta\rho_{(\theta,0)} + P_{1,k}^\theta\rho_{(\theta,1)}\bigr] ,\\
    \label{eq:app_trace2}
    \Tr\bigl(P_k(\id_A\otimes\rho_{E\Theta})\bigr)
    &= \frac{1}{2\nu}\sum_{\theta=1}^\nu
       \Tr\bigl[(P_{0,k}^\theta+P_{1,k}^\theta)(\rho_{(\theta,0)}+\rho_{(\theta,1)})\bigr] ,
\end{align}
where $\Tr[P\rho_x]$ stands for $\Tr[(\id_B\otimes P)\rho_x]$, and in Eq.~\eqref{eq:app_trace2} the $A$ register of $P_k$ is traced against $\id_A$, so both key values contribute to each projector.
Substituting into Eq.~\eqref{eq:entropy_bound} gives Eq.~\eqref{eq:objective}. For $\nu=1$ the basis register is trivial.

The states $\{\rho_x\}$ and Bob's measurements $\{M_{b|y}\}$ are
shared across all bases -- Eve's channel acts before the basis is announced -- while only her variational projectors carry the basis index. For each $k$
the sum over $\theta$ is therefore optimised as a single semidefinite program.

\begin{remark}\label{rem:basisindex}
Dropping the basis index of the projectors, $P^\theta_{a,k}\to P_{a,k}$, amounts to assuming that Eve never learns the basis and does not give a valid bound. For the protocols with several key bases it leads to overestimated rates; the protocols with a single key basis are unaffected.
\end{remark}

\section{Numerical implementation}
\label{app:implementation}

The numerical results of Sec.~\ref{sec:results} are obtained with the following choices; Appendix~\ref{app:checks} validates them.

\emph{Relaxation.} We use the reduced relaxation~\eqref{eq:swnpa} of Proposition~\ref{lem:onefactor}. The moments are taken real, which is no loss of generality since the problem is invariant under complex conjugation, and $M_{1|y}=\id-M_{0|y}$ is eliminated. In the program for node $k$, $W$ consists of words in Bob's projectors $M_{0|y}$ and Eve's projectors $P^\theta_{a,k}$:
\begin{itemize}
    \item for $\prot43$ and $\prot33$, all words of length at most three. For $\prot43$, words of length four change the entropy by at most $3\cdot10^{-4}$ bits at the honest $G$ and leave the thresholds unchanged to within the bisection precision of $5\cdot10^{-4}$;
    \item for $\prot44$, all words of length at most two, together with the products $M_{0|y}M_{0|y'}P^\theta_{a,k}$. These products raise the entropy by $0.034$ bits at $v=0.86$; the further products $M_{0|y}M_{0|y'}M_{0|y''}$ move the thresholds by less than $10^{-3}$;
    \item for $\prot63$, all words of length at most two except the products of Eve's projectors belonging to different bases, together with the products $M_{0|y}M_{0|y'}P^\theta_{a,k}$. At the honest $G$ the latter lower $\eta_c$ by $0.019$ and leave $v_c$ unchanged; the cross-basis products change the thresholds by at most $3\cdot10^{-3}$.
\end{itemize}
For $\prot43$ the program then has $1595$ moments and is solved in about one second; for $\prot44$ it has $9005$ moments and takes $23$\,s.

\emph{Entropy partition.} We use the partition of Appendix~\ref{app:coefficients} with the $26$ nodes $0.005$, $0.01$, $0.02$, $0.035$, $0.05$, $0.075$, $0.1$, $0.125$, $0.15$ and $0.2,0.25,\dots,1$. The $k=0$ term is merged into the $k=1$ term, so each data point requires $26$ semidefinite programs. The partition contains $t=1/2$, and is therefore exact for the ideal behaviour. At the points we tested, it is within $10^{-3}$ bits of a partition with $120$ nodes, which moves the thresholds by less than $3\cdot10^{-4}$.

\emph{Decoupling over $k$.} The objective in Eq.~\eqref{eq:objective} is a sum
over $k$ of the terms $\sum_\theta\Tr(\Omega_k^\theta)$, each involving its own
projectors $P_{a,k}^\theta$ but all sharing the states and Bob's measurements.
Rather than optimise jointly, we minimise each term separately and add the
results. Because
$\inf\sum_k f_k \geq \sum_k \inf f_k$, this is a further relaxation and the
sum remains a valid lower bound; its benefit is that the number of projector
variables per semidefinite program drops from $2\nu(r+1)$ to $2\nu$, and the
$r+1$ programs are solved independently~\cite{brown2024device}.

\emph{Numerical output.} The programs are generated with the Julia package PCPOP.jl~\cite{BermejoMoran2026PCPOP} and solved in dual form with MOSEK through JuMP, with the default tolerances. We use the dual objective only when MOSEK returns a feasible point, in which case it is a lower bound up to the feasibility tolerance. Every moment of Eq.~\eqref{eq:swnpa} is bounded, $|L(\rho_xw)|\leq1$ and $|L(\sigma w)|\leq G$. A residual $r$ in the dual equality constraints therefore changes the bound by at most $G\,\|r\|_1$, which converts the returned solution into a certified bound. We evaluated this correction, after projecting the dual matrices onto the positive semidefinite cone, at three representative points and at one control point. At the representative points it is at most $2\cdot10^{-5}$ bits, which moves the thresholds by less than $10^{-5}$, and we do not include it in the reported values. The control is the ideal point of the three-state protocol of Ref.~\cite{Tamaki2014}, which admits an exact classical attack (Sec.~\ref{sec:nokey}). There the relaxation gives $H(A|E\Theta)\approx10^{-6}$ bits before the correction, and a negative value after it. The key rate is evaluated in steps of $0.01$ in $v$ and in $\eta$, starting from the ideal point or, for the most expensive relaxations, from just above the threshold, and each threshold, the zero crossing of the resulting lower bound, is then located by bisection to $5\cdot10^{-4}$. The curves of $\prot44$ and $\prot63$ in Fig.~\ref{fig:keyrates} use steps of $0.02$.

\section{Numerical checks and further results}
\label{app:checks}

\emph{Validation of the relaxation.} 
The main numerical implementation in Julia based on the package PCPOP.jl~\cite{BermejoMoran2026PCPOP} has been checked against an independent implementation in Python (CVXPY with MOSEK), fully carried out by Claude Opus~5.5. It reproduces the entropies of the main code to within $10^{-5}$ at lengths two and three. 

\emph{Effect of the relaxation and comparison with the DI twin.} Table~\ref{tab:steps} follows the thresholds of $\prot43$ through the successive choices of Appendix~\ref{app:implementation}. Refining the partition matters little; the longer words matter most, and at length three the LB thresholds coincide with the device-independent ones.

\begin{table}[h]
    \centering
    \begin{tabular}{lcc}
        \toprule
        Relaxation & $v_c$ & $\eta_c$ \\
        \midrule
        $W=\{\id,M,P,MP\}$ (tracial level $1+B\!\cdot\!E$), $t_k=k/10$ & $0.869$ & $0.802$ \\
        same words, $26$-node partition & $0.868$ & $0.801$ \\
        all words of length $\leq2$ & $0.854$ & $0.795$ \\
        all words of length $\leq3$ (Table~\ref{tab:thresholds}) & $0.852$ & $0.791$ \\
        \midrule
        DI twin, NPA level $2+A\!\cdot\!B\!\cdot\!E$ & $0.852$ & $0.791$ \\
        \bottomrule
    \end{tabular}
    \caption{Thresholds of $\prot43$ at the honest $G$ for successively stronger relaxations, and for the device-independent bound of the entanglement-based twin. All entries are bisected zero crossings.}
    \label{tab:steps}
\end{table}

\emph{Device-independent twins.} For the twin of a protocol whose preparations are the eigenstates of observables at Bloch angles $\varphi_\theta$ in the $XZ$ plane, Alice measures the same observables on her half of the maximally entangled state $\ket{\Phi^+}$, depolarised with visibility $v$, and Bob measures as in the P\&M protocol, with detection efficiency $\eta$. We impose the full statistics, $\langle A_\theta\rangle=0$, $\langle B_y\rangle$ and $\langle A_\theta B_y\rangle=v\cos(\varphi_\theta-\varphi_y)$ with the no-click outcome binned into $b=0$, on an NPA relaxation in a normalised state. Alice's, Bob's and Eve's projectors commute across the parties, and Eve has one pair of projectors per key basis. The words are all those of length at most two together with the products $A_\theta B_yP^\theta_{a,k}$. The entropy bound, its partition and the decoupling over nodes are those of Appendix~\ref{app:implementation}, and $H(A|B)$ is the same as in the P\&M protocol.

\emph{Other protocols.} Table~\ref{tab:others} lists the thresholds of all the protocols we studied at the honest $G$, together with those of the min-entropy bound $H(A|E\Theta)\geq-\log_2\Pg(A|E\Theta)$, where Eve's guessing probability of the key is maximised over the same relaxation.

\begin{table}[h]
    \centering
    \begin{tabular}{llcccc}
        \toprule
        & & \multicolumn{2}{c}{von Neumann} & \multicolumn{2}{c}{min-entropy} \\
        \cmidrule(lr){3-4}\cmidrule(lr){5-6}
        Protocol & Words $W$ & $v_c$ & $\eta_c$ & $v_c$ & $\eta_c$ \\
        \midrule
        $\prot43$ & length $\leq3$ & $0.852$ & $0.791$ & $0.890$ & $0.828$ \\
        $\prot44$ & length $\leq2$, $+MM'P$ & $0.838$ & $0.779$ & $0.883^\dagger$ & $0.820^\dagger$ \\
        $\prot63$ & length $\leq2^\ddagger$, $+MM'P$ & $0.870$ & $0.764$ & $0.889$ & $0.824^\dagger$ \\
        $\prot33$ & length $\leq3$ & $0.891$ & $0.821$ & $0.920$ & $0.848$ \\
        \midrule
        $\prot63$, bases of $Z$, $X$, $(Z+X)/\sqrt2$ & length $\leq2^\ddagger$, $+MM'P$ & $0.892$ & $0.795$ & $0.909$ & $0.849^\dagger$ \\
        BB84 states, measurements $Z$, $X$, $(Z+X)/\sqrt2$ & length $\leq2$, $+MM'P$ & $0.919$ & $0.867$ & $0.942$ & $0.889$ \\
        $\ket0,\ket1,\ket+$, measurements $Z$, $(Z+X)/\sqrt2$ & length $\leq3$ & $0.939$ & $0.906$ & $0.959$ & $0.919$ \\
        $\ket0,\ket1,\ket+$, measurements $Z$, $X$, $(Z+X)/\sqrt2$ & length $\leq3$ & $0.938$ & $0.906$ & $0.958$ & $0.919$ \\
        $\ket0,\ket1,\ket+$, measurements $Z$, $X$~\cite{Tamaki2014} & --- & \multicolumn{4}{c}{no key (Sec.~\ref{sec:nokey})} \\
        \bottomrule
    \end{tabular}
    \caption{Thresholds at the honest $G$ for all the protocols we studied. The last five rows are variants discussed in Sec.~\ref{sec:results}. In the second one, key is generated in both BB84 bases. $^\dagger$Computed with the words of length at most two, without the products $MM'P$. $^\ddagger$Without the products of Eve's projectors belonging to different bases. For the two three-state variants with a single tilted measurement, $\eta_c=0.906$ is the efficiency at which the binned behaviour becomes classical (Sec.~\ref{sec:noise}). For the last protocol the relaxation returns $H(A|E\Theta)\approx10^{-6}$ at the ideal point, zero within numerical accuracy, in agreement with the exact classical attack of Sec.~\ref{sec:nokey}.}
    \label{tab:others}
\end{table}

\emph{Tilted measurements.} For $\prot43$ with $y=2,3\to\cos\gamma\,Z\pm\sin\gamma\,X$, the relaxation with all words of length at most three gives $v_c=0.874$, $0.859$, $0.852$, $0.854$ and $0.865$ for $\gamma=30^\circ$, $37.5^\circ$, $45^\circ$, $52.5^\circ$ and $60^\circ$, so the bisectors are optimal among the angles we tested. More generally, optimising the directions of the test measurements for the min-entropy rate near the thresholds changes the thresholds of $\prot43$, $\prot44$ and $\prot33$ by at most $2\cdot10^{-3}$, adding one or two test measurements to $\prot63$ lowers its thresholds by about $0.01$, and replacing the preparations by the six eigenstates of the Pauli operators brings no improvement.

\emph{Noise in the source.} With the bound set to the honest value for the depolarised source, $G=1+v$, the visibility thresholds are $0.727$ for $\prot43$, $0.715$ for $\prot44$ (words of length at most two) and $0.766$ for the asymmetric $\prot63$ (without cross-basis products).

\emph{Relaxed bound.} The thresholds of Fig.~\ref{fig:dG} use the relaxations of Table~\ref{tab:thresholds}. The largest $\delta G$ at which the ideal behaviour certifies key was located by raising $\delta G$ in steps of $0.01\,n$ until the certified rate fell below $10^{-4}$ bits, followed by bisection to $5\cdot10^{-4}\,n$. The tolerance is needed because, where the exact value of the relaxation is zero, the solver returns about $10^{-7}$ bits. The von Neumann and min-entropy bounds give the same values to this precision.

\emph{Entropy versus the CHSH value.} The bounds on $F(S,G)$ of Sec.~\ref{sec:chsh} use the relaxation~\eqref{eq:swnpa} for four preparations and Bob's measurements $y=2,3$ of $\prot43$. In place of the behaviour, the only linear constraint is $S_{\mathrm{CHSH}}\geq S$, and the key is taken from the $Z$ basis, $\theta=1$. $W$ consists of all words of length at most four in Bob's and Eve's projectors, which gives $41$ words and $505$ moments. We computed $F$ for $G=1.2$, $1.4$, $1.6$, $1.8$, $1.92$, $2$, $2.08$, $2.2$, $2.4$, $2.6$, $2.8$ and $3$, in steps of $0.02$ in $S$ up to the largest feasible value, and along the depolarised source, $S=2\sqrt2\,v$ with $G=1+v$.
\begin{itemize}
    \item Words of length five raise $F$ by at most $1.4\cdot10^{-4}$ bits at $G=2$. At $G=1.6$, $1.8$, $2.08$, $2.2$ and $2.4$ they raise it by at most $1.4\cdot10^{-3}$ for $S\leq2\sqrt2$, and they move the thresholds of Sec.~\ref{sec:chsh} by at most $2.5\cdot10^{-4}$.
    \item The $120$-node partition raises $F$ by at most $1.2\cdot10^{-3}$ bits. At $G=2$ it brings the curve to within $5\cdot10^{-5}$ bits of Eq.~\eqref{eq:di_chsh} for $S\geq2.1$.
\end{itemize}
For $G<2$, each curve has two kinks. The first is where the plateau at $2-G$ ends. The second is where Eve's guessing probability of the key bit falls below its largest possible value $G/2$. That bound holds because a guess of the key bit, together with the guess that the basis is $Z$, is a guess of $x$. The bound $F$ is non-decreasing in $S$ and non-increasing in $G$, so the tabulated values bound it rigorously between grid points.

\end{document}